%% file: fast-vertex-sparsest.tex
\documentclass[11pt,a4paper]{article}
\usepackage{algorithm2e}
\usepackage{amsfonts}
\usepackage{amssymb}
\usepackage{amstext}
\usepackage{amsmath}
\usepackage{xspace}
\usepackage{ntheorem}
\usepackage{graphicx}
\usepackage{fullpage}
\usepackage{graphics}
\usepackage{colordvi}
\usepackage{colordvi}
\usepackage{url}
\usepackage{soul}
\usepackage{xcolor}

\newcommand{\myparskip}{1pt}
\parskip \myparskip

\newcommand{\dist}{\operatorname{dist}}

\newcommand{\corepath}{\mbox{\sf{core-path}}}
\newcommand{\pquery}{\mbox{\sf{path-query}}\xspace}
\newcommand{\pathquery}{\mbox{\sf{path-query}}\xspace}
\newcommand{\distquery}{\mbox{\sf{dist-query}}\xspace}

\newcommand{\CONN}{\mbox{\sf{CONN}}}
\newcommand{\CONNSF}{\mbox{\sf{CONN-SF}}}
\renewcommand{\path}{\mbox{\sf{path}}}
\newcommand{\EST}{\mbox{\sf{ES-Tree}}\xspace}
\newcommand{\ESTs}{\mbox{\sf{ES-Trees}}\xspace}
\newcommand{\DSP}{\mbox{\sf{Proc-Degree-Pruning}}\xspace}
\newcommand{\PDS}{\mbox{\sf{Proc-Degree-Pruning}}\xspace}
\newcommand{\OK}{\overline{K}}

\newcommand{\deledge}{\mbox{\sf{Proc-Delete-Edge}}\xspace}

\newcommand{\PSP}{\mbox{\sf{Proc-Path-Simplify}}\xspace}
\newcommand{\PCL}{\mbox{\sf{Proc-Construct-Layers}}\xspace}
\newcommand{\PDV}{\mbox{\sf{Proc-Delete-Vertex}}\xspace}
\newcommand{\Lamda}{\Lambda}

\newcommand{\EEST}{\mbox{\sf{Extended ES-Tree}}\xspace}
\newcommand{\EESF}{\mbox{\sf{Extended ES-Forest}}\xspace}

\newcommand{\WSES}{\mbox{\sf{WSES}}\xspace}

\newcommand{\WSESIN}{\mbox{\sf{WSES-Insert}}\xspace}
\newcommand{\twin}{\mbox{\sf{WSES-Twin}}\xspace}

\newcommand{\WSESDEL}{\mbox{\sf{WSES-Delete}}\xspace}
\newcommand{\WSESVIN}{\mbox{\sf{WSES-Vertex-Insert}}\xspace}
\newcommand{\WSESSPLIT}{\mbox{\sf{WSES-Cluster-Split}}\xspace}

\newcommand{\EESTIN}{\mbox{\sf{EEST-Insert}}\xspace}
\newcommand{\EESTDEL}{\mbox{\sf{EEST-Delete}}\xspace}
\newcommand{\EESTSPLIT}{\mbox{\sf{EEST-Split}}\xspace}

\newcommand{\SSSP}{\mbox{\sf{SSSP}}\xspace}

\newcommand{\CONNIN}{\mbox{\sf{CONN-Insert}}}
\newcommand{\CONNDEL}{\mbox{\sf{CONN-Delete}}}
\newcommand{\conn}{\mbox{\sf{conn}}}

\newcommand{\cKRV}{\ensuremath{c_{\mbox{\textup{\tiny{KRV}}}}}\xspace}
\newcommand{\optp}{\mbox{\sf OPT}_{\mbox{\textup{\small{Primal}}}}}
\newcommand{\optd}{\mbox{\sf OPT}_{\mbox{\textup{\small{Dual}}}}}
\newcommand{\costd}{C_{\mbox{\textup{\small{Dual}}}}}

\newcommand{\out}{{\sf out}}

\newcommand{\hG}{\hat G}
\newcommand{\tG}{\tilde G}

\newcommand{\hn}{\check n}

\newcommand{\graph}{\check G}

\newcommand{\ceil}[1]{\ensuremath{\left\lceil#1\right\rceil}}
\newcommand{\floor}[1]{\ensuremath{\left\lfloor#1\right\rfloor}}

\newcommand{\round}{\mathsf{ROUND}}

\def\etal{et al.\xspace}
\newcommand{\mX}{W}

\newcommand{\opt}{\mbox{\sf OPT}}

\newcommand{\set}[1]{\left\{ #1 \right\}}

\newcommand{\mynote}[1]{\textcolor{red}{\sc\bf{[#1]}}}

\newcommand{\iset}{{\mathcal{I}}}
\newcommand{\pset}{{\mathcal{P}}}
\newcommand{\qset}{{\mathcal{Q}}}

\newcommand{\aset}{{\mathcal{A}}}
\newcommand{\cset}{{\mathcal{C}}}
\newcommand{\fset}{{\mathcal{F}}}
\newcommand{\dset}{{\mathcal{D}}}

\newcommand{\kset}{{\mathcal K}}

\newcommand{\rset}{{\mathcal{R}}}

\newcommand{\hset}{{\mathcal{H}}}
\newcommand{\sset}{{\mathcal{S}}}

\newcommand{\be}{\begin{enumerate}}
\newcommand{\ee}{\end{enumerate}}
\newcommand{\bd}{\begin{description}}
\newcommand{\ed}{\end{description}}
\newcommand{\bi}{\begin{itemize}}
\newcommand{\ei}{\end{itemize}}

\newtheorem{theorem}{Theorem}[section]
\newtheorem{lemma}[theorem]{Lemma}
\newtheorem{observation}[theorem]{Observation}
\newtheorem{corollary}[theorem]{Corollary}
\newtheorem{claim}[theorem]{Claim}

\newtheorem*{definition}{Definition.}
\newenvironment{proof}{\par \smallskip{\bf Proof:}}{\hfill\stopproof}
\newenvironment{proofof}[1]{\par \smallskip{\bf Proof of #1.}}%
       {\hfill\stopproof}

\def\stopproof{\square}
\def\square{\vbox{\hrule height.2pt\hbox{\vrule width.2pt height5pt \kern5pt
\vrule width.2pt} \hrule height.2pt}}

\newenvironment{prog}[1]{
\begin{minipage}{5.8 in}
\begin{center}
{\sc #1}
\end{center}
}
{
\end{minipage}}

\newcommand{\program}[2]{\fbox{\vspace{2mm}\begin{prog}{#1} #2 \end{prog}\vspace{2mm}}}

\renewcommand{\phi}{\varphi}
\newcommand{\eps}{\epsilon}

\newcommand{\half}{\ensuremath{\frac{1}{2}}}

\newcommand{\poly}{\operatorname{poly}}

\newenvironment{properties}[2][0]
{
\begin{enumerate} \setcounter{enumi}{#1}}{\end{enumerate}}

\begin{document}

\title{A New Algorithm for Decremental Single-Source Shortest Paths with Applications to Vertex-Capacitated Flow and Cut Problems\footnote{An extended abstract is to appear in STOC 2019}}


\author{Julia Chuzhoy\thanks{Toyota Technological Institute at Chicago. Email: {\tt cjulia@ttic.edu}. Part of the work was done while the author was a Weston visiting professor in the Department of Computer Science and Applied Mathematics, Weizmann Institute. Supported in part by NSF grant CCF-1616584.}  \and Sanjeev Khanna\thanks{University of Pennsylvania, Philadelphia, PA 19104. Email: {\tt sanjeev@cis.upenn.edu}. Supported in part by the National Science Foundation grant CCF-1617851.}
}

\begin{titlepage}
\maketitle

\thispagestyle{empty}

\begin{abstract}
We study the vertex-decremental Single-Source Shortest Paths (\SSSP) problem: given an undirected graph $G=(V,E)$ with lengths $\ell(e)\geq 1$ on its edges and a source vertex $s$, we need to support (approximate) shortest-path queries in $G$, as $G$ undergoes vertex deletions. In a shortest-path query, given a vertex $v$, we need to return a path connecting $s$ to $v$, whose length is at most $(1+\eps)$ times the length of the shortest such path, where $\eps$ is a given accuracy parameter. The problem has many applications, for example to flow and cut problems in vertex-capacitated graphs.

Decremental \SSSP is a fundamental problem in dynamic algorithms that has been studied extensively, especially in the more standard edge-decremental setting, where the input graph $G$ undergoes edge deletions. The classical algorithm of Even and Shiloach supports exact shortest-path queries in $O(mn)$ total update time. A series of recent results have improved this bound for approximate \SSSP to $O(m^{1+o(1)}\log L)$, where $L$ is the largest length of any edge. However, these improved results are randomized algorithms that assume an \emph{oblivious} adversary. To go beyond the oblivious adversary restriction, 
recently, Bernstein, and Bernstein and Chechik designed deterministic algorithms for the problem, with total update time $\tilde O(n^{2}\log L)$, that by definition work against an adaptive adversary. 
Unfortunately, these deterministic algorithms introduce a new limitation, namely, they can only return the approximate length of a shortest path, and not the path itself. Many applications of the decremental \SSSP problem, including the ones considered in this paper, crucially require both that the algorithm returns the approximate shortest paths themselves and not just their lengths, and that it works against an adaptive adversary.

Our main result is a randomized algorithm for vertex-decremental \SSSP with total expected update time $O(n^{2+o(1)}\log L)$, that responds to each shortest-path query in $O(n\log L)$ time in expectation, returning a $(1+\eps)$-approximate shortest path. The algorithm works against an adaptive adversary. The main technical ingredient of our algorithm is an $\tilde O(|E(G)|+ n^{1+o(1)})$-time algorithm to compute a \emph{core decomposition} of a given dense graph $G$, which allows us to compute short paths between pairs of query vertices in $G$ efficiently. We believe that this core decomposition algorithm may be of independent interest.

 We use our result for vertex-decremental \SSSP to obtain $(1+\eps)$-approximation algorithms for maximum $s$-$t$ flow and minimum $s$-$t$ cut in vertex-capacitated graphs, in expected time $n^{2+o(1)}$, and an $O(\log^4n)$-approximation algorithm for the vertex version of the sparsest cut problem with expected running time $n^{2+o(1)}$. These results improve upon the previous best known results for these problems in the regime where $m= \omega(n^{1.5 + o(1)})$.
\end{abstract}

\end{titlepage}

\input{intro}

\input{overview}

\input{prelims}

\input{new-proof}

\input{max-flow}

\input{sparsest-cut}

\section*{Acknowledgements}

We are grateful to Chandra Chekuri for pointing us to some closely related prior work, and to Shiri Chechik for sharing with us an early version of her paper~\cite{chechik}.

\appendix

\input{appendix}

\bibliographystyle{alpha}
\bibliography{fast-vertex-sparsest}

\end{document}

%% file: intro.tex
\section{Introduction}

In this paper we consider the {\em vertex}-decremental Single-Source Shortest
Paths (SSSP) problem in edge-weighted undirected graphs, and its applications to several cut and flow problems in
vertex-capacitated graphs. In the vertex-decremental SSSP, we are given an undirected graph 
$G$ with lengths $\ell(e)\geq 1$ on its edges, and a source vertex $s$. The goal is to support (approximate) shortest-path queries from the source vertex $s$, as the graph $G$ undergoes a sequence of online adversarial vertex deletions. 
We consider two types of queries: in a \pathquery, we are given a query vertex $v$, and the goal is to return a path connecting $s$ to $v$, whose length is at most $(1+\eps)$ times the length of the shortest such path, where $\eps$ is the given accuracy parameter. In a \distquery, given a vertex $v$, we need to report an (approximate) distance from $s$ to $v$. 
 We will use the term \emph{exact} \pathquery when the algorithm needs to report the shortest $s$-$v$ path, and \emph{approximate} \pathquery when a $(1+\eps)$-approximate shortest $s$-$t$ path is sufficient. We will similarly use the terms of exact and approximate \distquery. We also distinguish between an \emph{oblivious adversary} setting, where the sequence of vertex deletions is fixed in advance, and \emph{adaptive adversary}, where each vertex in the deletion sequence may depend on the responses of the algorithm to previous queries.

A closely related variation of this problem, that has been studied extensively, is the {\em edge}-decremental SSSP problem, where the graph $G$ undergoes edge deletions and not vertex deletions.
The edge-decremental SSSP captures the vertex-decremental version as a special case, and has a long history with 
many significant developments just in the past few years. We start by briefly reviewing the work on edge-decremental SSSP, focusing primarily on undirected graphs. 
The two parameters of interest are the {\em total update time}, defined as the total time spent by the algorithm on maintaining its data structures over the entire sequence of deletions, and \emph{query time}, defined as the time needed to respond to a single \pathquery or \distquery.  A classic result of Even and Shiloach~\cite{EvenS,Dinitz,HenzingerKing} gives an algorithm that supports exact \pathquery and \distquery with only $O(mn)$ total
update time over all edge deletions, with $O(1)$ query time for \distquery and $O(n)$ query time for \pathquery. 
While the $O(mn)$ update time represents a significant improvement over the naive algorithm that simply recomputes a shortest path tree after each edge deletion, it is far from the near-linear total update time results that are known for many other decremental problems in undirected graphs.
It remained an open problem for nearly 3 decades to improve upon the update time of the algorithm. Roditty and Zwick~\cite{RodittyZ11} highlighted a fundamental obstacle to getting past the $O(mn)$ time barrier using combinatorial approaches, even for unweighted undirected graphs, by showing that the long-standing problem of designing fast combinatorial algorithms for Boolean matrix multiplication can be reduced to the exact edge-decremental SSSP. Furthermore, in a subsequent work, Henzinger {\em et al.}~\cite{HenzingerKNS15} showed that, assuming the online Boolean matrix-vector multiplication conjecture, the $O(mn)$ time barrier for exact edge-decremental SSSP holds even for arbitrary algorithms for the problem. The obstacles identified by these conditional results, however, only apply to supporting {\bf exact} \distquery. Essentially all subsequent work on edge-decremental SSSP has thus focused on the task of supporting approximate \pathquery and \distquery. 
In the informal discussion below we implicitly assume that the accuracy parameter $\eps$ is a constant and thus ignore the dependence on this parameter in the time bounds (but we do not make this assumption in our algorithms, and the formal statements of our results give explicit dependence on $\eps$).

 Bernstein and Roditty~\cite{BernsteinR11} made the first major progress in breaking the $O(mn)$ update time barrier, by showing an algorithm that supports approximate \distquery in undirected unweighted graphs with $n^{2+o(1)}$ total update time, and $O(1)$ query time. Subsequently, Henzinger, Krinninger, and Nanogkai~\cite{HenzingerKN14_soda} improved this update time to $O(n^{1.8 + o(1)} + m^{1+o(1)})$, and shortly afterwards, the same authors~\cite{HenzingerKN14_focs} extended it to arbitrary edge lengths and improved it further to an essentially optimal total update time of $O(m^{1+o(1)} \log L)$ where $L$ is the largest length of an edge. All three algorithms are randomized, and moreover, they assume that the edge deletion sequence is given by an oblivious adversary. 
 For many applications, including the ones considered in this paper, it is crucial that the algorithm can handle an {\em adaptive} adversary, and support \pathquery. For instance, fast approximation schemes for computing a maximum multicommodity flow in a graph (see, for instance,~\cite{GK98, Fleischer00}) rely on a subroutine that can identify an approximate shortest $s$-$t$ path under suitably chosen edge lengths, and pushing flow along such a path. The edge lengths are then updated for precisely the edges lying on the path (as we will show later, such updates can be modeled by the deletion of edges or vertices on the path). Thus, the edges that are deleted  at any step strongly depend on the responses to the approximate path queries from previous steps. Moreover, these applications require that we obtain the actual approximate shortest paths themselves, and not just approximate shortest distances.

The goal of eliminating the oblivious adversary restriction initiated a search for deterministic edge-decremental SSSP algorithms, which, by definition, can handle adaptive deletion sequences. Bernstein and Chechik~\cite{BernsteinChechik} gave the first deterministic algorithm to break the $O(mn)$ total update time barrier. Their algorithm achieves a total update time of $\tilde{O}(n^2)$ and an $O(1)$ query time for approximate $\distquery$. In a subsequent work~\cite{BernsteinChechikSparse}, they improved this bound further for the regime of sparse graphs, obtaining a total update time of $\tilde{O}(n^{5/4}\sqrt{m}) = O(m n^{3/4})$, keeping the query time of $O(1)$ for approximate \distquery. Both these results required that the underlying graph is undirected and {\em unweighted}, that is, all edge lengths are unit. In a further progress, Bernstein~\cite{Bernstein} extended these results to edge-weighted undirected graphs obtaining a total update time of $\tilde{O}(n^2 \log L)$, where $L$ is the largest edge length, while still keeping the query time of $O(1)$ for approximate \distquery. While all these results successfully eliminated the oblivious adversary restriction 
required by the previous works that achieved better than an $O(mn)$ total update time, the core approach used in these works introduced another limitation: as noted by~\cite{Bernstein}, all three results only support approximate \distquery, but not approximate \pathquery. 

 At a high level, the approach used in these results is based on partitioning the edges of the underlying graph into a {\em light} sub-graph, where the average degree is small and a {\em heavy} sub-graph, where the degree of each vertex is high, say at least $\tau$. Any shortest $s$-$v$ path can be decomposed into segments that alternately traverse through the light and the heavy graph. The shortest path segments traversing through the light graph are explicitly maintained using the approach of Even and Shiloach~\cite{EvenS,Dinitz,HenzingerKing}, exploiting the fact that the edge density is low in the light graph. The shortest path segments traversing through the heavy graph, on the other hand, are not maintained explicitly. Instead, it is observed that any shortest  $s$-$v$ path may contain at most $O(n/\tau)$ edges from the heavy graph, so they do not contribute much to the path length.
This implicit guarantee on the total length of segments traversing the heavy graph suffices for obtaining an estimate on the shortest path length by only maintaining shortest paths in the light graph. However, it leaves open the task of finding these segments themselves.

Our main technical contribution is to design an algorithm that  allows us to support approximate \pathquery against an adaptive adversary, by explicitly maintaining short paths in the heavy graph. Specifically, we design an algorithm that, given a pair of vertices $u,u'$ that belong to the same connected component $C$ of the heavy graph, returns a short path connecting $u$ to $u'$ in $C$, where the length of the path is close to the implicit bound that was used in~\cite{BernsteinChechik,Bernstein}.

Formally,  assume that we are given a {\bf simple} undirected graph $G$ with a source vertex $s$ and lengths $\ell(e)>0$ on edges $e\in E(G)$, that undergoes {\bf vertex} deletions. Throughout the algorithm,  for every pair $u,v$ of vertices, the distance $\dist(u,v)$ between them is the length of the shortest path from $u$ to $v$ in the current graph $G$, using the edge lengths $\ell(e)$. We also assume that we are given an error parameter $0<\eps<1$.
We design an algorithm that supports approximate single-source shortest-path queries, denoted by $\pquery(v)$. The query receives as input a vertex $v$, and returns a path connecting $s$ to $v$ in the current graph $G$, if such a path exists, such that the length of the path is at most $(1+\eps)\dist(s,v)$. Our main result for vertex-decremental SSSP is summarized in the following theorem.


\begin{theorem}\label{thm: main for SSSP}
There is a randomized algorithm, that, given a parameter $0<\eps<1$, and a simple undirected $n$-vertex graph $G$ with lengths $\ell(e)>0$ on edges $e\in E(G)$, together with a special source vertex $s\in V(G)$, such that $G$ undergoes vertex deletions, supports queries $\pquery(v)$. For each query $\pquery(v)$, the algorithm returns a path from $s$ to $v$ in $G$, if such a path exists, whose length is at most $(1+\eps)\dist(s,v)$. The algorithm works against an {\bf adaptive} adversary. The total expected running time of the algorithm is $O\left(\frac{n^{2+o(1)}\cdot \log^3(1/\eps)\cdot \log L}{\eps^2}\right)$, where $L$ is the ratio of  largest to smallest edge length $\ell(e)$, and each query is answered in $O(n\cdot \poly\log n\cdot \log L\cdot \log (1/\eps))$ time in expectation.
\end{theorem}


We emphasize that the algorithm is Las Vegas: that is, it always returns a path with the required properties, but its running time is only bounded in expectation. The adversary is allowed to view the complete state of the algorithm, that is, the contents of all its data structures.

One of the main technical contributions of our algorithm is a \emph{core decomposition} of dense graphs. Suppose we are given an $n$-vertex graph $G$, such that every vertex in $G$ has degree at least $h$, where $h$ is sufficiently large, say $h\geq n^{1/\log\log n}$. Informally, a \emph{core} $K$ is an expander-like sub-graph of $G$, such that every vertex of $K$ has at least $h^{1-o(1)}$ neighbors in $K$. The ``expander-like'' properties of the core ensure that, even after $h^{1-o(1)}$ vertex deletions, given any pair $u,u'$ vertices of $K$, we can efficiently find a short path connecting $u$ to $u'$ in $K$ (the length of the path depends on the balancing of various parameters in our algorithm, and is $n^{o(1)}$). A \emph{core decomposition} of $G$ consists of a collection $K_1,\ldots,K_r$ of disjoint cores in $G$, such that $r\leq n/h^{1-o(1)}$. Additionally, if we denote by $U$ the set of vertices of $G$ that do not belong to any core, then we require that it is an $h$-universal set: that is, even after $h^{1-o(1)}$ vertices are deleted from $G$, every surviving vertex of $U$ can reach one of the cores through a path of length $O(\log n)$. We show a randomized algorithm that with high probability computes a valid core decomposition in a given graph $G$ in time $\tilde O(|E(G)|+ n^{1+o(1)})$; see Section~\ref{sec: techniques} for a more detailed overview of our techniques.

While the result above leaves open the question if a similar algorithm can also be obtained for edge-decremental SSSP, for many cut and flow problems on vertex-capacitated graphs, the vertex-decremental SSSP suffices as a building block. We describe next some of these applications. We note here that the idea of using dynamic graph data structures to speed up cut and flow computations is not new. In particular, Madry~\cite{Madry10_stoc} systematically explored this idea for the maximum multicommodity flow problem and the concurrent flow problem, significantly improving the previous best known results for these problems.

Our first application shows that there is an ${O}( n^{2+o(1)})$-time algorithm for computing approximate maximum $s$-$t$ flow and minimum $s$-$t$ cut in vertex-capacitated undirected graphs. 
For approximate maximum $s$-$t$ flow problem in edge-capacitated undirected graphs, a sequence of remarkable developments incorporating ideas from continuous optimization to speed-up maximum flow computation has culminated in an $\tilde{O}(m/\eps^2)$-time algorithm for computing a $(1+\eps)$-approximate flow~\cite{ChristianoKMST11, LeeRS13, Sherman13, KelnerLOS14, Peng16}. 
We refer the reader to~\cite{Madry18_ICM} for an excellent survey of these developments.
However, no analogous results are known for maximum flow in vertex-capacitated undirected graphs. The main technique for solving the vertex-capacitated version appears to be via the standard reduction to the edge-capacitated directed case, and relying on fast algorithms for maximum $s$-$t$ flow problem in edge-capacitated directed graphs. Two recent breakthrough results for exact maximum $s$-$t$ flow
in edge-capacitated directed graphs include an $\tilde{O}(m \sqrt{n} \log^{O(1)} C)$ time algorithm by Lee and Sidford~\cite{LeeS14}, and an $\tilde{O}(m^{10/7} \log C)$ time algorithm by Madry~\cite{Madry16}; here $C$ denotes the largest integer edge capacity. The two bounds are incomparable: the former bound is preferable for dense graphs, and the latter for sparse graphs. 
For approximate maximum $s$-$t$ flow problem in edge-capacitated directed graphs, approaches based on the primal-dual framework~\cite{GK98, Fleischer00}) (or equivalently, as an application of the multiplicative weights update method~\cite{AroraHK12}) can be used to compute a $(1+\eps)$-approximate $s$-$t$ flow in $f(n,m,\eps)O(m/\eps^2)$ time where $f(n,m,\eps)$ denotes the time needed to compute a $(1+\eps)$-approximate shortest path from $s$ to $t$. Our approach is based on this connection between approximate shortest path computations and approximate flows, and we obtain the following results.

\begin{theorem}
\label{thm: informal max s-t flow and min s-t cut}
There is a randomized algorithm, that, given a simple undirected graph $G=(V,E)$ with capacities $c(v)\geq 0$ on its vertices, a source $s$, a sink $t$,  and an accuracy parameter $\epsilon \in (0,1]$,  computes 
 a $(1+ \eps)$-approximate maximum  $s$-$t$ flow and a $(1+ \eps)$-approximate minimum vertex $s$-$t$ cut in ${O}(n^{2+o(1)}/\poly(\eps))$ expected time.
\end{theorem}

Our proof closely follows the analysis of the primal-dual approach for maximum multicommodity flow problem as presented in~\cite{GK98, Fleischer00}); this algorithm simultaneously outputs an approximate maximum $s$-$t$ flow and an approximate fractional minimum $s$-$t$ cut.
 The main primitive needed for this framework is the ability to compute a $(1+\eps)$-approximate shortest source-sink path in a vertex-weighted graph that is undergoing weight increases. We show that Theorem~\ref{thm: main for SSSP} can be used to implement these dynamic approximate shortest path computations in ${O}(n^{2+o(1)}/\poly(\eps))$  total expected time. The fractional $s$-$t$ cut solution can be rounded in $O(m)$ time by using the standard random threshold rounding. The running time obtained in Theorem~\ref{thm: informal max s-t flow and min s-t cut} outperforms previously known bounds in the regime of $m = \omega(n^{1.5 + o(1)})$. 
 
Our second application is a new algorithm for approximating vertex sparsest cut in undirected graphs.
A \emph{vertex cut} in a graph $G$ is  a partition $(A,X,B)$ of its vertices, so that there is no edge from $A$ to $B$ (where $A$ or $B$ may be empty). The \emph{sparsity} of the cut is $\frac{|X|}{\min\set{|A|,|B|}+|X|}$.  In the vertex sparsest cut problem, the goal is to compute a vertex cut of minimum sparsity.


\begin{theorem}
\label{thm: informal sparsest cut}
There is a randomized algorithm, that, given a simple undirected graph $G=(V,E)$,  computes an $O(\log^4 n)$-approximation to the vertex sparsest cut problem in ${O}(n^{2+o(1)})$ expected time.
\end{theorem}


To establish the above result, it suffices to design an algorithm that runs in ${O}(n^{2+o(1)})$ expected time, and for any target value $\alpha$, either finds a vertex cut of sparsity $O(\alpha)$ or certifies that the sparsity of any vertex cut is $\Omega(\alpha/\log^4 n)$. 
We design such an algorithm by using the cut-matching game of Khandekar, Rao, and Vazirani~\cite{KRV}. Roughly speaking, the game proceeds in rounds, where in each round a bipartition of vertices is given, and the goal is to find a routing from one set to the other with vertex congestion at most $1/\alpha$. This is essentially the vertex-capacitated $s$-$t$ flow problem, and we can use ideas similar to the one in Theorem~\ref{thm: informal max s-t flow and min s-t cut} to solve it. If every round of the cut-matching game can be successfully completed, then we have successfully embedded an expander that certifies that vertex sparsity is $\Omega(\alpha/\log^4 n)$. 
On the other hand, if any round of the game fails, then we show that we can output a vertex cut of sparsity at most $O(\alpha)$. The running time of this approach is governed by the time needed to solve the vertex-capacitated maximum $s$-$t$ flow problem, and we utilize Theorem~\ref{thm: informal max s-t flow and min s-t cut} to implement this step in ${O}(n^{2+o(1)})$ expected time.  
Alternatively, one can implement the vertex-capacitated maximum $s$-$t$ flow step using the algorithms for computing maximum $s$-$t$ flow in edge-capacitated directed graphs in $\tilde{O}(m \sqrt{n})$ time in dense graphs~\cite{LeeS14}, or in $\tilde{O}(m^{10/7})$ time in sparse graphs~\cite{Madry16}.
Thus an identical approximation guarantee to the one established in Theorem~\ref{thm: informal sparsest cut}
can be obtained in $\tilde{O}( \min\{ m \sqrt{n} , m^{10/7} \})$ time using previously known results~\cite{LeeS14, Madry16}. Another approach for the vertex sparsest cut problem is to use the primal-dual framework of Arora and Kale~\cite{AroraK16}, who achieve an $O(\sqrt{\log n})$ approximation for the directed sparsest cut problem in
$\tilde{O}(m^{1.5} + n^{2 + \eps})$ time and an $O(\log n)$-approximation in $\tilde{O}(m^{1.5})$ time. Since directed sparsest cut captures vertex sparsest cut in undirected graphs as a special case, these guarantees also hold for the vertex sparsest cut problem. 

As before, the running time obtained in Theorem~\ref{thm: informal sparsest cut} starts to outperform previously known bounds in the regime of $m = \omega(n^{1.5 + o(1)})$, albeit achieving a worse approximation ratio than the one achieved in~\cite{AroraK16}.

\paragraph{Other related work} All-Pairs Shortest-Paths (APSP) can be seen as a generalization of SSSP, where, instead of maintaining distances and shortest paths from a given source vertex $s$, we need to support \pathquery and \distquery for any given pair of vertices. Much of the work on APSP was done in the fully-dynamic setting, where the edges can be both inserted and deleted, with the focus on bounding the amortized and the worst-case time \emph{per update}. In this setting,
 a long line of work has culminated in  a breakthrough result of Demetrescu and Italiano~\cite{APSP-fully}, who designed an exact APSP algorithm for directed graphs with non-negative edge-lengths, with \emph{amortized} update time of $\tilde O(n^2)$ per update; the result was later extended to handle negative edge lengths by Thorup~\cite{thorupfully}, who also showed an algorithm that provides an $\tilde O(n^{2.75})$ worst-case update time per update, in a graph with non-negative edge lengths~\cite{thorup2005worst}.
All these papers consider \emph{vertex updates}, where all edges incident to a given vertex can be updated in a single operation. Note that when both insertions and deletions are allowed, individual edge updates can be implemented via vertex updates. 
Bernstein~\cite{APSPfully2} presented an algorithm that returns $(2+\eps)$-approximate answers to distance queries in undirected graphs with non-negative edge lengths with expected amortized $O(mn^{o(1)}\log L)$  update time per operation and $O(\log\log\log n)$ query time. Baswana, Khurana and Sarkar~\cite{APSPfully3} considered undirected unweighted graphs, and provided a $(4k)$-approximation algorithm with $O(n^{1+1/k+o(1)})$ amortized update time and $O(\log\log\log n)$ query time, for any given integral parameter $k$. Sankowsky~\cite{APSPfully4}
achieves an $O(n^{1.932})$ worst-case update time per operation and $O(n^{1.288})$ query time.

In the partially dynamic setting, we are allowed to only delete or to only insert edges. 
For the case where we are interested in obtaining a $(1+\eps)$-approximation to shortest path queries, the best current algorithm achieves a total update time of $\tilde O(mn\log L)$ \cite{BaswanaHS07,rodittyZ2,henzinger16, bernstein16} even on directed weighted graphs. 
Another setting that was studied is where we allow a higher approximation factor. Suppose we are given an integral parameter $k$, and we are interested in studying the tradeoff between the algorithm's approximation factor and  its total update time, as a function of $k$. For the deletion-only setting, Roditty and Zwick~\cite{rodittyZ2} achieve a $(2k-1)$-approximation, with $\tilde{O}(mn)$ total update time, and $O(m+n^{1+1/k})$ space. Bernstein and Roditty~\cite{BernsteinR11} provide an algorithm achieving an approximation of $(2k-1+\eps)$, and $\tilde O(n^{2+1/k+o(1)})$ total update time for unweighted undirected graphs. The results of \cite{HenzingerKN14_focs,abraham2014fully} give $(2+\eps)^k-1$ approximation, $O(k^k)$ query time, and  $O(m^{1+1/k+o(1)}\log^2L)$ total update time, in decremental setting for undirected weighted graphs. Lastly, a recent result of Chechik~\cite{chechik} obtains a $((2+\eps)k-1)$-approximation, $O(\log\log(nL))$ query time, and $O(mn^{1/k+o(1)}\log L)$ total update time. The algorithms of \cite{HenzingerKN14_focs,abraham2014fully,chechik} are all randomized and assume an oblivious adversary.

\paragraph{Subsequent Work}
In a follow-up work, Chuzhoy and Saranurak~\cite{SSSPandAPSP} have extended our results to edge-decremental \SSSP, obtaining total expected update time $\tilde O(n^2\log L/\eps^2)$. This immediately also improves the expected running times of the algorithms for approximate maximum $s$-$t$ flow, minimum $s$-$t$ cut and vertex sparest cut from Theorems~\ref{thm: informal max s-t flow and min s-t cut} and~\ref{thm: informal sparsest cut} to $\tilde O(n^2/\poly(\eps))$ and $\tilde O(n^2)$, respectively. They also obtain a new algorithm for edge-decremental All-Pairs Shortest Paths in unweighted undirected graphs with adaptive adversary.
The algorithm obtains a constant multiplicative and a $\poly\log n$ additive approximation factors, with expected total update time $O(n^{2.67})$. This is the first approximation algorithm for the problem in the adaptive adversary setting whose running time is asymptotically less than $\Theta(n^3)$. The algorithm builds on some of the ideas and techniques introduced in this paper.

\paragraph{Organization}
We start with an overview of the proof of Theorem~\ref{thm: main for SSSP} in Section~\ref{sec: techniques}. We then provide preliminaries in Section~\ref{sec: prelims}. Section~\ref{sec: main for SSSP} contains the proof of Theorem~\ref{thm: main for SSSP}, with the algorithm for computing the core decomposition deferred to Sections~\ref{sec: computing core decomp} and~\ref{sec: balanced cut or core}. Sections~\ref{sec: max flow}  and~\ref{sec: sparsest cut} contain the applications of our main result to vertex-capacitated maximum $s$-$t$ flow and minimum $s$-$t$ cut, and vertex sparsest cut, respectively.

%% file: overview.tex
\section{Overview of the Proof of Theorem~\ref{thm: main for SSSP}}\label{sec: techniques}

We now provide an overview of our main result, namely, the proof of Theorem~\ref{thm: main for SSSP}. This informal overview is mostly aimed to convey the intuition; in order to simplify the discussion, the values of some of the parameters and bounds are given imprecisely.
As much of the previous work in this area, our results use the classical Even-Shiloach trees~\cite{EvenS,Dinitz,HenzingerKing} as a building block. Given a graph $G$ with integral edge lengths, that is subject to edge deletions, a source vertex $s$, and a distance bound $D$, the Even-Shiloach Tree data structure, that we denote by $\EST(G,s,D)$, maintains a shortest-path tree $T$ of $G$, rooted at $s$, up to distance $D$. In other words, a vertex $v\in V(G)$ belongs to $T$ iff $\dist(s,v)\leq D$, and for each such vertex $v$, $\dist_T(s,v)=\dist_G(s,v)$. The total update time of the algorithm is $O(|E(G)|\cdot D\cdot \log n)$. More precisely, for every vertex $v\in V(G)$, whenever $\dist(s,v)$ increases (which may happen at most $D$ times over the course of the algorithm, since all edge lengths are integral), the algorithm performs an inspection of all neighbors of $v$, contributing $O(d(v)\log n)$ to the running time of the algorithm, where $d(v)$ is the degree of $v$ in $G$. A simple accounting shows that the total update time of this algorithm is indeed $O(|E(G)|\cdot D\cdot \log n)$. In addition to maintaining the shortest-path tree $T$, the data structure stores, with every vertex $v\in V(T)$, the value $\dist_G(s,v)$.

 At a high level, our algorithm follows the framework of~\cite{BernsteinChechik,Bernstein}.  Using standard techniques, we can reduce the problem to a setting where we are given a parameter $D=\Theta(n/\eps)$, and we only need to correctly respond to $\pquery(v)$ if $D\leq \dist(s,v)\leq 4D$; otherwise we can return an arbitrary path, or no path at all. Let us assume first for simplicity that all edges in the graph $G$ have unit length. In~\cite{BernsteinChechik,Bernstein}, the algorithm proceeds by selecting a threshold $\tau\approx \frac{n}{\eps D}$, and splitting the graph $G$ into two subgraphs, a sparse graph $G^L$, called the \emph{light graph}, and a dense graph $G^H$, called the \emph{heavy graph}. In order to do so, we say that a vertex $v\in V(G)$ is \emph{heavy} if $d(v)\geq\tau$, and it is \emph{light} otherwise. Graph $G^L$ contains all vertices of $G$ and all edges $e=(u,v)$, such that at least one of $u,v$ is a light vertex; notice that $|E(G^L)|\leq n\tau\leq O(n^2/\eps D)$. Graph $G^H$ contains all heavy vertices of $G$, and all edges connecting them.
 We assume for now for the simplicity of exposition that all vertex degrees in $G^H$ are at least $\tau$.
  The algorithm also maintains the \emph{extended light graph} $\hat G^L$, that is obtained from $G^L$, by adding, for every connected component $C$ of $G^H$, a vertex $v_C$ to $\hat G^L$, and connecting it to every heavy vertex $u\in C$ with an edge of weight $1/2$. So, in a sense, in $\hat G^L$, we create ``shortcuts'' between the heavy vertices that lie in the same connected component of $G^H$. The crux of the algorithm consists of two observations: (i) for every vertex $v$, if $D\leq \dist_G(s,v)<4D$, then $\dist_G(s,v)\approx\dist_{\hat G^L}(s,v)$; and (ii)  since graph $\hat G^L$ is sparse, we can maintain, for every vertex $v\in V(G)$ with $\dist_{\hat G^L}(s,v)\leq 4D$, the distances $\dist_{\hat G^L}(s,v)$ in total update time $O(n^2/\eps)$. In order to see the latter, observe that $|E(\hat G^L)|\leq |E(G^L)|+O(n)\leq O(n^2/\eps D)$. We can use the data structure $\EST(\hat G^L,s,D)$, with  total update time $O(|E(\hat G^L)|\cdot D\cdot\log n)=O(n^2\log n/\eps  )$ (in fact, the threshold $\tau$ was chosen to ensure that this bound holds). In order to establish (i), observe that graph $\hat G^L$ is obtained from graph $G$, by ``shortcutting'' the edges of $G^H$, and so it is not hard to see that $\dist_{\hat G^L}(s,v)\leq \dist_{G}(s,v)$ for all $v\in V(G)$. The main claim is that, if $\dist_G(s,v)<4D$, then $\dist_{G}(s,v)\leq (1+\eps)\dist_{G^L}(s,v)$ also holds, and in particular that for any path $P$ in $\hat G^L$ connecting the source $s$ to some vertex $v\in V(G)$, there is  a path $P'$ in $G$ connecting $s$ to $v$, such that the length of $P'$ is at most the length of $P$ plus $\eps D$. Assuming this is true, it is easy to verify that for every vertex $v$ with $D\leq \dist(s,v)\leq 4D$, $\dist_{G}(s,v)\leq \dist_{\hat G^L}(s,v)(1+\eps)$, and so it is sufficient for the algorithm to report, as an answer to a query $\pquery(v)$, the value $\dist_{\hat G^L}(s,v)$, which is stored in $\EST(\hat G^L,s,D)$. Consider now some path $P$ in $\hat G^L$, and let $C$ be any connected component of $G^H$, such that $v_C\in P$. Let $u,u'$ be the vertices of the original graph $G$ appearing immediately before and immediately after $v_C$ in $P$. Let $Q(u,u')$ be the shortest path connecting $u$ to $u'$ in the heavy graph $G^H$. As every vertex in $G^H$ is heavy, the length of $Q(u,u')$ is bounded by $4|V(C)|/\tau$: indeed, assume that $Q(u,u')=(u=u_0,u_1,\ldots,u_r=u')$, and let $S=\set{u_i\mid i=1\mod 4}$ be a subset of vertices on $Q(u,u')$. Then for every pair $u_i,u_j$ of distinct vertices in $S$, the set of their neighbors must be disjoint (or we could shorten the path $Q(u,u')$ by connecting $u_i$ to $u_j$ through their common neighbor). Since we have assumed that every vertex in $G^H$ has at least $\tau$ neighbors in $G^H$, $|S|\leq |V(C)|/\tau$, and so $Q(u,u')$ may contain at most $4|V(C)|/\tau$ vertices. Once we replace each such vertex $v_C$ on path $P$ with a path connecting the corresponding pair $u,u'$ of vertices in the original graph, the length of $P$ increases by at most $\sum_{C: v_C\in P}4|V(C)|/\tau\leq 4n/\tau=O(\eps D)$. This argument allows the algorithms of~\cite{BernsteinChechik,Bernstein}  to maintain approximate distances from the source $s$ to every vertex of $G$, by simply maintaining the data structure $\EST(\hat G^L,s,D)$. However,  in order to recover the {\bf path} connecting $s$ to the given query vertex $v$ in $G$, we should be able to compute all required paths in the heavy graph $G^H$. Specifically, we need an algorithm that 
 allows us to answer queries $\pquery(u,u',C)$: given a connected component $C$ of $G^H$, and a pair $u,u'$ of vertices of $C$, return a path connecting $u$ to $u'$ in $C$, whose length is at most $O(|V(C)|/\tau)$. The main contribution of this work is an algorithm that allows us to do so, when the input graph $G$ is subject to {\bf vertex} deletions. (We note that for technical reasons, the value $\tau$ in our algorithm is somewhat higher than in the algorithms of~\cite{BernsteinChechik,Bernstein}, which translates to a somewhat higher running time $O(n^{2+o(1)}\log^2(1/\eps)/\eps^2)$, where $o(1)=\Theta(1/\log\log n)$. We also define the light and the heavy graphs somewhat differently, in a way that ensures that all vertex degrees in $G^H$ are indeed at least $\tau$, while $|E(G^L)|=O(n\tau)$.)

\paragraph{A first attempt at a solution.}
For simplicity of exposition, let us assume that all vertices in the heavy graph $G^H$ have approximately the same degree (say between $h$ and $2h$, where $h\geq \tau$ is large enough, so, for example, $h\geq n^{1/\log\log n}$), so the number of edges in $G^H$ is $O(hn)$. Using the same argument as before, for every connected component $C$ in $G^H$, and every pair $u,u'\in V(C)$ of its vertices, there is a path connecting $u$ to $u'$ in $C$, of length $O(|V(C)|/h)$; we will attempt to return paths whose lengths are bounded by this value in response to queries.
A tempting simple solution to this problem is the following: for every connected component $C$ of $G^H$, select an arbitrary vertex $s(C)$ to be its source, and maintain the $\EST(C,s(C),D(C))$ data structure, for the distance bound $D(C)\approx |V(C)|/h$. Such a tree can be maintained in total time $\tilde O(|E(C)|\cdot |V(C)|/h)$, and so, across all components of $G^H$, the total update time is $\tilde O(|E(G^H)|n/h)=\tilde O(n^2)$. 
Whenever a query $\pquery(u,u',C)$ arrives, we simply concatenate the path connecting $u$ to $s(C)$ and the path connecting $u'$ to $s(C)$ in the tree $\EST(C,s(C),D(C))$; using the same argument as before, it is easy to show that the resulting path is guaranteed to be sufficiently short.
Note that, as the algorithm progresses, the connected component $C$ may decompose into smaller connected components, and we cannot afford to recompute the $\EST$ data structure for each connected component from scratch. A natural simple solution to this is the following. Since we maintain a shortest-path tree $T_C$ for $C$, rooted at $s(C)$, whenever $C$ is decomposed into two components $C_1$ and $C_2$, the tree $T$ naturally decomposes into two trees: tree $T_{C_1}$ spanning $C_1$ and tree $T_{C_2}$ spanning $C_2$. We could then continue maintaining Even-Shiloach trees for $C_1$ and $C_2$, respectively, using the roots of the trees $T_{C_1}$ and $T_{C_2}$ as sources, so we do not need to recompute the trees from scratch.

Unfortunately, this simple approach does not seem to work. Consider the following bad scenario. We partition our algorithm into phases. In every phase, the adversary considers the current $\EST$ maintained by the algorithm for the connected component $C$ of $G^H$, whose source is denoted by $s_1(C)$. The adversary then produces an edge-deletion sequence, that iteratively deletes every edge of $C$ incident to $s_1(C)$. Once the phase ends, vertex $s_1(C)$ becomes disconnected from $C$, and the algorithm is forced to select a new source vertex, say $s_2(C)$ (a natural choice would be the child $v$ of $s_1(C)$ in the tree that disconnected from it last, as in this case we can reuse the current tree structure  rooted at $v$ and do not need to recompute it from scratch). We then continue to the second phase, where the adversary deletes one-by-one every edge incident to $s_2(C)$, and so on. Recall that the analysis of the \EST algorithm relies on the fact that, whenever the distance of a vertex $v$ from the source of the tree increases, the vertex contributes $O(d(v)\log n)$ to the running time of the algorithm, and we may have at most $D$ such increases for each vertex over the course of the algorithm. In the above scenario, it is possible that, over the course of the first phase, the distance of every vertex of $C$ from $s_1(C)$ increases, say by $1$, and we are forced to spend $|E(C)|$ time to update the tree. However, once the phase ends and the new source vertex $s_2(C)$ is selected, it is possible (and in fact likely) that for every vertex $v\in C$, $\dist(s_2(C),v)<\dist(s_1(C),v)$, so all distances from the (new) source decrease back. On the one hand, the algorithm may be forced to spend $O(|E(C)|)$ time per phase, but on the other hand we no longer have a good bound on the number of phases, as the distances of the vertices from the successive sources may decrease and increase iteratively. 
We note that in the vertex-deletion setting, a similar bad scenario may occur, when the adversary iteratively deletes all vertices that are children of the current source vertex in the tree. 

Interestingly, this seemingly artificial bad scenario is likely to arise in applications of the algorithms for decremental \SSSP to the maximum flow problem, where the paths returned by the algorithm as a response to $\pquery(v)$ are used to route the flow, and the vertices lying on these paths are subsequently deleted. The algorithm outlined above computes paths that contain the sources $s(C)$ of the connected components $C$ of $G^H$, and so vertices that are children of the sources $s(C)$  in the current tree are most likely to be quickly deleted.

\paragraph{Solution: core decomposition.}
A natural approach to overcome this difficulty is to create, in every connected component $C$ of $G^H$ a ``super-source'', that would be difficult to disconnect from the rest of the component $C$. 
This motivates the notion of \emph{cores} that we introduce. Recall that we have assumed for now that the degrees of all vertices in $G^H$ are between $h$ and $2h$, where $h\geq n^{1/\log\log n}$. Intuitively, a core is a highly-connected graph. For example, a good core could be an expander graph $K$, such that every vertex $v\in V(H)$ has many neighbors in $K$ (say, at least $h/n^{o(1)}$). If we use a suitable notion of expander, this would ensure that, even after a relatively long sequence of vertex deletions (say up to $h/n^{o(1)}$), every pair of vertices in $K$ has a short path connecting them.  Intuitively, we would like to use the core as the ``super-source'' of the \EST structure. Unfortunately, the bad scenario described above may happen again, and the adversary can iteratively delete vertices in order to isolate the core from the remainder of the graph. To overcome this difficulty, we use the notion of \emph{core decomposition}. A core decomposition is simply a collection of disjoint cores in $G^H$, but it has an additional important property: If $U$ is the set of all vertices of $G^H$ that do not lie in any of the cores, then it must be an \emph{$h$-universal set}: namely, after a sequence of up to $h/n^{o(1)}$ deletions of vertices from $G^H$, each remaining vertex of $U$ should be able to reach one of the cores using a short path (say, of length at most $\poly\log n$). 
Our algorithm then uses the cores as the ``super-source'', in the following sense. We construct a new graph $\tilde G$, by starting from $G^H$ and contracting every core $K$ into a super-node $z(K)$. We also add a new source vertex $s$, that connects to each resulting super-node. Our algorithm then maintains $\EST(\tilde G,s,\poly\log n)$, that allows us to quickly recover a short path connecting any given vertex of $G^H$ to some core.
 One of the technical contributions of this paper is an algorithm that computes a core decomposition in time $\tilde O(|E(G^H)|+ n^{1+o(1)})$. Before we discuss the core decomposition, we quickly summarize how our algorithm processes shortest-path queries, and provide a high-level analysis of the total update time. 

\paragraph{Responding to queries.} Recall that in  $\pquery(u,u',C)$, we are given a connected component $C$ of $G^H$, and a pair $u,u'$ of its vertices. Our goal is to return a path connecting $u$ to $u'$, whose length is at most $O(|V(C)|/\tau)$; in fact we will return a path of length at most $O(|V(C)|/h)$. Recall that for every core $K$, we require that every vertex $v\in K$ has at least $h/n^{o(1)}$ neighbors in $K$ (so in particular $|V(K)|\geq h/n^{o(1)}$), and that all cores in the decomposition are disjoint. Therefore, the total number of cores contained in $C$ is at most $|V(C)|n^{o(1)}/h$. We will maintain a simple spanning forest data structure in graph $G^H$ that allows us, given a pair $u,u'$ of vertices that belong to the same connected component $C$ of $G^H$, to compute an arbitrary simple path $P$ connecting $u$ to $u'$ in $C$. Next, we \emph{label} every vertex $w$ of $P$ with some core $K$: if $w$ belongs to a core $K$, then the label of $w$ is $K$; otherwise, the label of $w$ is any core $K$, such that $w$ can reach $K$ via a short path (of length $\poly\log n$). The labeling is performed by exploiting the  $\EST(\tilde G,s,\poly\log n)$ data structure described above. Once we obtain a label for every vertex on the path $P$, we ``shortcut'' the path through the cores: if two non-consecutive vertices of $P$ have the same label $K$, then we  delete all vertices lying on $P$ between these two vertices, and connect these two vertices via the core $K$. As the number of cores in $C$ is at most $|V(C)|n^{o(1)}/h$, eventually we obtain a path connecting $u$ to $u'$, whose length is $|V(C)|n^{o(1)}\poly\log n/h=|V(C)|n^{o(1)}/h$, as required.

\paragraph{Running time analysis.} As already mentioned, our algorithm for computing the core decomposition takes time $\tilde O(|E(G)|+ n^{1+o(1)})=O(n^{1+o(1)}h)$; it seems unlikely that one can hope to obtain an algorithm whose running time is less than $\Theta(|E(G^H)|)=\Theta(nh)$. Our core decomposition remains ``functional'' for roughly $h/n^{o(1)}$ iterations, that is, as long as fewer than $h/n^{o(1)}$ vertices are deleted. Once we delete $h/n^{o(1)}$ vertices from the graph, we are no longer guaranteed that pairs of vertices within the same core have short paths connecting them (in fact they may become disconnected), and we are no longer guaranteed that the vertices of $U$ can reach the cores via short paths. Therefore, we partition our algorithm into phases, where every phase consists of  the deletion of up to $h/n^{o(1)}$ vertices. Once $h/n^{o(1)}$ vertices are deleted, we recompute the core decomposition, the graph $\tilde G$, and the $\EST(\tilde G,s,\poly\log n)$ data structure that we maintain. Note that, since $G$ has $n$ vertices, 
the number of phases is bounded by $n^{1+o(1)}/h$, and recall that we spend $O(n^{1+o(1)}h)$ time per phase to recompute the core decomposition. Therefore, the total update time of the algorithm is $O(n^{2+o(1)})$ (we have ignored multiplicative factors that depend on $\eps$).

\paragraph{Why our algorithm only handles vertex deletions.} As mentioned above, it is unlikely that we can compute a core decomposition in less than 
$\Theta(|E(G^H)|)=\Theta(nh)$ time. If our goal is a total update time of $O(n^{2+o(1)})$, then we can afford at most $O(n^{1+o(1)}/h)$ computations of the core decomposition. If we allow edge deletions, this means that a phase may include up to roughly $h^{2}/n^{o(1)}$ edge deletions, since $|E(G^H)|=\Theta(nh)$. Since the degrees of the vertices are between $h$ and $2h$, the cores cannot handle that many edge deletions, as they can cause an expander graph to become disconnected, or some vertices of $U$ may no longer have short paths connecting them to the cores. However, in the vertex-deletion model, we only need to tolerate the deletion of up to roughly $h/n^{o(1)}$ vertices per phase, which we are able to accommodate, as the degrees of all vertices are at least $h$.

\paragraph{The core decomposition.} The main technical ingredient of our algorithm is the core decomposition. In the vertex-deletion model, it is natural to define a core $K$ as a \emph{vertex expander}: that is, for every vertex-cut $(X,Y,Z)$ in $K$ (so no edges connect $X$ to $Z$ in $K$), $|Y|\geq \min\set{|X|,|Z|}/n^{o(1)}$ must hold. Additionally, as mentioned above, we require that every vertex in $K$ has at least $h/n^{o(1)}$ neighbors that lie in $K$. Unfortunately, these requirements appear too difficult to fulfill. For instance, a natural way to construct a core-decomposition is to iteratively decompose the graph $G^H$ into connected clusters, by computing, in every current cluster $R$, a sparse vertex cut $(X,Y,Z)$, and then replacing $R$ with two new graphs: $R[X\cup Y]$ and $R[Y\cup Z]$. We can continue this process, until every resulting graph is a vertex expander. Unfortunately, this process does not ensure that the resulting cores are disjoint, or that every vertex in a core has many neighbors that also belong to the core. Moreover, even if all pairs of vertices within a given core $K$ have short paths connecting them, it is not clear how to recover such paths, unless we are willing to spend $O(|E(K)|)$ time on each query. Therefore, we define the cores somewhat differently, by using the notion of a \emph{core structure}. Intuitively, a core structure consists of two sets of vertices: set $K$ of vertices -- the core itself, and an additional set $U(K)$ of at most $|K|$ vertices, called the \emph{extension of the core}. We will ensure that all core-sets $K$ are disjoint, but the extension sets may be shared between the cores. Additionally, we are given a sub-graph $G^K$ of $G^H$, whose vertex set contains $K$ and is a subset of $K\cup U(K)$. We will ensure that all such sub-graphs are ``almost'' disjoint in their edges, in the sense that every edge of $G^H$ may only belong to at most $O(\log n)$ such graphs, as this will be important in the final bound on the running time. Finally, the core structure also contains a \emph{witness graph} $W^K$ - a sparse graph, that is a $1/n^{o(1)}$-expander (in the usual edge-expansion sense), whose vertex set includes every vertex of $K$, and possibly some additional vertices from $U(K)$. We also compute an \emph{embedding} of $W^K$ into $G^K$, where each edge $e=(u,v)\in E(W^K)$ is mapped to some path $P_e$ in $G^K$, connecting $u$ to $v$, such that all such paths $P_e$ are relatively short, and they cause low vertex-congestion in $G^K$. The witness graph $W^K$ and its embedding into $G^K$ allow us to quickly recover short paths connecting pairs of vertices in the core $K$.

One of the main building blocks of our core decomposition is an algorithm that, given a subgraph $H$ of $G$, either computes a sparse and almost balanced vertex-cut in $H$, or returns a core containing most vertices of $H$. The algorithm attempts to embed an expander into $H$ via the cut-matching game of~\cite{KRV}. If it fails, then we obtain a sparse and almost balanced vertex-cut in $H$. Otherwise, we embed a graph $W$ into $H$, that is with high probability an expander. Graph $W$ then serves as the witness graph for the resulting core. The cut-matching game is the only randomized part of our algorithm. If it fails (which happens with low probability), then one of the queries to the heavy graph may return a path whose length is higher than the required threshold (that is known to the algorithm). In this case, we simply recompute all our data structures from scratch. This ensures that our algorithm always returns a correct approximate response to \pathquery, with the claimed expected running time, and is able to handle an adaptive adversary.

\paragraph{Handling arbitrary vertex degrees.}
Recall that so far we have assumed that all vertices in $G^H$ have similar degrees. This was necessary because, if some vertices of $G^H$ have low degrees (say $d$), but $|E(G^H)|$ is high (say $\Theta(nh)$ for some $h\gg d$), then we would be forced to recompute the core decomposition very often, every time that roughly $d$ vertices are deleted, while each such computation takes at least $\Theta(n^{1+o(1)}h)$ time, resulting in a total running time that is too high. To overcome this difficulty, we partition the heavy graph $G^H$ into graphs $\Lambda_1,\ldots,\Lambda_r$ that we call \emph{layers}, where for each $1\leq i\leq r$, all vertices in graph $\Lambda_i$ have degree at least $h_i$, while $|E(\Lambda_i)|\leq n^{1+o(1)}h_i$. We ensure that $h_1\geq h_2\geq\ldots,\geq h_r$, and that these values are geometrically decreasing. We maintain a core decomposition for each such graph $\Lambda_i$. For all $1\leq i\leq r$, roughly every $h_i/n^{o(1)}$ vertex deletions, we recompute the layers $\Lambda_i,\ldots,\Lambda_r$, and their corresponding core decompositions.

\paragraph{Handling arbitrary edge lengths.} So far we have assumed that all edge lengths are unit. When the edge lengths are no longer the same, we need to use the approach of~\cite{Bernstein}. We partition all edges into classes, where class $i$ contains all edges whose length is between $2^i$ and $2^{i+1}$. Unfortunately, we can no longer use the same threshold $\tau$ for the definition of the heavy and the light graph for all edge lengths. This is since we are only guaranteed that, whenever two vertices $u,u'$ belong to the same connected component $C$ of $G^H$, there is a path containing at most $|V(C)|/\tau$ edges connecting $u$ to $u'$ in $C$. But as some edges may now have large length, the actual length of this path may be too high. Following~\cite{Bernstein}, we need to define different thresholds $\tau_i$ for each edge class $i$, where roughly $\tau_i=\tau\cdot 2^i$, for the original threshold $\tau$. This means that graph $\hat G^L$ may now become much denser, as it may contain many edges from classes $i$ where $i$ is large. We use the Weight-Sensitive Even-Shiloach data structure of~\cite{Bernstein} in order to handle $\hat G^L$. Roughly speaking, his algorithm modifies the $\EST$ algorithm, so that edges with higher weight contribute proportionally less to the total update time of the algorithm.

%% file: prelims.tex
\section{Preliminaries}\label{sec: prelims}

We follow standard graph-theoretic notation. All graphs in this paper are undirected, unless explicitly said otherwise. Graphs may have parallel edges, except for simple graphs, that cannot have them.  Given a graph $G=(V,E)$ and two disjoint subsets $A,B$ of its vertices, we denote by $E_G(A,B)$ the set of all edges with one endpoint in $A$ and another in $B$, and by $E_G(A)$ the set of all edges with both endpoints in $A$. We also denote by $\out_G(A)$ the set of all edges with exactly one endpoint in $A$. 
We may omit the subscript $G$ when clear from context. Given a subset $S\subseteq V$ of vertices of $G$, we denote by $G[S]$ the sub-graph of $G$ induced by $S$.

A \emph{cut} in $G$ is a partition $(A,B)$ of $V$ into two disjoint subsets, with $A,B\neq\emptyset$. 
The \emph{sparsity} of the cut $(A,B)$ is $\frac{|E(A,B)|}{\min\set{|A|,|B|}}$.


\begin{definition}
We say that an graph $G$ is an $\alpha$-expander, for $\alpha>0$, iff every cut $(A,B)$ in $G$ has sparsity at least $\alpha$, or, equivalently, $|E(A,B)|\geq \alpha\cdot\min\set{|A|,|B|}$.
\end{definition}

We now define vertex cuts and their sparsity.
A \emph{vertex cut} in graph $G$ is a partition $(X,Y,Z)$ of $V(G)$ into three subsets, such that there is no edge in $G$ connecting a vertex of $X$ to a vertex of $Z$. The value of the cut is $|Y|$, and its sparsity is $\psi(X,Y,Z)=\frac{|Y|}{\min\set{|X|,|Z|}+|Y|}$.


\paragraph{The Cut-Matching Game.}
We use the cut-matching game of Khandekar, Rao and Vazirani~\cite{KRV}, defined as follows. We are given a set $V$ of $N$ vertices, and two players, called the cut player and the matching player. 
The game is played in iterations. We start with a graph $\mX$ with node set $V$ and an empty edge set. In every iteration, some edges are added to $\mX$. The game ends when $\mX$ becomes a $\half$-expander. 
The goal of the cut player is to construct a $\half$-expander in as few iterations as
possible, whereas the goal of the matching player is to prevent the construction of the expander for as long as possible. The iterations proceed as follows.  In every iteration $j$, the cut player chooses two disjoint subsets $Y_j,Z_j$ of $V$ with $|Y_j| = |Z_j|$ and the
matching player chooses a perfect matching $M_j$ that matches the nodes of $Y_j$ to the nodes of $Z_j$.  The edges of $M_j$ are then
added to $\mX$.  Khandekar, Rao, and Vazirani \cite{KRV} showed that there is a strategy for the cut player that guarantees that after
$O(\log^2{N})$ iterations the graph $\mX$ is a $(1/2)$-expander with high probability.\footnote{In fact, in the algorithm of~\cite{KRV}, $N$ is even, and the cut player computes, in each iteration $j$, a bi-parititon $(Y_j,Z_j)$ of $V$ into two equal-sized subsets. 
Their algorithm can be easily adapted to the setting where $N$ is odd: let $v,v'$ be two arbitrary distinct vertices from $V$; run the algorithm of~\cite{KRV} on $V\setminus\set{v}$, and then on $V\setminus{v'}$. The final graph, obtaining by taking the union of the two resulting sets of edges, is a $\half$-expander w.h.p.}
Orecchia \etal \cite{OrecchiaSVV08} strengthened
this result by showing that, after $O(\log^2{N})$ iterations, the graph $\mX$ is an $\Omega(\log{N})$-expander with constant probability, by using a different strategy for the cut player. 

Let $\mX_i$ denote the graph computed after $i$ iterations of~\cite{KRV}, so $\mX_0$ is a graph on $N$ vertices and no edges. The following is the main result of~\cite{KRV}.

\begin{theorem}[\cite{KRV}] \label{thm:cut-matching-game}
	There is a constant $\cKRV$ and a randomized algorithm, that, for each $i\geq 1$, given the current graph $\mX_i$, computes, in time $O(N\poly\log N)$ the subsets $A_{i+1},B_{i+1}$ of $V$ to be used as a response of the cut player, such that, regardless of the responses of the matching player, the graph $\mX_{T}$ obtained after $T=\floor{\cKRV\log^2N}$ iterations is a $1/2$-expander, with probability at least $1-1/\poly(N)$.
\end{theorem}

(Note that the graphs $\mX_i$ themselves do depend on the responses of the matching player).
Observe that the resulting expander may have parallel edges, and its maximum vertex degree bounded by $\cKRV\log^2N$.

\paragraph{Decremental Connectivity/Spanning Forest.}
We use the results of~\cite{dynamic-connectivity}, who provide a deterministic data structure, that we denote by $\CONNSF(G)$, that, given an $n$-vertex unweighted undirected graph $G$, that is subject to edge deletions, maintains a spanning forest of $G$, with total running time $O((m+n)\log^2n)$, where $n=|V(G)|$ and $m=|E(G)|$. Moreover, the data structure supports connectivity queries: given a pair  $u,v$ of vertices of $G$, return ``yes'' if $u$ and $v$ are connected in $G$, and ``no'' otherwise. The running time to respond to each such query is  $O(\log n/\log\log n)$; we denote by  $\conn(G,u,u')$ the connectivity query for $u$ and $u'$ in the data structure. Since the data structure maintains a spanning forest for $G$, we can also use it to respond to a query $\path(G,u,v)$: given two vertices $u$ and $v$ in $G$, return any simple path connecting $u$ to $v$ in $G$ if such a path exists, and return $\emptyset$ otherwise. If $u$ and $v$ belong to the same connected component $C$ of $G$, then the running time of the query is $O(|V(C)|)$.

\paragraph{Even-Shiloach Trees~\cite{EvenS,Dinitz,HenzingerKing}.}
Suppose we are given a graph $G=(V,E)$ with integral lengths $\ell(e)\geq 1$ on its edges $e\in E$, a source $s$, and a distance bound $D\geq 1$. Even-Shiloach Tree (\EST) algorithm maintains, for every vertex $v$ with $\dist(s,v)\leq D$, the distance $\dist(s,v)$, under the deletion of edges from $G$. Moreover, it maintains a shortest-path tree from vertex $s$, that includes all vertices $v$ with $\dist(s,v)\leq D$. We denote the corresponding data structure by $\EST(G,s,D)$. 
The total running time of the algorithm, including the initialization and all edge deletions, is $O(m\cdot D\log n)$, where $m=|E|$.

\paragraph{Low-Degree Pruning Procedure.}  
We describe a simple procedure that our algorithm employs multiple times. The input to the procedure is a simple graph $H$ and a degree bound $d$. The procedure returns a partition $(J_1,J_2)$ of $V(H)$ into two subsets, by employing the following simple greedy algorithm: start with $J_1=\emptyset$ and $J_2=V(H)$. While there is a vertex $v\in J_2$, such that fewer than $d$ neighbors of $v$ lie in $J_2$, move $v$ from $J_2$ to $J_1$.
We denote this procedure by $\DSP(H,d)$. We use the following simple claim.

\begin{claim}\label{claim:DSP}
Procedure $\DSP(H,d)$ can be implemented to run in time $O(|E(H)|+|V(H)|)$. At the end of the procedure, the degree of every vertex in graph $H[J_2]$ is at least $d$. Moreover, for any other partition $(J',J'')$ of $V(H)$, such that the degree of every vertex in $H[J'']$ is at least $d$, $J''\subseteq J_2$ must hold.
\end{claim}
\begin{proof}
We first show that the procedure can be implemented to run in time $O(|E(H)+|V(H)|)$.  
In order to do so, we maintain, for every vertex $v\in J_2$, the number $N(v)$ of its neighbors that belong to $J_2$. We initialize the values $N(v)$ for all $v\in V(H)$ in time $O(|E(H)|)$. We also maintain a set $Q$ of vertices to be deleted from $J_2$. To initialize $Q$, we scan all vertices of $J_2=V(H)$ once at the beginning of the algorithm, and add every vertex $v$ with $N(v)<d$ to $Q$. While $Q\neq \emptyset$, we remove any vertex $v$ from $Q$, and move it from $J_2$ to $J_1$. We then inspect every neighbor $u$ of $v$ that belongs to $J_2$, and decrease $N(u)$ by $1$. If $N(u)$ falls below $d$ and $u\not\in Q$, then we add $u$ to $Q$. The algorithm terminates once $Q=\emptyset$. It is easy to verify that the algorithm can be implemented in time $O(|E(H)|+|V(H)|)$.

It is immediate to verify that, when this algorithm terminates, the degree of every vertex in graph $H[J_2]$ is at least $d$. We now prove the last assertion. Let  $(J',J'')$ be any partition of $V(H)$, such that the degree of every vertex in $H[J'']$ is at least $d$, and assume for contradiction that $J''\not\subseteq J_2$. Denote $J_1=\set{v_1,\ldots,v_r}$, where the vertices are indexed in the order in which they where added to $J_1$. Then there must be some vertex $v\in J_1$ that belongs to $J''$. Let $v_i\in J_1\cap J''$ be such a vertex with the smallest index $i$. But then $v_1,\ldots,v_{i-1}\not\in J''$, so $v_i$ has fewer than $d$ neighbors in $J''$, a contradiction.
\end{proof}

We will also need the following lemma about Procedure \DSP.

\begin{lemma}\label{lemma: DSP universal}
Let $H$ be a simple graph containing at most $n$ vertices, and let $h$ be an integer, such that the degree of every vertex in $H$ is at least $h$. Let $(A,B)$ be any partition of vertices of $H$. Suppose we apply Procedure $\DSP(H[B],\tau)$ to graph $H[B]$, for any $\tau\leq h/(32\log n)$, and let $J_1\subseteq B$ be the subset $J_1$ obtained at any time over the course of the procedure. Then for every vertex $v\in J_1$, there is a set $\pset(v)$ at least $h/(2\log n)$ paths in $H[A\cup J_1]$ (computed with respect to the current set $J_1$), connecting $v$ to vertices of $A$, such that the length of each path is at most $\log n$, and the paths in $\pset(v)$ are completely disjoint except for sharing their endpoint $v$.
\end{lemma}



\begin{proof}
We iteratively construct a collection of vertex subsets, that we call \emph{layers}, as follows. At the beginning, we only have a single layer $L_0$, and every vertex of $A$ is added to $L_0$. Whenever a new vertex $v$ is added to $J_1$ by Procedure $\DSP(H[B],\tau)$, we consider the smallest integer $i$, such that at least $h/(2\log n)$ neighbors of $v$ belong to $L_i$, and add $v$ to layer $L_{i+1}$; if no such integer $i$ exists, then we discard $v$. The crux of the analysis is in the following claim.

\begin{claim}\label{claim: small number of layers}
In every iteration of the algorithm, for all $i>0$ with $L_i\neq \emptyset$, $|L_i|<|L_{i-1}|/2$. 
\end{claim}

Notice that, if the claim is true, then the total number of layers cannot exceed $\log n$. It then follows that no vertex of $J_1$ is ever discarded. Indeed, assume for contradiction that some vertex of $J_1$ is discarded, and let $v$ be the first such vertex. Recall that the degree of $v$ is at least $h$ in $H$, but it has at most $\tau\leq h/(32\log n)$ neighbors in $J_2$. Therefore, at least $h(1-\frac 1 {32\log n})\geq \frac{h}{2}$ neighbors of $v$ do not belong to $J_2$, and so they belong to the current layers. Since the number of layers cannot exceed $\log n$, there is some layer $L_i$, such that at least $h/(2\log n)$ neighbors of $v$ lie in $L_i$, so $v$ should not have been discarded.

Consider now the vertex set $J_1$ at any step of the algorithm, and let $v\in J_1$ be any vertex. For every consecutive pair $(L_i,L_{i+1})$ of layers, we direct all edges from $L_{i+1}$ to $L_{i}$. It is now enough to show that $v$ has at least $h/(2\log n)$ paths connecting it to vertices of $A$, that are completely disjoint except for sharing the endpoint $v$, in the resulting directed graph; all such paths are guaranteed to have length at most $\log n$. Assume otherwise. Then there is a set $R$ of fewer than $h/(2\log n)$ vertices, such that $v\not \in R$, and, if we delete the vertices of $R$ from the directed layered graph, then $v$ is disconnected from $A$. Assume w.l.o.g. that $v\in L_j$. Since $v$ has as least $h/(2\log n)$ neighbors in $L_{j-1}$ and $R<h/(2\log n)$, at least one neighbor $v_{j-1}\in L_{j-1}$ of $v$ does not lie in $R$. Using the same arguments, vertex $v_{j-1}$ must have at least one neighbor $v_{j-2}\in L_{j-2}$ that does not lie in $R$. We can continue like that until we reach $A$, obtaining a path connecting $v$ to a vertex of $A$, a contradiction. Therefore, there is a set $\pset(v)$ at least $h/(2\log n)$ paths in $H[A\cup J_1]$, connecting $v$ to vertices of $A$, such that the length of each path is at most $\log n$, and the paths in $\pset(v)$ are completely disjoint except for sharing their endpoint $v$. It now remains to prove Claim~\ref{claim: small number of layers}.


\begin{proofof}{Claim~\ref{claim: small number of layers}}
We fix some index $i>0$, and show that, throughout the algorithm, if $L_i\neq \emptyset$, then $|L_i|<|L_{i-1}|/2$. We let $E'$ be the set of all edges $e=(u,v)$, such that $u\in L_{i-1}$, $v\in L_i$, and either $u\in L_0$, or $v$ was added after $u$ to $J_1$.
Notice that, as the sets $L_{i-1}$ and $L_i$ change, the set $E'$ of edges evolves. 
 We claim that at every point of the algorithm, every vertex $u\in L_{i-1}$ is incident to at most $h/(32\log n)$ edges of $E'$, while every vertex $v\in L_i$ is incident to at least $h/(2\log n)$ such edges. 
 The latter is immediate to see: when we add $v$ to $L_i$, then it must have at least  $h/(2\log n)$ neighbors in $L_{i-1}$.
To see the former, recall that, when $u$ is added to $J_1$, it has fewer than $\tau\leq h/(32\log n)$ neighbors that belong to $J_2$. The edges of $E'$ may only  connect $u$ to vertices that were added to $J_1$ after $u$, and their number is bounded by $h/(32\log n)$.

Therefore, throughout the algorithm, the following two inequalities hold: (i) $|E'|\geq |L_i|\cdot\frac{h}{2\log n}$; and (ii) $|E'|< |L_{i-1}|\cdot \frac{h}{32\log n}$, and so $|L_{i}|< |L_{i-1}|/2$.
\end{proofof}
\end{proof}

The following corollary follows immediately from Lemma~\ref{lemma: DSP universal}
\begin{corollary}\label{cor: degree sep: after deleting small number}
Let $H$ be a simple graph containing at most $n$ vertices, and let $h$ be an integer, such that  the degree  of every vertex in $H$ is at least $h$. Let $(A,B)$ be any partition of vertices of $H$, and let $(J_1,J_2)$ be the output of $\DSP(H[B],\tau)$, for any $\tau\leq h/(32\log n)$. Then for any subset $R$ of fewer than $h/(2\log n)$ vertices of $H$, for every vertex $v\in J_1\setminus R$, there is a path in graph $H[A\cup J_1]\setminus R$, connecting $v$ to a vertex of $A$,  that contains at most $\log n$ edges.
\end{corollary}

%% file: new-proof.tex
\section{Decremental Single-Source Shortest Paths}\label{sec: main for SSSP}
This section is dedicated to the proof of Theorem~\ref{thm: main for SSSP}, with some details deferred to later sections.
As in much of previous work, we consider each distance scale separately. For each distance scale, we employ the following theorem.

\begin{theorem}\label{thm: main for SSSP w distance bound}
There is a randomized algorithm, that, given parameters $0<\eps<1$ and $D>0$ and a simple undirected $n$-vertex graph $G$ with lengths $\ell(e)>0$ on edges $e\in E(G)$, together with a special source vertex $s\in V(G)$, such that $G$ undergoes vertex deletions, supports queries $\pquery_D(v)$. For each query $\pquery_D(v)$, the algorithm returns a path from $s$ to $v$ in $G$, of length is at most $(1+\eps)\dist(s,v)$, if $D\leq \dist(s,v)\leq 2D$; otherwise, it either returns an arbitrary path connecting $s$ to $v$, or correctly establishes that $\dist(s,v)>2D$. The algorithm works against an {\bf adaptive} adversary. The total expected running time of the algorithm is $O\left (\frac{n^{2+o(1)}\cdot \log^3(1/\eps)}{\eps^2}\right )$, and each query is answered in expected time $O(n\poly\log n\log(1/\eps))$.
\end{theorem}

It is now easy to obtain Theorem~\ref{thm: main for SSSP} from Theorem~\ref{thm: main for SSSP w distance bound}. Let $\SSSP(G,s,D,\eps)$ be the data structure maintained by the algorithm from Theorem~\ref{thm: main for SSSP w distance bound} for graph $G$, source vertex $s$, and parameters $D$ and $\eps$ (we view the edge lengths as part of the definition of $G$). We assume w.l.o.g. that all edge lengths in $G$ are between $1$ and $L$. For $1\leq i\leq \floor{\log(Ln)}$, let $D_i=2^i$.
Given an input graph $G$, a source vertex $s$, and a parameter $\eps$, for each $1\leq i\leq \log(Ln)$, the algorithm maintains the data structure $\SSSP(G,s,D_i,\eps)$. When a vertex of $G$ is deleted, all these data structures are updated accordingly. The total expected update time for all these data structures is  $O\left (\frac{n^{2+o(1)}\cdot \log^3(1/\eps)\log L}{\eps^2}\right)$. In order to answer a query $\pquery(v)$, we run, for each $1\leq i\leq \log(Ln)$, the query $\pquery_{D_i}(v)$ in the corresponding data structure $\SSSP(G,s,D_i,\eps)$. We then return the shortest paths that was returned by any such query (if no path was returned by any query, then we report that $s$ is not connected to $v$ in $G$). It is easy to verify, from Theorem~\ref{thm: main for SSSP w distance bound}, that, if there is a path from $s$ to $v$ in $G$, then the above algorithm returns a path from $s$ to $v$ of length at most $(1+\eps)\dist(s,v)$. The expected running time required to process a query is $O(n\poly\log n\cdot\log(1/\eps)\cdot \log L)$. In order to complete the proof of Theorem~\ref{thm: main for SSSP}, it is now enough to prove Theorem~\ref{thm: main for SSSP w distance bound}. From now on we focus on the proof of this theorem.

Throughout the proof, we denote by $G$ the current graph, obtained from the input graph after the current sequence of vertex deletions, and $n$ is the number of vertices present in $G$ at the beginning of the algorithm. When we say that an event holds with high probability, we mean that the probability of the event is at least $(1-1/n^c)$ for some large enough constant $c$. We will assume throughout the proof that $n>c_0$ for some large enough constant $c_0$, since otherwise we can simply run Dijkstra's shortest path algorithm in $G$ to respond to the path queries.
We assume that the distance bound $D$ is fixed from now on. It would be convenient for us to ensure that $D= \ceil{4n/\eps}$, and that all edge lengths are integers between $1$ and $4D$. In order to achieve this, we discard all edges whose length is greater than $2D$, and we change the length of each remaining edge $e$ to be $\ell'(e)=\ceil{4n\ell(e)/(\eps D)}$. For every pair $u,v$ of vertices, let $\dist'(u,v)$ denote the  distance between $u$ and $v$ with respect to the new edge length values. Notice that for all $u,v$:

\[\frac{4n}{\eps D}\dist(u,v)\leq \dist'(u,v)\leq \frac{4n}{\eps D}\dist(u,v)+n,\]

since the shortest $s$--$v$ path contains at most $n$ vertices. Moreover, if $\dist(u,v)\geq D$, then 
$n\leq \dist(u,v)\cdot \frac n D$, so $\dist'(u,v)\leq \frac{4n}{\eps D}\dist(u,v)+\frac n D\dist(u,v)\leq \frac{4n}{\eps D}\dist(u,v)(1+\eps/4)$.
 Notice  also that for every vertex $v$ with $D\leq \dist(u,v)\leq 2D$, $\ceil{\frac{4n}{\eps}}\leq \dist'(u,v)\leq 4\ceil{\frac{4n}{\eps}}$.
Therefore, from now on we can assume that $D= \ceil{4n/\eps}$, and for simplicity, we will denote the new edge lengths by $\ell(e)$ and the corresponding distances between vertices by $\dist(u,v)$. 
From the above discussion,  all edge lengths are integers between $1$ and $4D$. It is now sufficient that the algorithm, given query $\pquery_D(v)$, returns a path from $s$ to $v$ in $G$, of length is at most $(1+\eps)\dist(s,v)$, if $D\leq \dist(s,v)\leq 4D$; otherwise, it can either return an arbitrary path connecting $s$ to $v$, or correctly establish that $\dist(s,v)>4D$.

At a very high level, our proof follows the algorithm of~\cite{Bernstein}. We partition all edges of $G$ into $\lambda=\floor{\log (4D)}$ classes, where for $1\leq i\leq\lambda$, edge $e$ belongs to \emph{class $i$} iff $2^i\leq \ell(e)<2^{i+1}$. We denote the set of all edges of $G$ that belong to class $i$ by $E_i$. Next, for each  $1\leq i\leq \lambda$, we define a threshold value $\tau_i$. For technical reasons, these values are somewhat different from those used in~\cite{Bernstein}. In order to define $\tau_i$, we need to introduce a number of parameters that we will use throughout the algorithm

\paragraph{Parameters.}
The following parameters will be used throughout the algorithm.

\begin{itemize}
\item We let $\alpha^*=1/2^{3\sqrt{\log n}}$ -- this will be the expansion parameter for the cores.
\item We let $\ell^*=\frac{16c^*\log^{12}n}{\alpha^*}$, for some large enough constant $c^*$ that we set later. This parameter will serve as an upper bound on the lengths of paths between pairs of vertices in a core. Observe that $\ell^*=2^{O(\sqrt{\log n})}$.
\item Our third main parameter is $\Delta=256c^*\log^{20}n/\alpha^*=2^{O(\sqrt{\log n})}=n^{o(1)}$. This parameter will be used in order to partition the algorithm into phases. 
\item For each  $1\leq i\leq \lambda$,  we set $\tau_i=\max\set{4n^{2/\log\log n},\frac{n}{\eps D}\cdot 2^{21} \cdot \ell^*\cdot \Delta\cdot \log^4 n\cdot \lambda\cdot 2^i}$. Notice that $\tau_i=\max\set{n^{o(1)},\frac{n^{1+o(1)}\cdot 2^i\cdot \log D}{\eps D}}$.
\end{itemize}

 Bernstein~\cite{Bernstein} used the threshold values $\tau_i$ in order to partition the edges of $G$ into two subsets, which are then used to define two graphs: a {\em light graph} and a {\em heavy graph}. We proceed somewhat differently. First, it would be more convenient for us to define a separate heavy graph for each edge class, though we still keep a single light graph. Second, our process of partitioning the edges between the heavy graphs and the light graph is somewhat different from that in~\cite{Bernstein}. However, we still ensure that for each $1\leq i\leq \lambda$, the light graph contains at most $n\tau_i$ edges of $E_i$ throughout the algorithm; this is a key property of the light graph that the algorithm of~\cite{Bernstein} exploits.

Fix an index $1\leq i\leq \lambda$, and let $G_i$ be the sub-graph of $G$ induced by the edges in $E_i$. We run Procedure \DSP on graph $G_i$ and degree threshold $d=\tau_i$. Recall that the procedure computes a partition $(J',J'')$ of $V(G_i)$, by starting with $J'=\emptyset$ and $J''=V(G_i)$, and then iteratively moving from $J''$ to $J'$ vertices $v$ whose degree in $G_i[J'']$ is less than $d$. The procedure can be implemented to run in time $O(|E_i|+n)$. We say that the vertices of $J'$ are \emph{light for class $i$}, and the vertices of $J''$ are \emph{heavy for class $i$}. We now define the graph $G_i^H$ -- the heavy graph for class $i$, as follows. The set of vertices of $G_i^H$ contains all vertices that
are heavy for class $i$. The set of edges contains all edges of $E_i$ whose {\bf both} endpoints are heavy for class $i$. We also define a light graph $G_i^{L}$ for class $i$, though we will not use it directly. Its vertex set is $V(G)$, and the set of edges contains all edges of $E_i$ that do not belong to graph $G_i^H$. Clearly, every edge of $G_i^{L}$ is incident to at least one vertex that is light for class $i$, and it is easy to verify that $|E(G_i^{L})|\leq n\tau_i$. As the algorithm progresses and vertices are deleted from $G$, some vertices that are heavy for class $i$ may become light for it (this happens when a vertex $v$ that was heavy for class $i$ has fewer than $\tau_i$ neighbors that are also heavy for class $i$). Once a vertex $v$ becomes light for class $i$, every edge in $G_i^H$ that is incident to $v$ is removed from $G_i^H$ and added to $G_i^{L}$, and $v$ is deleted from $G_i^H$. In particular, $E(G_i^H)$ and $E(G_i^{L})$ always define a partition of the current set $E_i$ of edges. Moreover, it is easy to verify that the total number of edges that are ever present in $G_{i}^{L}$ is bounded by $n\tau_i$, and that, throughout the algorithm, every vertex of $G_i^H$ has degree at least $\tau_i$ in $G_i^H$. The main technical contribution of this paper is the next theorem, that allows us to deal with the heavy graphs.

\begin{theorem}\label{thm: main for maintaining a heavy graph}
There is a randomized algorithm, that, for a fixed index $1\leq i\leq \lambda$, supports queries $\pquery(u,v,C)$: given two vertices $u$ and $v$ that belong to the same connected component $C$ of graph $G_i^H$, returns a path, connecting $u$ to $v$ in $C$, that contains at most $2^{13}\frac{|V(C)|}{\tau_i}\cdot \Delta \cdot \ell^*\cdot \log^4 n$ edges. The total expected update time of the algorithm is $O(n^{2+o(1)})$, and each query $\pquery(u,v,C)$ is processed in expected time $O(|V(C)|\log^4n)$. The algorithm works against an adaptive adversary.
\end{theorem}

Next, we define an extended light graph, and provide an algorithm for handling it, which is mostly identical to the algorithm of~\cite{Bernstein}. Our starting point is the graph $G^L$, that we refer to as the \emph{light graph}, with $V(G^L)=V(G)$, and $E(G^L)=\bigcup_{i=1}^{\lambda}E(G_i^L)$. We now define the {\em extended light graph} $\hat G^L$, as follows. We start with $\hat G^L=G^L$; the vertices of $G^L$ are called \emph{regular vertices} and the edges of $G^L$ are called \emph{regular edges}. Next, for every $1\leq i\leq\lambda$, for every connected component $C$ of $G_i^H$, we add a vertex $v_C$ to $\hat G^L$, that we call a \emph{special vertex}, and connect it to every vertex $u\in C$ with an edge of length $1/4$. The edges incident to the special vertices are called \emph{special edges}. As the algorithm progresses and edges and vertices are deleted or moved from the heavy graphs $G^H_i$ to the light graphs $G^L_i$, the connected components of the graphs $G^H_i$ may change. We will always keep the graph $\hat G^L$ updated with respect to the current connected components of the graphs $G^H_i$, and with respect to the edges currently in the light graphs $G^L_i$. The following observation follows immediately from the assumption that all edge lengths in $G$ are at least $1$.

\begin{observation}\label{obs: distances in G and light graph}
Throughout the algorithm, for every vertex $v\in V(G)$, $\dist_{\hat G^L}(s,v)\leq \dist_G(s,v)$.
\end{observation}

The theorem below follows from the same arguments as in~\cite{Bernstein}. However, since our setting is slightly different, we provide its proof in the Appendix for completeness.
For convenience, we denote by $E^S$ and by $E^R$ the sets of all special and all regular edges, respectively, that are ever present in graph $\hat G^L$.

\begin{theorem}\label{thm: main for maintaining a light graph}
There is a deterministic algorithm, that maintains an (approximate) single-source shortest-path tree $T$ of $\hat G^L$ from the source $s$, up to distance $8D$. Tree $T$ is a sub-graph of $\hat G^L$, and for every vertex $v\in V(\hat G^L)$, such that $\dist_{\hat G^L}(s,v)\leq 8D$, the distance from $s$ to $v$ in $T$ is at most $(1+\eps/4)\dist_{\hat G^L}(s,v)$. The total update time of the algorithm is $O(n^2\log D\log^2n)+O\left(\frac{nD\log^2 n\log^2D}{\eps}\right)+O\left (\sum_{e\in E^R}\frac{D\log n\log D}{\eps \ell(e)}\right )$.
\end{theorem}

We first bound the contribution of the term $O\left (\sum_{e\in E^R}\frac{D\log n\log D}{\eps \ell(e)}\right )$ to the algorithm's running time. For all $1\leq i\leq \lambda$, we let $E'_i=E^R\cap E_i$ be the set of all regular edges of class $i$ that are ever present in $\hat G^L$. Recall that, from the definition of light and heavy vertices for class $i$, $|E'_i|\leq n \tau_i\leq \max\set{n^{1+o(1)},\frac{n^{2+o(1)} \cdot 2^i\cdot \log D}{\eps D}}$, since:

 $$\tau_i=\max\set{4n^{2/\log\log n},\frac{n}{\eps D}\cdot 2^{21} \cdot \ell^*\cdot \Delta\cdot \log^4 n\cdot \lambda\cdot 2^i}=\max\set{n^{o(1)},\frac{n^{1+o(1)}\cdot 2^i\cdot \log D}{\eps D}}.$$
 
  Therefore, the total contribution of all regular edges of class $i$ to the running time is at most:

\[O\left(\frac{D|E_i'|\log n\log D}{\eps 2^i}\right )\leq O\left (\frac{n^{2+o(1)} \log^2 D}{\eps^2}\right )+O\left(\frac{n^{1+o(1)}D\log D}{\eps}\right ).\]

Overall, the total contribution of all regular edges from all $\lambda=O(\log D)$ classes to the running time of the algorithm, is $O\left (\frac{n^{2+o(1)} \log^3 D}{\eps^2}\right )+O\left(\frac{n^{1+o(1)}D\log^2 D}{\eps}\right )$,

By combining the algorithms from Theorem~\ref{thm: main for maintaining a heavy graph} for all $1\leq i\leq \lambda$, and Theorem~\ref{thm: main for maintaining a light graph}, and recalling that $D\leq \ceil{4n/\eps}$, we get that the total expected update time of our algorithm is $O\left(\frac{n^{2+o(1)}\cdot \log^3(1/\eps)}{\eps^2}\right)$; this includes time $O(|E(G)|+n)$ that is needed to partition the edges into classes and to construct the heavy and the light graphs for each class. Next, we describe how our algorithm responds to queries $\pquery_D(v)$.

Notice that, given a vertex $v\in V(G)$ with $\dist_G(s,v)\leq 4D$, we are guaranteed that $\dist_{\hat G^L}(s,v)\leq \dist_G(s,v)$ from Observation~\ref{obs: distances in G and light graph}, and we can use the tree $T$ (given by Theorem~\ref{thm: main for maintaining a light graph}) in order to find a simple path $P$ in $\hat G^L$ from $s$ to $v$, whose length is at most $(1+\eps/4)\dist_{\hat G^L}(s,v)\leq (1+\eps/4)\dist_G(s,v)$. This path can be found in time $O(n)$ by simply following the edges of the tree $T$ from $v$ to $s$. Next, we show how to transform the path $P$ into a path $P^*$ in the original graph $G$, connecting $s$ to $v$, such that the length of $P^*$ is at most $(1+\eps/2)\dist_G(s,v)$, provided that $D\leq \dist_G(s,v)\leq 4D$.

Let $v_{C_1},\ldots,v_{C_r}$ be all special vertices that appear on the path $P$. For $1\leq j\leq r$, let $u_j$ be the regular vertex preceding $v_{C_j}$ on $P$, and let $u'_j$ be the regular vertex following $v_{C_j}$ on $P$. Denote $n_j=|V(C_j)|$, and assume that $C_j$ is a connected component of graph $G^H_{i_j}$. For each $1\leq j\leq r$, we run the query $\pquery(u_j,u'_j,C_j)$ in the data structure from Theorem~\ref{thm: main for maintaining a heavy graph} for graph $G^H_{i_j}$, to obtain a path $Q_j$, connecting $u_j$ to $u'_j$ in $C_j$, that contains at most $2^{13}\frac{n_j}{\tau_{i_j}}\cdot \Delta\cdot \ell^*\cdot \log^4n$ edges, in expected time $O(n_j\log^4n)$. Since path $Q_j$ only contains edges of $E_{i_j}$, its total length is at most:

 \[\begin{split}
 2^{13}\frac{n_j}{\tau_{i_j}}\cdot \Delta\cdot \ell^*\cdot \log^4 n\cdot 2^{i_j+1}&\leq 2^{13}\frac{n_j}{(2^{21}n\cdot 2^{i_j}\cdot \Delta\cdot \ell^*\cdot \log^4 n\cdot \lambda)/(D\eps)}\cdot \Delta\cdot \ell^*\cdot \log^4 n\cdot 2^{i_j+1}\\
 &\leq \frac{\eps D n_j }{8n\lambda}.
 \end{split}
 \]

We obtain the final path $P^*$, by replacing, for each $1\leq j\leq r$, the vertex $v_{C_j}$ of $P$ with the path $Q_j$ (excluding its endpoints).
For each $1\leq i\leq \lambda$, let $\cset_i$ be the set of all connected components $C_j$ of $G^H_i$ with $v_{C_j}$ lying on the path $P$.
Let $\qset_i$ be the set of all paths $Q_j$ with $C_j\in \cset_i$. The total length of all paths in $\qset_i$ is then bounded by $\sum_{C_j\in \cset_i}\frac{\eps D n_j}{8n\lambda }\leq \frac{\eps D}{8\lambda}$ (since the path $P$ is simple, and so all connected components $C_j$ of $\cset_i$ are distinct). The total length of all paths $Q_1,\ldots,Q_r$ is then bounded by $\eps D/8$.
Since we have assumed that $D\leq \dist_G(s,v)\leq 4D$, we get that the total length of the paths $Q_1,\ldots,Q_r$ is at most $\eps \dist_G(s,v)/4$, and the final length of the path $P^*$ is therefore bounded by $(1+\eps/2)\dist_G(s,v)$.
We now analyze the total expected running time required to respond to the query. As observed before, path $P$ can be computed in time $O(n)$, and each query $\pquery(u_j,u'_j,C_j)$ is processed in expected time $O(n_j\log^4n)$.
 Since all components $C_j\in \cset_i$ are disjoint, $\sum_{C_j\in \cset_i}n_j\leq n$, and so the total expected time to process such a query is $O(n\log^4n\log D)=O(n\poly\log n \log(1/\eps))$.

It now remains to prove Theorems~\ref{thm: main for maintaining a heavy graph} and \ref{thm: main for maintaining a light graph}. The remainder of this section and the subsequent section are dedicated to the proof of Theorem~\ref{thm: main for maintaining a heavy graph}. 

In order to simplify the notation, we will denote the graph $G^H_i$ by $G^*$ from now on. We will use $n$ for the number of vertices of the original graph $G$ throughout the proof. We also denote $\tau=\tau_i$, and we will use the fact that $\tau\geq 4n^{2/\log\log n}$, and that every vertex in $G^*$ has degree at least $\tau$. The central notions that we use in our proof are those of a \emph{core structure} and a \emph{core decomposition}. Our algorithm will break the graph $G^*$ into sub-graphs and will compute a core decomposition in each such subgraph. In every subgraph that we will consider, the degrees of all vertices are at least $ n^{1/\log\log n}$. In the following subsections we define core structures and a core decomposition and develop the technical machinery that we need to construct and maintain them.

\label{--------------------------------------subsec: core def------------------------------------------}
\subsection{Core Structures and Cores}
In this subsection, we define core structures and cores, that play a central role in our algorithm. We also establish some of their properties, and provide an algorithm that computes short paths between a given pair of vertices of a core.

Throughout this subsection, we will assume that we are given some graph, that we denote by $\hG$, that is a subgraph of our original $n$-vertex graph $G$. Therefore, throughout this subsection, for every graph $\hG$ that we consider, we assume that $|V(\hG)|\leq n$. We also assume that we are given a parameter $h>n^{1/\log\log n}$, and that every vertex in $\hG$ has degree at least $h$. 

A central notion that we use is that of a core and a core structure. 
Recall that we have defined two parameters: $\alpha^*=1/2^{3\sqrt{\log n}}$ and $\ell^*=\frac{16c^*\log^{12}n}{\alpha^*}$, for some large enough constant $c^*$ that we set later. Observe that $\ell^*=2^{O(\sqrt{\log n})}$.

\begin{definition}
Given a graph $\hG$ with $|V(\hat G)|\leq n$, a \emph{core structure} $\kset$ in $\hG$ consists of:

\begin{itemize}
\item two disjoint vertex sets: a set $K\neq \emptyset$ of vertices, that we refer to as the core itself, and a set $U(K)$ of at most $|K|$ vertices, that we call the \emph{extension} of the core $K$;

\item a connected subgraph $\hG^K\subseteq \hG[K\cup U(K)]$, with $V(\hG^K)=K\cup U(K)$; 

\item a graph $W^K$, that we refer to as the \emph{witness graph} for $K$, with 
$K\subseteq V(W^K)\subseteq K\cup U(K)$, such that  the maximum vertex degree of $W^K$ is at most $\log^3n$; and

\item for every edge $e=(x,y)\in E(W^K)$, a path $P(e)$ in $\hG^K$, that connects $x$ to $y$, such that:

\begin{itemize}
\item every path in set $\set{P(e)\mid e\in E(W^K)}$ contains at most $c^*\log^8n$ vertices (here $c^*$ is the constant that appears in the definition of $\ell^*$); and
\item every vertex of $\hG^K$ participates in at most $c^*\log^{19}n$ paths of $\set{P(e)\mid e\in E(W^K)}$.
\end{itemize}
\end{itemize}
If, additionally, $W^K$ is an $\alpha^*$-expander, then we say that $\kset=(K,U(K),\hG^K,W^K)$ is a \emph{perfect core structure}.
\end{definition}

We call the set $\Psi(W^K)=\set{P(e)\mid e\in E(W^K)}$ of paths the \emph{embedding of $W^K$ into $\hG^K$}, and we view this embedding as part of the witness graph $W^K$.

\begin{definition}
We say that a core structure $\kset=(K,U(K),\hG^K,W^K)$ is an $h$-core structure iff  for every vertex $v\in K$ of the core, there is a set $N(v)\subseteq V(W^K)$ of at least $h/(64\log n)$ vertices, such that for every vertex $u\in N(v)$, the edge $(u,v)$ lies in $G^K$. A perfect core structure with this property is called a \emph{perfect $h$-core structure}.
\end{definition}

The following observation follows immediately from the definition of a perfect $h$-core structure.

\begin{observation}\label{obs: size of h-core}
Let $\kset=(K,U(K),\hG^K,W^K)$ be an $h$-core structure in a graph $\hat G$ with $|V(\hat G)|\leq n$. Then $|K|\geq h/(128\log n)$.
\end{observation}
\begin{proof}
By the definition of a core structure, $K\neq \emptyset$. In particular, $K$ contains at least one vertex, that we denote by $v$. From the definition of an  $h$-core structure, at least $h/(64\log n)$ neighbors of $v$ in $G^K$ belong to $K\cup U(K)$, so $|K\cup U(K)|\geq h/(64\log n)$. Since $|K|\geq |U(K)|$ from the definition of a core structure, we get that  $|K|\geq h/(128\log n)$.
\end{proof}

Recall that we have defined a parameter $\Delta=256c^*\log^{20}n/\alpha^*=2^{O(\sqrt{\log n})}$. Since $h\geq n^{1/\log\log n}$, and $n$ is large enough, we can assume that $h>\Delta^2$.
 We will show that the following property holds for every perfect $h$-core structure:

\begin{properties}{P}
\item Let $\kset=(K,U(K),\hG^K,W^K)$ be a perfect $h$-core  structure in $\hG$. Then for every subset $R$ of at most $h/\Delta$ vertices of $\hG$, for every pair $u,v\in K\setminus R$ of vertices in the core, graph $\hG^K\setminus R$ contains a path from $u$ to $v$ of length at most $\ell^*$.
\end{properties}

In fact, we prove a stronger result: we provide an algorithm, that, given a graph $\hG$, undergoing deletions of up to $h/\Delta$ vertices, and a perfect $h$-core  structure $\kset=(K,U(K),\hG^K,W^K)$ in $\hG$, supports queries $\corepath(u,v)$: given a pair $u,v\in K$ of core vertices that were not deleted yet, return  a path of length at most $\ell^*$, connecting $u$ to $v$ in the current graph $\hG^K$. If the input core  structure $\kset$ is not a perfect core  structure, then the algorithm will either return a path as required, or it will determine that $\kset$ is not a perfect core  structure.

\begin{theorem}\label{thm: maintaining a core}
There is a deterministic algorithm, that, given a graph $\hG$ with $|V(\hG)|\leq n$, undergoing at most $h/\Delta$ vertex deletions, and an $h$-core  structure $\kset=(K,U(K),\hG^K,W^K)$ in $\hG$, supports queries $\corepath(u,v)$. Given such a query, the algorithm either returns a path from $u$ to $v$ in the current graph $\hG^K$, of length at most $\ell^*$, or correctly determines that $\kset$ is not a perfect core  structure, (that is, $W^K$ is not an $\alpha^*$-expander). The total running time of the algorithm is $O(|E(\hG^K)|\poly\log n)$, and the total running time to process each query is $O(\ell^*+|K|\log^3n)$.
\end{theorem}

We emphasize that, if the core  structure $\kset$ that serves as input to Theorem~\ref{thm: maintaining a core} (that is, before any vertices were deleted) is a perfect $h$-core  structure, then the algorithm is guaranteed to return a path from $u$ to $v$ of length at most $\ell^*$ in the current graph $\hat G^K$.

\begin{proof}
Let $\kset=(K,U(K),\hG^K,W^K)$ be the given $h$-core  structure.  Recall that every vertex of $\hG^K$ participates in at most $O(\log^{19}n)$ paths in the set $\set{P(e)\mid e\in E(W^K)}$, and the length of each such path is $O(\log^8n)$.

We maintain the graphs $\hG^K$ and $W^K$, stored as adjacency lists. Additionally, every vertex $v\in V(\hG^K)$ stores a pointer to every edge $e\in E(W^K)$, such that $v\in P(e)$. Initializing these data structures takes time $O(|E(\hG^K)|)+O(|E(W^K)|\log^8n)+O(|V(\hG^K)|)=O(|E(\hG^K)|)+O(|E(W^K)|\log^8n)=O(|E(\hG^K)|\poly\log n)$, as the length of each path $P(e)$ is $O(\log^8n)$, graph $\hat G^K$ is connected, and every vertex of $W^K$ has degree at most $\log^3n$.
We now show how to handle vertex deletions and how to respond to queries $\corepath(u,v)$.

\paragraph{Vertex Deletions.}
When a vertex $v$ is deleted from $\hG^K$, we delete from $W^K$ the vertex itself, and also every edge $e$ with $v\in P(e)$. Since $v$ stores pointers to each such edge, and there are at most $O( \log^{19}n)$ paths $P(e)$ in which $v$ participates, and since at most $h/\Delta$ vertices are deleted, this takes total time $O\left (\frac h \Delta \cdot \log^{19}n\right )=O\left (\frac{\alpha^* h}{256c^*\log^{20} n}\cdot \log^{19}n\right )=O(\alpha^*|K|\poly\log n)$ (we have used here the fact that, from Observation~\ref{obs: size of h-core}, $|K|\geq \Omega(h/\log n)$). Therefore, the total time for initializing and maintaining the data structure is at most $O(|E(\hG^K)|\poly\log n+\alpha^*|K|\poly\log n)=O(|E(\hG^K)|\poly\log n)$.

\paragraph{Responding to queries.} We now show an algorithm for responding to a query $\corepath(u,v)$, where $u,v\in K$ is a pair of core vertices that were not deleted yet. Recall that, from the definition of an $h$-core  structure, vertex $u$ originally had at least $h/(64\log n)$ neighbors in $\hG^K$ that lie in $V(W^K)$. Let $S_0$ denote the subset of these neighbors that have not yet been deleted. Since at most $h/\Delta\leq h/(256\log n)$ vertices can be deleted, $|S_0|\geq h/(128 \log n)$. Similarly, vertex $v$ originally had at least $h/(64\log n)$ neighbors in $\hG^K$ that lie in $V(W^K)$; we denote by $T_0$ the subset of these neighbors that were not deleted yet, so $|T_0|\geq h/(128\log n)$.

 Let $\tilde W$ denote the current graph $W^K$, after the deletions of vertices from $\hG^K$ made by the algorithm, and the corresponding deletions of vertices and edges from $W^K$.
Intuitively, we perform a BFS search in $\tilde W$ from $S_0$ and from $T_0$, until the two searches meet, and then map the resulting path to the graph $\hG^K$, using the embedding $\Psi(W^K)$.

Specifically, for $i\geq 0$, while $|S_i\cap T_i|=\emptyset$, we let $S_{i+1}$ contain all vertices of $S_i$ and all their neighbors in $\tilde W$. We define $T_{i+1}$ similarly for $T_i$. 
At the end of this algorithm, we obtain a path $Q$ connecting a vertex $u_0\in S_0$ to a vertex $v_0\in T_0$ in $\tilde W$, if such a path exists. 
Denote this path by $(u_0,u_1,\ldots,u_r=v_0)$. 
First, if $r>8\log^4n/\alpha^*$, then we terminate the algorithm and report that $\kset$ is not a perfect $h$-core  structure. Otherwise, we output the final path $Q^*$, connecting $u$ to $v$ in the current graph $\hG^K$, by first replacing every edge $e\in E(Q)$ with its embedding $P(e)$, and then adding the edges $(u,u_0)$ and $(v_0,v)$ to the resulting path. Recall that the length of each path $P(e)$ in the embedding $\Psi(G^K)$ is at most $c^*\log^8n$, so the length of the final path is at most $8c^* \log^{12}n/\alpha^*+2\leq \ell^*$, as required.
Notice that the algorithm for processing a single query consists of two BFS searches in $W^K$, that take time $O(|E(W^K)|)=O(|K|\log^3n)$, and additional time of at most $O(\ell^*)$ to compute the final path, giving the total running time of $O(\ell^*+|K|\log^3n)$ as required.
In order to complete the analysis of the algorithm, it is enough to prove the following lemma.

\begin{lemma}\label{lemma: bound r}
If the core  structure $\kset$ is a perfect core  structure, then $r\leq 8\log^4n/\alpha^*$.
\end{lemma}
\begin{proof}
The proof uses the following claim.
\begin{claim}\label{claim: BFS grows fast}
Assume that $W^K$ is an $\alpha^*$-expander. Then
for all $i\geq 0$ with $|S_i|\leq |V(W^K)|/2$,  $|S_{i+1}|\geq |S_i|\left (1+\frac{\alpha^*}{2\log^3n}\right )$. Similarly, if $|T_i|\leq |V(W^K)|/2$, then $|T_{i+1}|\geq |T_i|\left (1+\frac{\alpha^*}{2\log^3n}\right )$. 
\end{claim}

Notice that, if the claim is true, then $r\leq 8\log^4n/\alpha^*$ must hold, as otherwise, both $S_{\floor{4\log^4n/\alpha^*}}$ and $T_{\floor{4\log^4n/\alpha^*}}$ must contain more than half the vertices of $W^K$, and hence they must intersect.

\begin{proofof}{Claim~\ref{claim: BFS grows fast}}
Consider some index $i\geq 0$, such that $|S_i|\leq |V(W^K)|/2$. Since $W^K$ is an $\alpha^*$-expander, $|\out_{W^K}(S_i)|\geq \alpha^*|S_i|$. 
In particular, $|\out_{W^K}(S_i)|\geq \alpha^*|S_0|\geq \alpha^*h/(128\log n)$.

Since at most $h/\Delta=\alpha^*h/(256c^*\log^{20}n)$ vertices of $\hG$ were deleted, and, for every vertex $v$ of $\hG$ there are at most $c^*\log^{19}n$ edges $e\in E(W^K)$ with $v\in P(e)$, at most $\frac{\alpha^*h}{256c^*\log^{20} n}\cdot c^*\log^{19}n=\frac{\alpha^*h}{256\log n}$ edges of $W^K$ were deleted, and so $|\out_{\tilde W}(S_i)|\geq |\out_{W^K}(S_i)|/2\geq \alpha^*|S_i|/2$.

Let $N_i$ denote the set of all vertices of $\tilde W$ that do not belong to $S_i$ but serve as endpoints of the edges in $\out_{\tilde W}(S_i)$. As the degree of every vertex in $W^K$ is at most $\log^3n$, $|N_i|\geq \frac{|\out_{\tilde W}(S_i)|}{\log^3n}\geq  \frac{\alpha^*|S_i|}{2\log^3n}$. 

We conclude that $|S_{i+1}|\geq |S_i|\left (1+\frac{\alpha^*}{2\log^3n}\right )$. The proof for $T_i$ is similar.
\end{proofof}
\end{proof}
\end{proof}

\label{--------------------------------------subsec: core decomposition def------------------------------------------}
\subsection{Core Decomposition}

In addition to core  structures, our second main tool is a core decomposition. In this subsection we define core decomposition and we state a theorem that allows us to compute them.
Before we define a core decomposition, we need to define an $h$-universal set of vertices.

\begin{definition}
Suppose we are given a subgraph $\hat G\subseteq G$, and a set $S$ of its vertices. Let $J$ be another subset of vertices of $\hat G$. We say that $J$ is an \emph{$h$-universal set with respect to $S$} iff for every vertex $u\in J$ and for every subset $R$ of at most $h/\Delta$ vertices of $\hat G\setminus\set{u}$, there is a path in $\hat G[S\cup J]\setminus R$, connecting $u$ to a vertex of $S$, whose length is at most $\log n$.
\end{definition}



Finally, we are ready to define a core decomposition.

\begin{definition}
An \emph{$h$-core decomposition} of a graph $\hG$  with $|V(\hG)|\leq n$ is a collection

 $$\fset=\set{(K_i,U(K_i),\hG^{K_i},W^{K_i})}_{i=1}^r$$
 
  of $h$-core  structures in $\hG$, such that $K_1,\ldots,K_r$ are mutually disjoint (but a vertex $v\in V(\hG)$ may belong to a number of extension sets $U(K_i)$, in addition to belonging to some core $K_j$), and every edge of $\hG$ participates in at most $\log n$ graphs $\hG^{K_1},\ldots,\hG^{K_r}$. Additionally, if we denote $\tilde K=\bigcup_{i=1}^rK_i$ and $J=V(\hG)\setminus \tilde K$, then set $J$ is $h$-universal with respect to 
$\tilde K$. We say that this decomposition is a \emph{perfect $h$-core decomposition} iff every core structure in $\fset$ is a perfect $h$-core  structure.
\end{definition}

The main building block of our algorithm is the following theorem, whose proof is deferred to Section~\ref{sec: computing core decomp}.

\begin{theorem}\label{thm: find a core decomposition}
There is a randomized algorithm, that, given a sub-graph $\hG\subseteq G$ and a parameter $h\geq n^{1/\log\log n}$, such that every vertex of $\hG$ has degree at least $h$ in $\hG$, computes an $h$-core decomposition of $\hG$. Moreover, with high probability, the resulting core decomposition is perfect. The running time of the algorithm is $O((|E(\hG)|+|V(\hG)|^{1+o(1)})\poly\log n)$.
\end{theorem}


\label{--------------------------------------subsec: rest of proof of thm------------------------------------------}
\subsection{Completing the Proof of Theorem~\ref{thm: main for maintaining a heavy graph}} 
We use the parameter $\Delta$ defined in previous subsections; recall that $\Delta=256c^*\log^{20}n/\alpha^*=2^{O(\sqrt{\log n})}=n^{o(1)}$.
We start with a high-level intuition to motivate our next steps.
Consider the graph $G^*=G_i^H$, and let $d$ be its average vertex degree. For simplicity, assume that $d=\Delta^j$ for some integer $j$. Let us additionally assume, for now, that the degree of every vertex in $G^*$ is at least $h=\Delta^{j-1}$ (this may not be true in general). We can then compute an $h$-core decomposition $\fset$ of $G^*$ using Theorem~\ref{thm: find a core decomposition}. Note that, as long as we delete fewer than $h/\Delta$ vertices from $G^*$, the current core decomposition remains functional, in the following sense: for every core structure $\kset=(K,U(K),(G^*)^K,W^K)\in \fset$, for every pair $u,v\in K$ of vertices in the core that were not deleted yet, we can use Theorem~\ref{thm: maintaining a core} to compute a path of length at most $\ell^*$ between $u$ and $v$; and for every vertex $w$ of $G^*$ that does not lie in any core $K$, there is a path of length at most $\log n$ connecting it to some core, from the definition of the $h$-universal set. Both these properties are exploited by our algorithm in order to respond to queries $\pquery$. Note that computing the core decomposition takes time $O((|E(G^*)|+n^{1+o(1)})\poly\log n)=O(n^{1+o(1)}\Delta^j)$, and the total time required to maintain the data structures from Theorem~\ref{thm: maintaining a core} for every core is also bounded by this amount, since every edge of $G^*$ may belong to at most $\log n$ graphs $(G^*)^K$, where $K$ is a core in the decomposition. We can partition the algorithm into phases, where in every phase, $h/\Delta=\Delta^{j-2}$ vertices are deleted from $G^*$. Once a phase ends, we recompute the core decomposition. Since the number of phases is bounded by $n/\Delta^{j-2}$, and the total running time within each phase is $O(n^{1+o(1)}\Delta^j)$, the overall running time of the algorithm would be at most $n^{2+o(1)}$, as required. The main difficulty with this approach is that some vertices of $G^*$ may have degrees that are much smaller than the average vertex degree. Even though we could still compute the core decomposition, we are only guaranteed that it remains functional for a much smaller number of iterations -- the number that is close to the smallest vertex degree in $G^*$. We would then need to recompute the core decomposition too often, resulting in a high running time.

In order to get around this difficulty, we partition the vertices of $G^*$ into ``layers''. Let $z_1$ be the smallest integer, such that the maximum vertex degree in $G^*$ is less than $\Delta^{z_1}$, and let $z_2$ be the largest integer, such that $\Delta^{z_2}<\tau/(64\log n)$. Let $r=z_1-z_2$, so $r\leq \log n$. We emphasize that the values $z_1,z_2$ and $r$ are only computed once at the beginning of the algorithm and do not change as vertices are deleted from $G^*$. We will split the graph $G^*$ into $r$ layers, by defining sub-graphs $\tilde \Lambda_1,\ldots,\tilde \Lambda_r$ of $G^*$, that are disjoint in their vertices. For each $0\leq j\leq r$, we use a parameter $h_j=\Delta^{z_1-j}$, so that $h_0=\Delta^{z_1}$ upper-bounds the maximum vertex degree in $G^*$, $h_r=\Delta^{z_2}<\tau/(64\log n)$, and for $1<j\leq r$, $h_{j}=h_{j-1}/\Delta$. We will ensure that for each $1\leq j\leq r$, graph $\tilde \Lambda_j$ contains at most $nh_{j-1}$ edges, and that every vertex in $\tilde \Lambda_j$ has degree at least $h_j$. Additionally, for each $1< j\leq r$, we will define a set $D_j$ of \emph{discarded vertices}: intuitively, these are vertices $v$, such that $v$ does not belong to layers $1,\ldots,j-1$, but almost all neighbors of $v$ do. We need to remove these vertices since otherwise the average vertex degree in subsequent layers may fall below $\tau$, even while some high-degree vertices may still remain. For each $1\leq j\leq r$, we also define a graph $\Lambda_j$, which is the sub-graph of $G^*$ induced by all vertices of $\tilde \Lambda_j,\ldots,\tilde \Lambda_r$ and of $D_{j+1},\ldots,D_{r}$ (for consistency, we set $D_1=\emptyset$). For each $1\leq j\leq r$, roughly every $h_j/\Delta$ iterations  (that is, when $h_j/\Delta$ vertices are deleted), our algorithm will recompute the graphs $\tilde \Lambda_j,\ldots, \tilde \Lambda_r$, the corresponding sets $D_{j+1},\ldots,D_r$ of vertices, and the $h_{j'}$-core decomposition of each graph $\tilde \Lambda_{j'}$, for all $j\leq j'\leq r$. This is done using Procedure $\PCL(\Lambda_j)$, that is formally defined in Figure~\ref{fig: PCL for computing layers}. This procedure is also used at the beginning of the algorithm, with $\Lambda_1=G^*$, to compute the initial partition into layers.
Note that some layers may be empty. Note also that, from our choice of parameters, $h_r\geq \tau/(64\Delta\log n)\geq n^{1/\log\log n}$, since $\tau\geq 4n^{2/\log\log n}$, $\Delta=2^{O(\sqrt{ \log n})}$, and $n$ is large enough. This ensures that we can apply Theorem~\ref{thm: find a core decomposition} to each graph $\tilde \Lambda_j$.

When a vertex is deleted from the original graph $G$, we will use procedure $\PDV(G^*,v)$ that we describe later, in order to update our data structures. As the result of this deletion, some vertices may stop being heavy for class $i$, and will need in turn be deleted from $G^*$. Procedure $\PDV(G^*,v)$ will iteratively delete all such vertices from the current graph, and then procedure $\PCL$ may be triggered as needed, as part of Procedure \PDV. When we say that some invariant holds throughout the algorithm, we mean that it holds between the different calls to procedure $\PDV(G^*,v)$, and it may not necessarily hold during the execution of this procedure.

\begin{figure}
\program{Procedure $\PCL(\Lambda_j,j)$}{

Input: an integer $1\leq j\leq r$ and a vertex-induced subgraph $\Lambda_j\subseteq G^*$ containing at most $\Delta nh_j$ edges, such that the degree of every vertex in $\Lambda_j$ is at least $h_r$.

\begin{enumerate}
\item If $i=r$, then set $\tilde \Lambda_r=\Lambda_r$; Compute the $h_r$-core decomposition $\fset_r$ of $\tilde \Lambda_r$ in time $O((|E(\tilde \Lambda_r)|+n^{1+o(1)})\poly\log n)=O(\Delta nh_r+n^{1+o(1)})=O(n^{1+o(1)}h_r)$ and terminate the algorithm.

From now on we assume that $j<r$.
\item Run Procedure $\PDS(\Lambda_j,h_j)$ to partition $V(\Lambda_j)$ into two subsets, $J_1,J_2$, in time $O(|E(\Lambda_j)|+|V(\Lambda_j)|)=O(\Delta nh_j)$.

\item Set $\tilde \Lambda_j=G^*[J_2]$; observe that every vertex in $\tilde \Lambda_j$ has degree at least $h_j$ and $|E(\tilde \Lambda_j)|\leq \Delta n h_j$.

\item Compute the $h_j$-core decomposition $\fset_j$ of $\tilde \Lambda_j$ in time  $O((|E(\tilde \Lambda_j)|+n^{1+o(1)})\poly\log n)=O(\Delta nh_j+n^{1+o(1)})=O(n^{1+o(1)}h_j)$. 

\item Temporarily set $\Lambda_{j+1}=G^*[J_1]$. Observe that $\Lambda_{j+1}$ has at most $nh_j=n\Delta h_{j+1}$ edges.

\item Run Procedure $\PDS( \Lambda_{j+1},h_r)$, to compute a partition $(R',R'')$ of $V(\Lambda_{j+1})$, so that every vertex of $R''$ has at least $h_r$ neighbors in $R''$, in time $O(|E(\tilde \Lambda_{j+1})|+|V(\tilde \Lambda_{j+1})|)\leq O(nh_j)$.

\item Set $D_{j+1}=R'$ and delete all vertices of $R'$ from $ \Lambda_{j+1}$.

\item Run $\PCL(\Lambda_{j+1},j+1)$.

\end{enumerate}
}
\caption{Procedure \PCL \label{fig: PCL for computing layers}}
\end{figure}

Note that the running time of Procedure $\PCL(\Lambda_j,j)$, excluding the recursive calls to the same procedure with graph $\Lambda_{j+1}$, is $O(n^{1+o(1)}h_j)$. Since the values $h_j$ form a geometrically decreasing sequence, the total running time of Procedure $\PCL(\Lambda_j,j)$, including all recursive calls is also bounded by $O(n^{1+o(1)}h_j)$. We will invoke this procedure at most $n\Delta/h_j$ times over the course of the algorithm -- roughly every $h_j/\Delta$ vertex deletions. Therefore, in total, all calls to Procedure \PCL will take time $O(n^{2+o(1)})$.

Throughout the algorithm, for each $1\leq j\leq r$, we denote by $K^*_j$ the set of all vertices that lie in the cores of $\fset_j$, that is, $K^*_j=\bigcup_{(K,U(K),\tilde \Lambda_j^K,W^k)\in \fset_j}K$, and by $\OK^*_j$ the set of the remaining vertices of $\tilde \Lambda_j$; recall that these vertices form an $h$-universal set in $\tilde \Lambda_j$ with respect to $K^*_j$.

The following claim and its corollary will allow us to deal with the discarded vertices, by showing that each such vertex can reach a core vertex via a short path.

For simplicity, for each $1\leq j\leq r$, we denote $S_j=V(\tilde \Lambda_j)\cup D_j$, and by $U_j=\bigcup_{j'\leq j}S_{j'}$. 

\begin{claim}\label{claim: extra vertex level i to level i-1}
Throughout the algorithm, for all $1<j\leq r$, for every vertex $v\in D_j$, there is a path $P(v)$ connecting $v$ to a vertex of $U_{j-1}$, of length at most $\log n$, such that $P(v)$ only contains vertices of $D_j$, except for its last vertex.
\end{claim}

\begin{proof}
Fix some index $1<j\leq r$. Let $t$ be the time right after the last execution  of $\PCL(\Lambda_{j-1},j-1)$ so far, and let $G'$ be the graph $G^*$ at time $t$. Denote $A=U_{j-1}$ at time $t$ and let $B$ contain the remaining vertices of $G'$, so $B=V(\Lambda_j)\cup D_j$ at time $t$. Let $(R_1,R_2)$ be the partition of $B$ produced by Procedure $\PDS(G'[B],h_r)$ when Procedure $\PCL(\Lambda_{j-1},j-1)$ was last invoked. Recall that we have set $D_j=R_1$.

Let us now consider the current graph $G^*$, and let $A'=U_{j-1}$ in the current graph, and let $B'$ contain the remaining vertices of $G^*$. Since $G^*$ contains fewer vertices than $G'$, if we run Procedure  $\PDS(G^*[B],h_r)$, and denote by $(R_1,R_2)$ the resulting partition of $B$, where each vertex in $R_2$ has at least $h_r$ neighbors in $R_2$, then $R_1$ will contain all vertices that currently belong to $D_j$.

Lastly, we consider the graph $G''=G^*[A'\cup D_j]$. We claim that every vertex of $D_j$ has degree at least $\tau/2$ in $G''$. Indeed, assume for contradiction that some vertex $v\in D_j$ has fewer than $\tau/2$ neighbors in $G''$. Since every vertex whose degree in $G^*$ falls below $\tau$ is deleted from $G^*$, vertex $v$ must have at least $\tau/2$ neighbors that do not lie in $A'\cup D_j$. Each such neighbor then must belong to the set $B\setminus D_j$. But then $v$ has at least $\tau/2>h_r$ neighbors in $B\setminus D_j$, so it should not have been added to the set $D_j$ when Procedure $\PDS(G'[B],h_r)$ was executed.

We conclude that every vertex of $D_j$ has degree at least $\tau/2$ in $G''$. However, if we run procedure $\PCL(G''[D_j],h_r)$, and let $(J',J'')$ be its outcome, then we know that $J''=\emptyset$ (this is because every vertex of $D_j$ was added to set $R_1$ by Procedure $\PDS(G'[B],h_r)$, and $D_j\subseteq B$). Since $h_r\leq \tau/(64\log n)$, we can now use Corollary~\ref{cor: degree sep: after deleting small number} with the partition $(A',D_j)$ of vertices of $G''$ and $R=\emptyset$, to conclude that every vertex of $D_j$ has a path of length at most $\log n$ connecting it to a vertex of $A'$ in $G''$, such that every inner vertex of the path lies in $D_j$.
%
\end{proof}

\begin{corollary}\label{cor: extra vertices to core}
Throughout the algorithm, for all $1< j\leq r$, for every vertex $v\in D_j$, there is a path $P$ of length at most $j\log n$, connecting $v$ to a vertex of $\bigcup_{j'< j}K^*_{j'}$, such that every vertex of $P$ lies in $\left(\bigcup_{j'< j}\tilde \Lambda_{j'}\right )\cup \left (\bigcup_{j'\leq j} D_{j'}\right )$.
\end{corollary}

\begin{proof}
The proof is by induction on $j$. For $j=1$, $D_1=\emptyset$, so the corollary trivially holds. Consider now some $j>1$, and some vertex $v\in D_j$.  From Claim~\ref{claim: extra vertex level i to level i-1}, there is a path $P$, connecting $v$ to a vertex of $U_{j-1}$,  of length at most $\log n$, such that $P$ only contains vertices of $D_j$, except for its last vertex, that we denote by $v'$. Assume that $v'\in S_{j_1}$, for some $j_1<j$.  We now consider three cases. First, if $v'\in K^*_{j_1}$, then we are done, and we can return the path $P$. Otherwise, if $v'\in \OK^*_{j_1}$, then, from the definition of the $h_{j_1}$-core decomposition, and from the fact that we re-compute this decomposition once $h_{j_1}/\Delta$ vertices are deleted, there is a path $P'$ of length at most $\log n$ connecting $v'$ to a vertex of $K^*_{j_1}$. We return a path obtained by concatenating $P$ and $P'$. Otherwise, $v'\in D_{j_1}$. From the induction hypothesis, there is a path $P'$ of length at most $j_1\log n\leq (j-1)\log n$, connecting $v'$ to a vertex of $\bigcup_{j'< j_1}K^*_{j'}$, such that every vertex of $P'$ lies in $\left(\bigcup_{j'< j_1}\tilde \Lambda_{j'}\right )\cup \left (\bigcup_{j'\leq j_1} D_{j'}\right )$. Concatenating paths $P'$ and $P$ gives the desired path.
\end{proof}

 \subsubsection{Data Structures}
 Our algorithm maintains the following data structures.
 
 First, we maintain the connectivity/spanning data structure $\CONNSF(G^*)$ for the graph $G^*$. Recall that the total time required to maintain this data structure under edge deletions is $O((|E(G^*)|+n)\log^2|V(G^*)|)=O((m+n)\log^2n)$, where $m=|E(G)|$ is the total number of edges in the original input graph $G$. Recall that the data structure can process queries of the form $\path(G^*,u,v)$: given two vertices $u$ and $v$ in $G^*$, return any simple path connecting $u$ to $v$ in $G^*$ if such a path exists, and return $\emptyset$ otherwise. If $u$ and $v$ belong to the same connected component $C$ of $G^*$, then this query can be processed in time $O(|V(C)|)$. 
 
 For every level $1\leq j\leq r$, we maintain the graphs $\Lambda_j$ and $\tilde \Lambda_j$, together with the $h_j$-core decomposition $\fset_j$ of  $\tilde \Lambda_j$, and the set $D_j$ of discarded vertices. As already discussed, all these are recomputed at most $n\Delta/h_j$ times over the course of the algorithm, by calling procedure $\PCL(\Lambda_j,j)$. Each call to the procedure requires running time $n^{1+o(1)}h_j$, and so overall, the running time spent on executing the procedure $\PCL(\Lambda_j,j)$, over the course of the algorithm, for all $1\leq j\leq r$, is at most $n^{2+o(1)}$.
 
 For every level $1\leq j< r$, for every vertex $v\in \Lambda_{j+1}$, we maintain a list $\delta_j(v)$ of all neighbors of $v$ in $G^*$ that lie in $\tilde \Lambda_j\cup D_j$. This list is recomputed from scratch every time Procedure $\PCL(\Lambda_j,j)$ is executed. It is easy to verify that this can be done without increasing the asymptotic running time of the procedure.

 For every level $1\leq j\leq r$, and every core structure $\kset=(K,U(K),\tilde \Lambda_j^K,W^K)\in \fset_j$, we maintain the data structure from Theorem~\ref{thm: maintaining a core}, that supports queries that, given a pair $u,v\in K$ of vertices of the core that were not deleted yet, return a path of length at most $\ell^*$ connecting $u$ to $v$ in $\tilde \Lambda_j$, or correctly establish that $\kset$ is not a perfect core structure, that is, $W^K$ is not an $\alpha^*$-expander. The total running time required to maintain this data structure for $\kset$ is $O(|E(\tilde \Lambda_j^K)|\poly\log n)$. Since the core decomposition of $\tilde \Lambda_j$ ensures that every edge of $\tilde \Lambda_j$ belongs to at most $\log n$ graphs $\tilde \Lambda_j^K$, where $K$ is a core from the decomposition, the total time required to maintain this data structure for all cores in $\fset_j$ is at most $|E(\tilde \Lambda_j)|\poly\log n=O(n\Delta h_j\poly\log n)=O(n^{1+o(1)}h_j)$. The core decomposition for $\tilde \Lambda_j$ is computed at most $n\Delta/h_j$ over the course of the algorithm, and for each such new core decomposition, we may spend up to $O(n^{1+o(1)}h_j)$ time maintaining its cores. Therefore, the total time spent on maintaining all cores, across all levels $1\leq j\leq r$, is at most $O(n^{2+o(1)}\Delta\log n)=O(n^{2+o(1)})$.
 
 For every level $1\leq j\leq r$, we maintain a counter $N(j)$, that counts the number of vertices that were deleted from the graph $G^*$ since the last time the procedure $\PCL(\Lambda_j,j)$ was called.
 
 
 Finally, we need to maintain data structures that allow us to find short paths from the vertices of $\OK^*_j$ to the vertices of $K^*_j$ for all $1\leq j\leq r$, and from the vertices of $D_j$ to the vertices of $\bigcup_{j'<j}K^*_j$. Let us fix a level $1\leq j\leq r$.
 
 First, we construct a new graph $H_j$, obtained from graph $\tilde \Lambda_j$, as follows. Let $\kset_1,\ldots,\kset_z\in \fset_j$ be the core structures that are currently in the core decomposition, and let $K_1,\ldots,K_z$ be their corresponding cores. Starting from graph $\tilde \Lambda_j$, we contract every core $K_y$ into a vertex $v(K_y)$. We then add a source vertex $s$, and connect it to each such new vertex $v(K_y)$.  All edges have unit length. The resulting graph is denoted by $H_j$. We maintain an Even-Shiloach tree for $H_j$, from the source vertex $s$, up to distance $(\log n+1)$: $\EST(H_j,s,(\log n+1))$. The total time required to maintain this tree is $O(|E(H_j)|\log^2 n)=O(n\Delta h_j\log^2 n)$. Graph $H_j$ and the tree $\EST(H_j,s,(\log n+1))$ will be recomputed at most $n\Delta/h_j$ times -- every time that the procedure $\PCL(\Lambda_j,j)$ is called. Therefore, the total time needed to maintain all these trees throughout the algorithm is $O(n^{2+o(1)})$.

 Lastly, we construct a new graph $H_j'$, as follows. We start with the sub-graph of $G^*$ induced by the vertices of $D_j$, and add a source vertex $s$ to it. We connect $s$ to every vertex $v\in D_j$ that has a neighbor in $\bigcup_{j'<j}(\tilde \Lambda_{j'}\cup D_{j'})$; in other words, for some $j'<j$, the list $\delta_{j'}(v)$ is non-empty. Note that every edge of $H'_j$, except for those incident to $s$, belongs to $\Lambda_{j-1}$, so $|E(H_j')|\leq n\Delta h_j$. 
 We maintain an Even-Shiloach tree of $H'_j$, from the source vertex $s$, up to distance $\log n$: $\EST(H'_j,s,\log n)$. The total time required to maintain this tree is $O(|E(H'_j)|\log n)=O(n\Delta h_j\log n)$. Graph $H'_j$ and the tree $\EST(H'_j,s,(\log n+1))$ will be recomputed at most $n\Delta/h_{j-1}=n/h_j$ times -- every time that the procedure $\PCL(\Lambda_{j-1},j-1)$ is called. Therefore, the total time needed to maintain all these trees throughout the algorithm is $O(n^{2+o(1)})$.

 \subsubsection{Vertex Deletion}
 We now describe an update procedure when a vertex $v$ is deleted from the graph $G$. First, if $v\not\in G^*$, then there is nothing to be done. Otherwise, we will maintain a set $Q$ of vertices to be deleted, that is initialized to $Q=\set{v}$. While $Q\neq \emptyset$, we let $u$ be any vertex in $Q$. We delete $u$ from $G^*$, updating the connectivity data structure $\CONNSF(G^*)$, and from all graphs $\Lambda_{j},\tilde\Lambda_j$, $H_j$, $H'_j$, to which $u$ belongs. We also update the affected Even-Shiloach trees for $H_j$ and $H'_j$ for all $j$. For every neighbor $u'$ of $u$ in $G^*$, we decrease $d(u')$ by $1$.   If $d(u')<\tau$ but $u'\not \in Q$, we add $u'$ to $Q$.
 
 Assume  that $v\in \Lambda_{j^*}\cup D_{j^*}$. 
 For every neighbor $u$ of $v$ that lies in $\Lambda_{j^*+1}$, we delete $v$ from the list $\delta_{j^*}(u)$. If $u\in D_{j'}$ for some $j'>j^*$, and all lists $\delta_{j''}(u)$ for $j''<j'$ become empty, then we delete the edge $(s,u)$ from graph $H'_{j'}$ and update the \EST accordingly.
 
  We also update the counters $N(j)$ with the number of deleted vertices. Once $Q=\emptyset$, we check whether we need to call procedure $\PCL(\Lamda_j,j)$ for any index $j$. In order to do so, for every $1\leq j\leq r$, we check whether $N(j)\geq h_j/\Delta$. If this is true for any $j$, we select the smallest such $j$, and run the procedure $\PCL(\Lambda_j,j)$. We also set the counters $N(j')$ for all $j'\geq j$ to $0$. We have already accounted for the time needed to maintain all our data structures. Additional running time required by the vertex deletion procedure is bounded by the sum of degrees of all vertices deleted from $G^*$ times $O(\log n)$, and so the total time incurred by the vertex deletion procedure over the course of the algorithm is $O(|E(G^*)|\log n)$.

 Overall, the running time of the whole algorithm is $n^{2+o(1)}$. It now remains to describe an algorithm for responding to queries.

 \subsubsection{Responding to Queries}
 Suppose we are given a query $\pquery(u,u',C)$, where $C$ is some connected component of $G^*$, and $u,u'\in C$. Our goal is to return a path connecting $u$ to $u'$ in $C$, of length at most $2^{13}|V(C)|\cdot \Delta \cdot \ell^*\cdot \log^4 n/\tau$, in expected time $O(|V(C)|\log^4n)$. 
 
 Our first step is to compute a simple path $P$ connecting $u$ to $u'$ in $C$, by calling  Procedure  $\path(G^*,u,u')$
 in the connectivity data structure  $\CONNSF(G^*)$. This query can be processed in time $O(|V(C)|)$.
 We denote this path by $P = (u_1,u_2,\ldots,u_z)$, where $u_1=u$ and $u_z=u'$.
 
 Let $\rset$ be the collection of all core sets $K$, whose corresponding core structure $\kset$ lies in $\bigcup_{j=1}^r\fset_j$. 
 We let $\rset'\subseteq \rset$ be the set of cores $K$ that are contained in $C$.
 Next, we label every vertex $u_a$ of $P$ with a core $K\in \rset'$, such that there is a path $P(u_a)$ of length at most $\log^2 n$ in $C$, connecting $u_a$ to a vertex of $K$. We will also store the path $P(u_a)$ together with $u_a$.  In order to do so, we consider every vertex $u_a\in P$ in turn.  If $u_a$ belongs to some core $K\in \rset$ (in which case $K\in \rset'$ must hold), then we assign to $u_a$ the label $K$, and we let $P(u_a)$ be the path containing a single vertex -- the vertex $u_a$. Otherwise, if $u_a\in \OK^*_j$ for some $1\leq j\leq r$, then we know that there is a path of length at most $\log n$, connecting $u_a$ to some core $K$ whose corresponding core structure lies in $\fset_j$, from the definition of core decomposition and $h_j$-universal sets. In order to find such a core $K$ and the corresponding path, we consider the graph $H_j$ and its corresponding tree $\EST(H_j,s,\log n+1)$. This tree must contain a path from $u_a$ to $s$, of length at most $\log n+1$. Let $v(K)$ be the penultimate vertex on this path. Then $K\in \rset'$, and we assign to $u_a$ the label $K$.
 We also store the path $P(u_a)$, connecting $u_a$ to a vertex of $K$, that we obtain by traversing this tree from $u_a$ to $s$; the length of $P(u_a)$ is at most $\log n$.
 
  Finally, assume that $u_a\in D_j$ for some $1\leq j\leq r$. We know that there is a path of length $O(\log^2n)$ connecting $u_a$ to some core $K\in \rset'$ from Corollary~\ref{cor: extra vertices to core}. In order to find such a core and the corresponding path, we start with the tree $\EST(H'_j,s,\log n+1)$, and retrace the path from $u_a$ to $s$ in this tree. The length of this path is at most $\log n$, and we let $v^1_a$ be the penultimate vertex on this path, and $P^1(u_a)$ the sub-path of this path connecting $u_a$ to $v^1_a$.
  Recall that $v^1_a$ must have a neighbor, that we denote by $u^1_a$, lying in $\Lambda_{j'}\cup D_{j'}$ for some $j'<j$, which can be found by inspecting the lists $\delta_{j'}(v^1_a)$. Let $e=(u^1_a,v^1_a)$ be the corresponding edge.
   We then consider three cases. First, if $u^1_a$ belongs to some core $K\in \rset$, then we terminate the algorithm and label $u_a$ with $K$; we also store the path $P(u_a)$, obtained by concatenating path $P^1(u_a)$ with edge $e$, together with $u_a$. Otherwise, if $u^1_a\in \OK^*_{j_1}$ for some $j_1<j$, then we compute a path $P^2(u_a)$, connecting $u^1_a$ to some core $K\in \rset'$ exactly as in the previous case, and label $u_a$ with $K$. We also store path $P(u_a)$, obtained from concatenating the path $P^1(u_a)$, the edge $e$, and the path $P^2(u_a)$, together with $u_a$. Finally, if neither of the above two cases happen, then $u^1_a\in D_{j_1}$ for some $j_1<j$. We then proceed to inspect the graph $H'_{j_1}$ and compute a path $P^2(u_a)$, connecting $u^1_a$ to the vertex $s$ in $\EST(H'_{j_1},s,\log n)$. We denote by $v^2_a$ the penultimate vertex on this path, and continue as before. Eventually, after at most $j$ iterations, we will construct a path $P(u_a)$, connecting $u_a$ to a vertex of some core $K\in \rset$, such that the length of the path is at most $j \log n\leq \log^2n$. We then label $u_a$ with $K$, and we store $P(u_a)$ together with $u_a$. The time required to find a label and a path $P(u_a)$ for every vertex $u_a$ is proportional to the length of the path, and is bounded by $O(\log^2n)$. Therefore, the total running time of this part of the algorithm is $O(|V(P)|\log^2n)\leq O(|V(C)|\log^2n)$.
  
  Our next step is to shortcut the path $P$: we would like to ensure that every label $K$ appears at most twice on the path $P$, and these two appearances are consecutive. In order to do this, we first create an array $A$ that contains an entry $A[K]$ for every label $K$ that appears on the path $P$; the number of such labels is at most $|V(P)|\leq |V(C)|$. Throughout the algorithm, entry $A[K]$ will contain a pointer to the first vertex, from among the currently processed vertices, on the current path $P$, whose label is $K$.   We process the vertices of $P$ one-by-one in their natural order along $P$. When a vertex $u_a$ is processed, we consider the label $K$ of $u_a$. If the entry $A[K]$ is currently empty, then we store in $A[K]$ a pointer to the vertex $u_a$ on path $P$. Otherwise, entry $A[K]$ contains a pointer to some vertex $u_b$, that appears before $u_a$ on path $P$, such that the label of $u_b$ is also $K$. If $u_b$ does not appear immediately before $u_a$ on $P$, then we discard the section of the path $P$ between $u_b$ and $u_a$ (but we keep these two vertices). For each discarded vertex $u_w$, if the label of $u_w$ is $K'$, then we delete from $A[K']$ a pointer that was stored there. Notice that, if $u_b$ and $u_a$ are consecutive on the path $P$, and they have the same label $K$, then only a pointer to $u_b$ is stored in $A[K]$. Observe that we process every vertex of $P$ at most twice - once when we inspect it for the first time, and once when we discard it. Therefore, the running time of this step of the algorithm is $O(|V(P)|)=O(|V(C)|)$.
  
  Let $Q$ be the sequence of vertices obtained from $P$ after the last step. 
 We denote $Q=(q_1,q_2,\ldots,q_{z'})$, where $q_1=u$ and $q_{z'}=u'$. Notice that for every consecutive pair $q_a,q_{a+1}$ of vertices in $Q$, either there is an edge $(q_a,q_{a'})$ in $G^*$, or these two vertices have the same label. Moreover, every label $K$ may appear at most twice in $Q$, as a label of two consecutive vertices.
 We claim that the length of $Q$ is at most $2^{15}|V(C)|\Delta \log^2 n/\tau$. Indeed, the length of $Q$ is bounded by $2|\rset'|$, where $\rset'$ is the collection of all cores $K$ contained in $C$.  Recall that for each $1\leq j\leq r$, every core structure $\kset\in \fset_j$ is an $h_j$-core structure, and so, from Observation~\ref{obs: size of h-core}, its corresponding core $K$ contains at least $h_j/(128\log n)$ vertices. As all cores $K$ in the decomposition $\fset_j$ are mutually vertex-disjoint, the total number of core structures in $\fset_j$, whose corresponding core is contained in $C$, is at most $128|V(C)|\log n/h_j$, and the total number of cores in $\rset'$ is at most $\sum_{j=1}^r128|V(C)|\log n/h_j\leq 256|V(C)|\log n/h_r\leq 2^{14}|V(C)|\Delta \log^2n/\tau$ (since $h_r\geq \tau/(64\Delta\log n)$ from the definition of $z_2$). Therefore, the length of $Q$ is at most $2^{15}|V(C)|\Delta \log^2 n/\tau$.
 
 Finally, we turn $Q$ into a path in $G^*$, by iteratively performing the following process. Let $q_a,q_{a+1}$ be a pair of consecutive vertices on $Q$, such that there is no edge connecting $q_a$ to $q_{a+1}$ in $G^*$. Then both $q_a$ and $q_{a+1}$ have the same label, that we denote by $K$, and we have stored two paths: path $P(q_a)$, connecting $q_a$ to some vertex $q'_a\in K$, and path $P(q_{a+1})$, connecting $q_{a+1}$ to some vertex $q'_{a+1}\in K$. The lengths of both paths are at most $\log^2n$. Assume that the core structure $\kset$ corresponding to $K$ lies in $\fset_j$. We then run the algorithm from Theorem~\ref{thm: maintaining a core} on $\kset$, $q'_a$ and $q'_{a+1}$. If the outcome of this algorithm is a path $Q_a$, of length at most $\ell^*$, connecting $q'_a$ to $q'_{a+1}$ in the current graph $\tilde\Lambda_j^K$, then we insert the concatenation of the paths $P(q_a),Q_a,P(q_{a+1})$ between $q_a$ and $q_{a+1}$ into $Q$, and continue to the next iteration. The running time for the current iteration is $O(\ell^*+|K|\log^3n)$. Otherwise, the algorithm correctly establishes that the core structure $\kset$ is not perfect, that is, $W^K$ is not an $\alpha^*$-expander. Since our core decomposition algorithm ensures that with high probability every core structure it computes is perfect, the probability that this happens is at most $1/n^c$ for some large constant $c$. In this case, we run Procedure $\PCL(\Lambda_1,1)$ and restart the algorithm for computing the path connecting $u$ to $u'$ in $C$ from scratch. The running time in this case is bounded by $O(n^{2+o(1)})$, but, since the probability of this event is at most $1/n^c$, the expected running time in this case remains $O(\ell^*+|K|\log^3n)$.
 
 We assume that every time Theorem~~\ref{thm: maintaining a core} is called, a path connecting the two corresponding vertices $q'_a$ and $q'_{a+1}$ is returned (as otherwise we start the algorithm from scratch). Once we process every consecutive pair $q_a,q_{a+1}$ of vertices on $Q$ that have no edge connecting them in $G^*$, we obtain a path connecting $u$ to $u'$ in $C$. The length of the path is bounded by $|Q|(\ell^*+\log^2n)$, where $|Q|$ is the length of the original sequence $Q$, so $|Q|\leq 2^{15}|V(C)|\Delta\log^2 n/\tau$. Therefore, the final length of the path that we obtain is at most $2^{13}|V(C)|\ell^*\Delta \log^4 n/\tau$, as required. We now bound the total expected running time of the last step. We invoke Theorem~\ref{thm: maintaining a core} at most once for every core $K$ that serves as a label of a vertex on $Q$, and each such call takes expected time 
 $O(\ell^*+|K|\log^3n)$. Recall that for all $1\leq j\leq r$, for all core structures $\kset\in \fset_j$, their corresponding cores are vertex-disjoint. Therefore, the total running time of this step is bounded by $O(\ell^*|\qset|+r|V(C)|\log^3n)=O(\ell^*\Delta |V(C)|\log^2 n/\tau)+O(|V(C)|\log^4n)=O(|V(C)|\log^4n)$, as $\tau\geq \ell^*\Delta$.

\label{--------------------------------------------------SEC: COMPUTING CORE DECOMP-------------------------------}
\section{Computing the Core Decomposition -- Proof of Theorem~\ref{thm: find a core decomposition}}\label{sec: computing core decomp}
The proof of Theorem~\ref{thm: find a core decomposition} relies on the following observation and theorem.

\begin{observation}\label{obs: find universal set}
Let $S\subseteq V(\hG)$ be any subset of vertices of $\hG$, and let $(J_1,J_2)$ be a partition of $V(\hG)\setminus S$ computed by  $\DSP(\hG\setminus S,d)$, where $d=h/(32\log n)$. Then:

\begin{itemize}
\item set $J_1$ is $h$-universal with respect to $S$ in graph $\hG$; and
\item the minimum vertex degree in $\hG[J_2]$ is at least $h/(32\log n)$.
\end{itemize}
\end{observation}

The observation immediately follows from Lemma~\ref{lemma: DSP universal} and Corollary~\ref{cor: degree sep: after deleting small number}, and from the fact that $\Delta>32\log n$.

\begin{theorem}\label{thm: find many cores}
There is a randomized algorithm, that, given a connected sub-graph $\tilde G\subseteq \hat G$, such that every vertex of $\tilde G$ has degree at least $h/(32\log n)$ in $\tilde G$, computes a collection $\fset=\set{(K_i,U(K_i),\tilde G^{K_i},W^{K_i})}_{i=1}^{r}$ of $h$-core  structures in $\tG$, for some $r>0$, such that:

\begin{itemize}
\item the sets $K_1,\ldots,K_r$ of vertices are mutually disjoint;
\item every edge of $\tG$ belongs to at most one graph of $\tG^{K_1},\ldots,\tG^{K_{r}}$; and
\item $\sum_{i=1}^{r}|K_i|\geq |V(\tG)|/2$.
\end{itemize}

Moreover, with high probability, each resulting core  structure in $\fset$ is perfect.
The running time of the algorithm is $O((|E(\tG)|+|V(\tG)|^{1+o(1)})\poly\log n)$.
\end{theorem}


We provide the proof of the theorem below, after we complete the proof of Theorem~\ref{thm: find a core decomposition} using it.
The algorithm employs Theorem~\ref{thm: find many cores} at most $O(\log n)$ times, and so with high probability the algorithm from  Theorem~\ref{thm: find many cores}  succeeds in all these executions, that is, all $h$-core  structures that we compute throughout the algorithm by invoking Theorem~\ref{thm: find many cores} are perfect $h$-core  structures. We assume that this is the case from now on.

Our algorithm performs a number of iterations. The input to the $i$th iteration is a family $\fset_i=\set{(K_j,U(K_j),{\tG}^{K_j},W^{K_j})}_{j=1}^{r_i}$ of perfect $h$-core  structures, such that the sets  $K_1,\ldots,K_{r_i}$ of vertices are mutually disjoint, and every edge of $\hG$ belongs to at most $(i-1)$ graphs in $\set{\hG^{K_1},\ldots,\hG^{K_{r_i}}}$. Let $S^i=K_1\cup\cdots\cup K_{r_i}$ and let $J^i=V(\hG)\setminus S^i$. We are also given a partition $(J^i_1,J^i_2)$ of $J^i$ into two subsets, such that set $J^i_1$ is $h$-universal for $S^i$, and, if we denote ${\tG}_i=\hG[J^i_2]$, then every vertex of ${\tG}_i$ has degree at least $h/(32\log n)$. In the input to the first iteration, $\fset_1=\emptyset$, $J^1_1=\emptyset$, and $J^1_2=V(\hG)$. Recall that all vertex degrees in $\hat G$ are at least $h$.

In order to execute the $i$th iteration of the algorithm, we apply Theorem~\ref{thm: find many cores} to every connected component of the graph ${\tG}_i$. Let $\fset'_i=\set{(K_j,U(K_j),{\tG}_i^{K_j},W^{K_j})}_{j=1}^{r'_i}$ be the union of the families of perfect $h$-core  structures that the theorem computes for all these components. Notice that every edge of $\tilde G$ belongs to at most one graph in $\set{\tG_i^{K_j}}_{j=1}^{r'_i}$. We then set $\fset_{i+1}=\fset_i\cup \fset'_i$, and denote $|\fset_{i+1}|$ by $r_{i+1}$. 
From the construction of ${\tG}_i$, we are guaranteed that all cores $K_j$ corresponding to the core structures in the resulting family $\fset_{i+1}$ are mutually disjoint, and that every edge of $\hG$ participates in at most $i$ graphs $\hG^{K_j}$. 
In order to construct the sets $J^{i+1}_1$, $J^{i+1}_2$ of vertices, we let $S^{i+1}=K_1\cup\cdots\cup K_{r_{i+1}}$ and   $J^{i+1}=V(\hG)\setminus S^{i+1}$, and apply Procedure $\DSP$ to graph $\hG\setminus S^{i+1}$ with the parameter $d=h/(32\log n)$. Let $(J^{i+1}_1, J^{i+1}_2)$ be the resulting partition of $J^{i+1}$. If $J^{i+1}_2\neq \emptyset$, then we continue to the next iteration. Otherwise, we terminate the algorithm. Notice that, from Observation~\ref{obs: find universal set}, set $J^{i+1}$ is $h$-universal with respect to $S^{i+1}$. We will prove below that the number of iterations in the algorithm is bounded by $\log n$, and so every edge of $\tG$ may belong to at most $\log n$ graphs $\tG^K$, where $K$ is a core in $\fset_{i+1}$, as it belongs to at most one such graph for every collection $\fset'_1,\ldots,\fset'_i$ of cores structures. Therefore, the current collection $\fset_{i+1}$ of core  structures defines a valid core decomposition. 

We now analyze the running time of the algorithm. It is easy to verify that every iteration takes time $O((|E(\hG)|+|V(\hG)|^{1+o(1)})\poly\log n)$. It is now sufficient to show that the number of iterations is bounded by $\log n$. From Claim~\ref{claim:DSP}, we are guaranteed that for all $i$, $J^{i+1}_2\subseteq J^{i}_2$, while the algorithm from Theorem~\ref{thm: find many cores} guarantees that the number of vertices that participate in the cores $K_1,\ldots,K_{r'_i}$ is at least $|V({\tG}_{i+1})|/2=|J^i_2|/2$. Therefore, for all $i$, $|J^i_2|\leq |J^{i-1}|/2$, and the total number of iterations is bounded by $\log n$.
It now remains to prove Theorem~\ref{thm: find many cores}.

\label{--------------------------------------subsec: find many cores------------------------------------------}
\subsection{Proof of Theorem~\ref{thm: find many cores}}

The basic block in the proof of Theorem~\ref{thm: find many cores} is the following theorem.

\begin{theorem}\label{thm: partition or core}
There is a randomized algorithm, that, given a connected sub-graph $\graph\subseteq G$, whose vertices are partitioned into a set $\Gamma$ of boundary vertices and a set $\Upsilon$ of non-boundary vertices, such that $|\Gamma|\leq |V(\graph)|/4$, and every vertex of $\Upsilon$ has degree at least $h/(32\log n)$ in $\graph$, returns one of the following:

\begin{itemize}
\item either a vertex cut $(X,Y,Z)$ of $\graph$ with $|Y|\leq \frac{\min\set{|X|,|Z|}}{\log^6n}$ and $|X|,|Z|\geq \frac{|V(\graph)|}{\log^4n}$ (an almost-balanced sparse vertex-cut); or

\item an $h$-core  structure $(K,U(K),\graph^K,W^K)$ in $\graph$, where $K\subseteq \Upsilon$, and $K$ contains all but at most $16|V(\graph)|/\log n$ vertices of $\Upsilon$, $\graph^K=\graph$, and $K\cup U(K)=V(\graph)$. Moreover, with high probability, $W^K$ is an $\alpha^*$-expander; in other words, with high probability, $(K,U(K),\graph^K,W^K)$ is a perfect core  structure.
\end{itemize}

The running time of the algorithm is  $O((|E(\graph)|+ |V(\graph)|^{1+o(1)})\poly\log n)$.
\end{theorem}



We delay the proof of this theorem to Section~\ref{sec: balanced cut or core}, and prove Theorem~\ref{thm: find many cores} using it here.
Throughout the algorithm, we maintain a family $\hset$ of connected sub-graphs of $\tilde G$, that we call clusters, and a partition of $\hset$ into three subsets: 
set $\aset$ of {\em active clusters}, set $\iset$ of {\em inactive clusters}, and set $\dset$ of {\em discarded clusters}. Additionally, we maintain a set $\Gamma\subseteq V(\tG)$ of vertices that we refer to as \emph{boundary vertices}, and a collection $\fset$ of $h$-core  structures, such that with high probability all core structures in $\fset$ are perfect. We will maintain the following invariants:

\begin{properties}{I}
\item for every cluster $H\in \aset\cup \iset$, $|\Gamma\cap V( H)|\leq |V( H)|/4$, and for every cluster $H\in \dset$, $|\Gamma\cap V( H)|> |V( H)|/4$; \label{inv: active cluster has few boundary and many vertices}

\item for every inactive cluster $H\in \iset$, there an $h$-core  structure $\kset(H)=(K,U(K),H^K,W^K)$ in $\fset$, with $H^K=H$, $K\cap \Gamma=\emptyset$, and $K\cup U(K)=V(H)$, such that $K$ contains all but at most $16 |V(H)|/\log n$ vertices of $V(H)\setminus \Gamma$, and with high probability, $W^K$ is an $\alpha^*$-expander;\label{inv inactive cluster has a core}

\item for every pair $ H, H'\in \hset$ of distinct clusters, $E( H)\cap E( H')=\emptyset$ and $V( H)\cap V( H')\subseteq \Gamma$;  \label{inv: disjointness of cluster}

\item for every cluster $ H\in \hset$, if $v\in V( H)$ is a non-boundary vertex (that is, $v\not\in \Gamma$), then $v$ has degree at least $h/(32\log n)$ in $H$, and  every neighbor of $v$ in $\tG$ belongs to $ H$; and \label{inv: inner vertices keep their neighbors}
\item every vertex of $\tG$ belongs to at least one cluster in $\hset$.
\end{properties}

At the beginning, $\hset$ contains a single cluster - the graph $\tilde G$; we also set $\aset=\hset$, $\iset=\dset=\emptyset$, $\fset=\emptyset$, and $\Gamma=\emptyset$. Notice that all invariants hold for this setting.
The algorithm consists of a number of phases. In each phase, we process every active cluster $ H\in \aset$.

Consider now some active cluster $H\in\aset$, and let $\Gamma(H)=\Gamma\cap V(H)$, and $\Upsilon(H)=V(H)\setminus \Gamma$. 
Notice that from Invariants~(\ref{inv: active cluster has few boundary and many vertices}) and (\ref{inv: inner vertices keep their neighbors}), $|\Gamma(H)|\leq |V(H)|/4$, and every vertex of $\Upsilon$ has degree at least $h/(32\log n)$ in $H$.
In order to process the cluster $H\in \aset$, we apply Theorem~\ref{thm: partition or core} to it, with the set $\Gamma(H)$ of boundary vertices, and the set $\Upsilon(H)$ of non-boundary vertices. If the outcome is an $h$-core  structure $(K,U(K), H^K,W^K)$, then we add this core  structure as $\kset(H)$ to $\fset$, and move $ H$ from $\aset$ to $\iset$. Note that all invariants continue to hold. 

Otherwise, we obtain a vertex cut $(X,Y,Z)$ of $V( H)$, with $|Y|\leq \frac{\min\set{|X|,|Z|}}{\log^6n}$ and $|X|,|Z|\geq \frac{|V( H)|}{\log^4n}$. We define two new graphs, $H_1$ and $H_2$ as follows. We start with $ H_1= H[X\cup Y]$ and $ H_2= H[X\cup Z]$, and then delete, from both graphs, all edges whose both endpoints belong to $Y$. This ensures that these two new graphs do not share any edges. The vertices of $Y$ are added to $\Gamma$, where they become boundary vertices.  Next, we remove the cluster $H$ from $\aset$, and consider every connected component $\tilde H$ of $H_1$ and $H_2$ one-by-one. For each such component $\tilde H$, if $|V(\tilde H)\cap \Gamma|>|V(\tilde H)|/4$, then we add $\tilde H$ to $\dset$, and otherwise, we add it to $\aset$. It is easy to verify that all invariants continue to hold. This completes the description of a phase. The processing of a single cluster $H\in \aset$ takes time $O((|E(H)|+ |V(H)|^{1+o(1)})\poly\log n)$, and, since the clusters are disjoint in their edges, and since, in every active cluster $H$, a constant fraction of its vertices are non-boundary vertices that are not shared with other clusters, the total running time of every phase is $O((|E(\tG)|+ |V(\tG)|^{1+o(1)})\poly\log n)$. The algorithm terminates once $A=\emptyset$.
 Next, we bound the number of phases in the following claim.
 
 \begin{claim}\label{claim: num of phases}
 There are at most $\log^5n$ phases in the algorithm.
 \end{claim}
 \begin{proof}
  Consider an iteration when a cluster $H\in \aset$ was processed, and assume that the algorithm from Theorem~\ref{thm: partition or core} returned a vertex cut $(X,Y,Z)$ for this cluster. Then, since $|X|,|Z|\geq \frac{|V(H)|}{\log^4n}$, we are guaranteed that for every connected component $\tilde H$ of $H_1$ and $H_2$, $|V(\tilde H)|\leq \left (1-\frac{1}{\log^4n}\right ) |V(H)|$. Therefore, at the end of every phase $i$, for every cluster $ H\in \aset$, we are guaranteed that $|V( H)|\leq \left (1-\frac{1}{\log^4n}\right )^i|V(\tG)|$, and so the total number of phases is bounded by $\log^{5}n$.
\end{proof}

The final collection $\fset$ of core structures contains all core structures $\kset(H)$, for $H\in \iset$.  Denote $\fset=\set{(K_i,U(K_i),\tilde G^{K_i},W^{K_i})}_{i=1}^{r}$. Recall that every core structure in $\fset$ is an $h$-core  structure, and with high probability, all these core structures are perfect. Invariant~(\ref{inv: disjointness of cluster}) ensures that every edge of $\tG$ belongs to at most one cluster $H$, and hence to at most one graph of  $\tG^{K_1},\ldots,\tG^{K_{r}}$. Since, for every cluster $H\in \iset$, the vertices lying in the core $K$ of the corresponding core structure $\kset(H)$ are non-boundary vertices, from Invariant~(\ref{inv: disjointness of cluster}), the sets $K_1,\ldots,K_r$ of vertices are mutually disjoint. Notice that the total running time of the algorithm is $O((|E(\tG)|+ |V(\tG)|^{1+o(1)})\poly\log n)$.  It now only remains to show that $\sum_{i=1}^{r}|K_i|\geq |V(\tG)|/2$.
For convenience, we denote $K^*=\bigcup_{i=1}^rK_i$. Notice that, if $v\not\in K^*$, then one of the following three cases must happen: either (i) $v\in \Gamma$; or (ii) $v\in V(H)$ where $H\in\dset$ is a discarded cluster; or (iii) $v\in \Upsilon(H)$ for some cluster $H\in \iset$, but $v$ does not belong to the corresponding core; there are at most $16|V(H)|/\log n$ vertices of the latter type for each $H\in \iset$. We now bound the sizes of each of these three vertex sets in turn.

We let $\Gamma^+$ be a multi-set of boundary vertices, where for each boundary vertex $v\in \Gamma$, the number of copies of $v$ that are added to $\Gamma^+$ is precisely the number of clusters in $\hset$ containing $v$. At the beginning of the algorithm, $\Gamma^+=\emptyset$. As the algorithm progresses, new vertices (or copies of old vertices) are added to $\Gamma^+$.

\begin{claim}\label{claim: bound on boundary vertices}
At the end of the algorithm, $|\Gamma^+|<|V(\tG)|/128$.
\end{claim}
\begin{proof}
Recall that the number of phases in our algorithm is bounded by $\log^{5}n$. We now bound the number of new vertices added to $\Gamma^+$ in every phase. Consider some phase of the algorithm, and let $H\in \aset$ be a cluster that was processed during that phase. If we found a core  structure in $H$ and moved $H$ to $\iset$, then no new vertices where added to $\Gamma^+$ while processing $H$. Assume now that we have computed a vertex cut $(X,Y,Z)$ in $H$. Then up to two new copies of every vertex in $Y$ are added to $\Gamma^+$. Let $N(H)$ denote the number of non-boundary vertices in $H$. From Invariant~(\ref{inv: active cluster has few boundary and many vertices}), $|N(H)|\geq 3|V(H)|/4$, and the algorithm from Theorem~\ref{thm: partition or core} guarantees that $|Y|\leq \frac{|V(H)|}{2\log^6n}\leq \frac{2|N(H)|}{3\log^6n}$. Therefore, at most $\frac{4|N(H)|}{3\log^6n}$ new vertices are added to $\Gamma^+$ when cluster $H$ is processed. Since, from Invariant~(\ref{inv: disjointness of cluster}), a non-boundary vertex may belong to at most one cluster, the total number of vertices added to $\Gamma^+$ over the course of a single phase is at most $\frac{4|V(\tG)|}{3\log^6n}$, and, since the number of phases is at most $\log^5n$, the total number of vertices that belong to $\Gamma^+$ at the end of the algorithm is at most $O(|V(\tG)|/\log n)<|V(\tG)|/128$. (We have used the fact that $n$ is large enough).
\end{proof}

Let $D$ denote the set of all non-boundary vertices that lie in the clusters of $\dset$. Since, from Invariant~(\ref{inv: active cluster has few boundary and many vertices}), for every cluster $H\in \dset$, $|V(H)\cap \Gamma|\geq |V(H)|/4$, we get that $|D|\leq 4|\Gamma^+|\leq |V(\tG)|/32$.

Lastly, let $R$ denote the set of all vertices $v$, such that (i) there is an inactive cluster $H\in \iset$ with $v\in V(H)\setminus \Gamma$, and (ii) $v\not\in K^*$. Recall that each inactive cluster $H$ contributes at most $16|V(H)|/\log n$ vertices to $R$. Therefore:

\[|R|\leq \frac{16}{\log n}\sum_{H\in \iset} |V(H)|\leq \frac{16}{\log n}(|V(\tG)|+|\Gamma^+|)\leq \frac{16}{\log n}\cdot\frac{129}{128}|V(\tG)|\leq \frac{|V(\tG)|}{128}, \]

if $n$ is sufficiently large. Therefore, overall:

\[|\Gamma|+|D|+|R|\leq |V(\tG)|\left(\frac{1}{128}+\frac{1}{32}+\frac{1}{128}\right )<\frac{|V(\tG)|}{2}.\]

Since $K^*=V(\tG)\setminus (\Gamma\cup D\cup R)$, we get that $|K^*|\geq |V(\tG)|/2$.

\label{---------------------------------------------------SEC: FIND BALANCED CUT OR CORE------------------------------------------}
\section{Proof of Theorem~\ref{thm: partition or core}}\label{sec: balanced cut or core}
This section is dedicated to the proof of Theorem~\ref{thm: partition or core}. For convenience, we denote $\hn=|V(\graph)|$.


We say that a vertex cut $(X,Y,Z)$ of $\graph$ is \emph{acceptable}, if $|Y|\leq \frac{\min\set{|X|,|Z|}}{\log^6 n}$ and $|X|,|Z|\geq \frac{\hn}{\log^4n}$. Our goal is to either compute an acceptable vertex cut, or an $h$-core structure with the required properties.


 The proof  of Theorem~\ref{thm: partition or core} consists of three parts. In the first part, we will either return an acceptable cut $(X,Y,Z)$ in $\graph$, in which case we terminate the algorithm and return this cut; or we will embed a graph $W$ into $\graph$, that is ``almost" an expander, in the sense that every balanced cut in $W$ is large. Every edge $e$ of $W$ is mapped to a path $P(e)$ in $\graph$ that is sufficiently short, and every vertex of $\graph$ participates in a small number of such paths. In the second part, we find a core  structure in $\graph$ by computing an $\alpha^*$-expander $W'\subseteq W$. The embedding of $W$ that was computed in the first part then immediately defines an embedding of $W'$ into $\graph$, and graph $W'$ will serve as the witness graph for the core  structure. The vertices of $W'$ become the core itself, and the remaining vertices of $\graph$ become the extension of the core. 
 In the third and the last part, we turn the resulting core structure into an $h$-core structure.
 The first two parts of the proof use the cut-matching game from Theorem~\ref{thm:cut-matching-game}, on $\hn$ vertices. Notice that the theorem only guarantees that the construction of the expander is successful with probability at least $(1-1/\poly(\hn))$, while we need our algorithm to succeed with probability at least $(1-1/\poly(n))$. In order to improve the probability of success, it is enough to repeat the algorithm $O(\log n/\log \hn)=O(\log\log n)$ times since $\hn = \Omega(n^{1/\log \log n}/\log n)$.
Also, recall that every run of the algorithm requires $O(\log^2\hn)$ iterations. In order to simplify the calculations, every time we need to use the cut-matching game, we will run it for $\floor{\log^3n}$ iterations altogether; assuming that $n$ is large enough, this ensures that the probability of success is at least $(1-1/n^c)$, for a large enough constant $c$, e.g., $c=1000$.

\label{--------------------------------------subsec: part 1------------------------------------------}
\subsection{Part 1 of the Algorithm.} 
The first part of the algorithm relies on the following theorem.

\begin{theorem}\label{thm: embed almost expander or a balanced cut}
There is a randomized algorithm that, given the graph $\graph$ as in the statement of Theorem~\ref{thm: partition or core}, and parameters $z>0,\ell>2\sqrt{\log n}$, computes, in time $O(|E(\graph)|\cdot \ell^3\log^{3}n+|V(\graph)|\poly\log n)$
one of the following:

\begin{itemize}
\item Either a vertex cut $(X,Y,Z)$ in $\graph$ with $|Y|\leq \frac{8\log n}{\ell}\min\set{|X|,|Z|}$ and $|X|,|Z|\geq \frac{z}{2}$;
\item or a graph $W$ with $V(W)=V(\graph)$, where the degree of every vertex in $W$ is at most $\log^3n$, together with a path $P(e)$ for every edge $e=(u,v)\in E(W)$, such that $P(e)$ connects $u$ to $v$ in $\graph$, and:

\begin{itemize}
\item The length of each path in $\set{P(e)\mid e\in E(W)}$ is at most $\ell$;
\item Every vertex of $\graph$ participates in at most $\ell^2\log^3n$  paths  in $\set{P(e)\mid e\in E(W)}$; and

\item with high probability, for every partition $(A,B)$ of the vertices of $\graph$ with $|A|\leq |B|$, if $|E_W(A,B)|<|A|/4$, then $|A|\leq 4 z\log^3n$.
\end{itemize}
\end{itemize}
\end{theorem}

\begin{proof}
The main tool that we use in the proof of the theorem is the following lemma.

\begin{lemma}\label{lem: flow or cut}
There is an algorithm, that, given  the graph $\graph$ as in the statement of Theorem~\ref{thm: partition or core}, two disjoint equal-cardinality subsets $A,B$ of $V(\graph)$, and parameters $z>0,\ell>2\sqrt{\log n}$, computes one of the following:

\begin{itemize}
\item Either a collection $\pset$ of at least $|A|-z$ paths in $\graph$, where each path connects a distinct vertex of $A$ to a distinct vertex of $B$; every path has length at most $\ell$; and every vertex of $\graph$ participates in at most $\ell^2$ paths; or

\item A vertex-cut $(X,Y,Z)$ in $\graph$, with $|Y|\leq \frac{8\log n}{\ell}\min\set{|X|,|Z|}$, and $|X|,|Z|\geq z/2$.
\end{itemize}

The running time of the algorithm is $ O(|E(\graph)|\ell^3)$.
\end{lemma}
We defer the proof of Lemma~\ref{lem: flow or cut} for later, after we prove Theorem~\ref{thm: embed almost expander or a balanced cut} using it.

We start with the graph $W$ containing all the vertices of $\graph$ and no edges. We then run the cut-matching game for $\floor{\log^3n}$ iterations. Recall that in each iteration $i$, we are given two disjoint equal-cardinality subsets $A_i,B_i$ of $V(W)$, and our goal is to return a complete matching $M_i$ between $A_i$ and $B_i$. The edges of $M_i$ are then added to $W$. 

We now describe the execution of the $i$th iteration. We apply Lemma~\ref{lem: flow or cut} to graph $\graph$, with the sets $A_i,B_i$ of vertices, and the same parameters $\ell,z$. If the outcome of the lemma is a vertex cut $(X,Y,Z)$, then we return this cut and terminate the algorithm -- it is immediate to verify that this cut has the required properties. Therefore, we assume that the algorithm has returned a set $\pset_i$ of paths, connecting at least $|A_i|-z$ pairs of vertices from $A_i\times B_i$, such that the length of each path is at most $\ell$, the paths are disjoint in their endpoints, and they cause vertex-congestion at most $\ell^2$. We let $M'_i$ be the set of pairs of vertices matched by the paths in $\pset_i$, and we let $F_i$ be an arbitrary matching of the remaining vertices, so that  $M_i=M'_i\cup F_i$ is a complete matching between $A_i$ and $B_i$. We add the edges of $M_i$ to $W$, and we call the edges of $F_i$ \emph{fake edges}. Notice that the number of fake edges is at most $z$.
This concludes the description of the $i$th iteration. The running time of an iteration is $O(|E(\graph)|\cdot \ell^3)$ plus the time required to compute the partition $(A_i,B_i)$, which is bounded by $ O(\hn\poly\log n)$ from Theorem~\ref{thm:cut-matching-game}.

If, at any time during the algorithm, we compute a vertex cut $(X,Y,Z)$ with the required properties, then the algorithm terminates and we return this cut. Therefore, we assume that the algorithm always computes the matchings $M_i$. The final graph $W$ is then a $1/2$-expander with probability at least $(1-1/\poly(n))$. Let $W'$ be the graph obtained from $W$ after we delete all fake edges from it. Notice that the total number of the fake edges in $W$ is at most $ z\log^3n$. 
We immediately obtain, for every edge $e\in E(W')$, a path $P(e)$ in $\graph$ that connects its endpoints and has length at most $\ell$. Since we have $\floor{\log^3n}$ iterations, the paths in  $\set{P(e)\mid e\in E(H)}$ cause vertex-congestion at most $\ell^2\log^3n$ in $\graph$. Finally, 
let $(A,B)$ be any partition of $V(W')$, with $|A|\leq |B|$, and assume that $|E_{W'}(A,B)|<|A|/4$. Assuming that $W$ was indeed a $1/2$-expander, $|E_W(A,B)|\geq |A|/2$. Therefore, there are at least $|A|/4$ fake edges in $E_W(A,B)$. As the total number of the fake edges is bounded by 
$z\log^3n$, we get that $|A|\leq 4 z\log^3n$.

The running time of every iteration is $O(|E(\graph)|\cdot \ell^3)$, and, since we have $O(\log^3n)$ iterations, the total running time is $O(|E(\graph)|\cdot \ell^3\log^{3}n+|V(\graph)|\poly\log n)$.

In order to complete the proof of Theorem~\ref{thm: embed almost expander or a balanced cut}, it is now enough to prove Lemma~\ref{lem: flow or cut}

\begin{proofof}{Lemma~\ref{lem: flow or cut}}
We partition the algorithm into phases. The input to phase $i$ are subsets $A_i\subseteq A$, $B_i\subseteq B$ of vertices that were not routed yet, with $|A_i|=|B_i|$. We will ensure that during the $i$th phase, we either compute a set $\pset_i$ of at least $\frac{|A_i|\log n}{\ell^2}$ node-disjoint paths, connecting  vertices of $A_i$ to  vertices of $B_i$, such that the length of every path in $\pset_i$ is at most $\ell$; or we will return a vertex-cut $(X,Y,Z)$ with the required properties. The algorithm terminates once $|A_i|\leq z$.
Since we are guaranteed that for every $i$, $|A_i|\leq |A_{i-1}|(1-\log n/\ell^2)$, the number of phases is bounded by $\ell^2$. The final set of paths is $\pset=\bigcup_i\pset_i$, and, since the paths in every set $\pset_i$ are node-disjoint, the paths in $\pset$ cause vertex-congestion at most $\ell^2$. We will also ensure that every phase runs in time $O(|E(\graph)|\ell)$, which will ensure that the total running time is  $ O(|E(\graph)|\ell^3)$, as required. The input to the first phase is $A_1=A$ and $B_1=B$.
It is now enough to describe the execution of a single phase. The next claim will then finish the proof of the lemma.

\begin{claim}\label{claim: flow or cut one phase}
There is a deterministic algorithm, that, given a connected graph $\graph$ as in the statement of Theorem~\ref{thm: embed almost expander or a balanced cut}, two equal-cardinality subsets $A',B'$ of $V(\graph)$, and parameters $z>0,\ell>2\sqrt{\log n}$, computes one of the following:

\begin{itemize}
\item Either a collection $\pset'$ of at least $\frac{|A'|\log n}{\ell^2}$ node-disjoint paths in $\graph$, where each path connects a distinct vertex of $A'$ to a distinct vertex of $B'$ and has length at most $\ell$; or

\item A vertex-cut $(X,Y,Z)$ in $\graph$, with $|Y|\leq \frac{8\log n}{\ell}\min\set{|X|,|Z|}$, and $|X|,|Z|\geq |A'|/2$.
\end{itemize}

The running time of the algorithm is $ O(|E(\graph)|\ell)$.
\end{claim}

\begin{proof}
We construct a new graph $H$: start with  graph $\graph$, and add a source vertex $s$ that connects to every vertex in $A'$ with an edge; similarly, add a destination vertex $t$, that connects to every vertex in $B'$ with an edge. Set up single-source shortest path data structure $\EST(H,s,\ell)$, up to depth $\ell$ in $H$, with $s$ being the source. Initialize $\pset'=\emptyset$. While the distance from $s$ to $t$ is less than $\ell$, choose any path $P$ in $H$ connecting $s$ to $t$, that has at most $\ell$ inner vertices. Add $P$ to $\pset'$, and delete all inner vertices of $P$ from $H$. Notice that finding the path $P$ takes time $O(\ell)$, since we simply follow the path from $t$ to $s$ in the $\EST$. The total update time of the data structure is $ O(|E(\graph)|\ell)$, and the total running time of the algorithm, that includes selecting the paths and deleting their inner vertices from $H$, is bounded by $O(|E(\graph)|\ell)$. We now consider two cases. First, if $|\pset'|\geq  \frac{|A'|\log n}{\ell^2}$ at the end of the algorithm, then we terminate the algorithm, and return the set $\pset'$ of paths.

Otherwise, consider the current graph $H'$, that is obtained from $H$ after all vertices participating in the paths in $\pset'$ were deleted. We perform a BFS from the vertices of $A'$ in this graph: start from $S_0=A'\cap V(H')$. Given the current set $S_j$, let $S_{j+1}$ contain all vertices of $S_{j}$ and all neighbors of $S_j$ in $H'\setminus\set{s,t}$. 
Similarly, we perform a BFS from the vertices of $B'$ in $H'$: 
start from $T_0=B'\cap V(H')$. Given the current set $T_j$, let $T_{j+1}$ contain all vertices of $T_{j}$ and all neighbors of $T_j$ in $H'\setminus\set{s,t}$. 

We claim that there must be some index $j<\ell/2$, such that one of the following happen: either (i) $|S_{j+1}|\leq |V(H')|/2$ and
 $|S_{j+1}|<|S_j|\left (1+\frac{2\log n}{\ell}\right)$; or (ii) $|T_{j+1}|\leq |V(H')|/2$ and $|T_{j+1}|<|T_j|\left (1+\frac{2\log n}{\ell}\right)$. Indeed, if no such index exists, then $|S_{\ell/2}|>\hn/2$ and $|T_{\ell/2}|>\hn/2$, so
 there is still a path from $s$ to $t$ containing at most $\ell$ vertices in $H'$, and the algorithm should not have terminated.
 
 We assume w.l.o.g. that  $|S_{j+1}|\leq |V(H')|/2$ and $|S_{j+1}|<|S_j|\left (1+\frac{2\log n}{\ell}\right)$. We now define a vertex cut  $(X,Y,Z)$ in $\graph$, as follows. Set $X$ contains all vertices of $S_j$. Notice that, in particular, $X$ contains all vertices of $A'$ that still need to be routed, so $|X|\geq |A'|\left(1-\frac{\log n}{\ell^2}\right )\geq |A'|/2$. Set $Z$ contains all vertices of $H'\setminus(\set{s,t}\cup S_{j+1})$. Since we have assumed that no path of length $\ell$ connecting a vertex of $A'$ to a vertex of $B'$ exists in $H'$, set $Z$ contains all vertices of $B'$ that still need to be routed, so $|Z|\geq |B'|/2$. Set $Y$ contains all remaining vertices -- the vertices that are neighbors of $S_j$ in $H'$ -- their number is at most $\frac{2|S_j|\log n}{\ell}$, and the vertices that participate in the paths in $\pset'$ -- their number is at most $ |\pset'|\cdot \ell\leq \frac{|A'|\log n}{\ell^2}\cdot \ell\leq \frac{|A'|\log n}{\ell}\leq \frac{2|S_j|\log n}{\ell}$. Therefore, altogether, $|Y|\leq \frac{4|X|\log n}{\ell}$. Since $|S_{j+1}|\leq |V(H')|/2$, we get that $|Z|\geq |V(H')|/2\geq |X|/2$. In particular, $|Y|\leq \frac{8|Z|\log n}{\ell}$. From the above discussion, $|X|,|Z|\geq |A'|/2$.
 
 The running time of the first part of the algorithm, when the paths of $\pset'$ are computed is $O(|E(\graph)|\cdot \ell)$, as discussed above. The second part only involves computing two BFS searches in graph $H'$, and has running time of $O(|E(\graph)|)$.
\end{proof}
\end{proofof}
\end{proof}

We apply Theorem~\ref{thm: embed almost expander or a balanced cut} to the input graph $\graph$, with $z=\frac{2\hn}{\log^4 n}$, and $\ell=8\log^8 n$. Note that the total running time of the algorithm from Theorem~\ref{thm: embed almost expander or a balanced cut} becomes $O(|E(\graph)|\cdot \poly\log n)$. If the outcome is a vertex cut $(X,Y,Z)$ with $|Y|\leq \frac{8\log n}{\ell}\min\set{|X|,|Z|}=\frac{\min\set{|X|,|Z|}}{\log^7n}$, and $|X|,|Z|\geq z/2=\hn/\log^4 n$, then we obtain an acceptable cut. We terminate the algorithm and return this cut. Therefore, we assume from now on, that the outcome of the theorem is a graph $W$, with $V(W)=V(\graph)$, such that the maximum vertex degree in $W$ is at most $\log^3n$. Additionally, for every edge $e\in E(W)$, we are given a path $P(e)$ in $\graph$ connecting its endpoints, such the length of the path is at most $\ell=O(\log^8n)$, and every vertex in $\graph$ participates in at most $\ell^2\log^3n=O(\log^{19}n)$ such paths. Moreover, with high probability, for every partition $(A,B)$ of $V(W)$ with $|A|\leq |B|$ and $|E_W(A,B)|< \frac{|A|}4$, the following inequality must hold:

\[|A|\leq 4z\log^3n=\frac{8\hn\log^3n}{\log^4 n}\leq  \frac{8\hn}{\log n}.\]


\label{--------------------------------------subsec: part 2------------------------------------------}
\subsection{Part 2 of the Algorithm}

The goal of the second part is to prove the following theorem. Recall that $\alpha^*=1/2^{3\sqrt{\log n}}$.
\begin{theorem}\label{thm: fast expander construction from large balanced cut}
There is a randomized algorithm, that, given  a graph $W$ on $\hn$ vertices, where $\frac{n^{1/\log\log n}}{32\log n}\leq \hn\leq n$, such that:

\begin{itemize}
\item The maximum vertex degree in $W$ is at most $\log^3n$; and
\item for every partition $(A,B)$ of $V(W)$ with $|A|,|B|\geq\frac{8\hn}{\log n}$, $|E_W(A,B)|\geq \frac{\min\set{|A|,|B|}}4$,
\end{itemize}

 returns, in time $O(\hn^{1+o(1)}\poly\log n)$, a subgraph $W^*\subseteq W$, with $|V(W^*)|\geq \left(1-\frac{8}{\log n}\right)\hn$, such that with high probability, $W^*$ is a $4\alpha^*$-expander. 
\end{theorem}

Note that using the above theorem, we immediately obtain a core  structure $(K,U(K),\graph^K,W^K)$, where $K=V(W^*)$, $U(K)=V(W\setminus W^*)$, so that $K\cup U(K)=V(\graph)$, $\graph^K=\graph$, and $|U(K)|\leq 8\hn/\log n$. The witness graph is $W^K=W^*$, and its embedding consists of the set $\set{P(e)\mid e\in E(W^*)}$ of paths computed in the first part of the algorithm. The total running time of the first two parts of the algorithm is $O(|E(\graph)|\poly\log n+\hn^{1+o(1)}\poly\log n)$. In the remainder of this subsection, we focus on proving Theorem~\ref{thm: fast expander construction from large balanced cut}. Throughout the proof, we will only be concerned with regular edge-cuts of graphs: that is, a partition of the vertices of the graph into two disjoint non-empty subsets.

Given a graph $H$ and a parameter $\alpha$, we say that a cut $(A,B)$ of $H$ is \emph{$\alpha$-sparse} iff $|E_H(A,B)|\leq \alpha\min\set{|A|,|B|}$.  Notice that, if $\alpha'<\alpha$, then every cut that is $\alpha'$-sparse is also $\alpha$-sparse. The \emph{profit} of the cut $(A,B)$ is $\min\set{|A|,|B|}$. 
Using this language, we can rephrase the assumption in Theorem~\ref{thm: fast expander construction from large balanced cut} as follows:

\begin{properties}{A}
\item For every $\alpha\leq 1/4$, every $\alpha$-sparse cut in $W$ has profit at most $\frac{8\hn}{\log n}$. \label{prop: no small balanced cut}
\end{properties}

We will frequently invoke the following observation, that immediately follows from this assumption, and from the fact that $n$ is large enough.

\begin{observation}\label{obs: large to very large side}
Let $(A,B)$ be any cut in $W$, such that the cut $(A,B)$ is $1/4$-sparse, and $|A|\geq n/4$. Then $|A|\geq  \left(1-\frac{8}{\log n}\right)\hn$.
\end{observation}


Our proof of Theorem~\ref{thm: fast expander construction from large balanced cut} is almost identical to the arguments used in~\cite{Saranurak} to compute a global expander decomposition. The main difference is that their algorithm relies on existing algorithms for computing approximate maximum $s$-$t$ flow and minimum $s$-$t$ cut  (in the standard edge-capacitated version), while our algorithm avoids this by designing a simple 
algorithm that provides a rough solution to the maximum flow and the minimum cut problems that need to be solved. The setting of our parameters is also somewhat different.

We now provide a high-level overview of the proof of Theorem~\ref{thm: fast expander construction from large balanced cut}. The algorithm uses a procedure that we informally call ``cutting $\alpha$-sparse cuts off of $W$''. The procedure, at a high level, proceeds as follows. We are given some parameter $\alpha<1/4$. Start with $W'=W$, and then iterate. In every iteration $j$, we compute an $\alpha$-sparse cut $(A_j,B_j)$ of $W'$. Assuming w.l.o.g. that $|A_j|\geq |B_j|$, we set $W'=W'[A_j]$, and continue to the next iteration. The algorithm terminates once the current graph $W'$ does not have $\alpha$-sparse cuts. We then know that $W'$ is an $\alpha$-expander. Suppose the algorithm terminates after $r$ iterations, and let $W^*$ be the graph $W'$ at the end of the last iteration. 
The key is to observe that, since each cut $(A_j,B_j)$ was $\alpha$-sparse, then so is the final cut $(V(W^*),\bigcup_{i=1}^rB_j)$. Moreover, since, in every iteration, we have assumed that $|A_j|\geq |B_j|$, from repeatedly applying Property~(\ref{prop: no small balanced cut}) and Observation~\ref{obs: large to very large side}, $W^*$ contains at least $(1-8/\log n)\hn$ vertices. In general, every iteration, in which we compute an $\alpha$-sparse (or an approximately $\alpha$-sparse) cut can be implemented in time $O(|E(W)|\poly\log n)=O(\hn\poly\log n)$ using known algorithms for the sparsest cut problem (but our final algorithm does not rely on them). The difficulty with this approach is that the number of iterations can be very large, and this may result in a large running time overall.

In order to overcome this difficulty, we exploit the fact that every $\alpha$-sparse cut in $W$, for $\alpha\leq 1/4$, has profit at most $\frac{8\hn}{\log n}$. Suppose we were given an algorithm, that, given the promise that the maximum profit of an $\alpha$-sparse cut in the given graph $W$ is $z$, returns an $\alpha$-sparse cut $(A,B)$ in $W$ of profit at least $z/x$, for some $x=\hn^{o(1)}$. We could then use this algorithm to compute the cuts $(A_j,B_j)$ in our procedure for cutting $\alpha$-sparse cuts off of $W$, to ensure that the number of iterations is bounded by $O(x\log \hn)$. Indeed, using the observation that, after each iteration $j$ of the algorithm, if $W'$ is the current graph, then $(V(W'),B_1\cup\cdots\cup B_j)$ is an $\alpha$-sparse cut, and the fact that $|B_1|,\ldots,|B_j|\geq z/x$, we get that the maximum profit of an $\alpha$-sparse cut in $W'$ must reduce by at least factor $2$ every $O(x)$ iterations, and so the number of iterations is bounded by $O(x\log n)$.

Unfortunately, we do not have an algorithm that,  given the promise that the maximum profit of an $\alpha$-sparse cut in the given graph $W$ is $z$, returns an $\alpha$-sparse cut $(A,B)$ in $W$ of profit at least $z/x$, as such an algorithm, in particular, would have to solve the sparsest cut problem exactly. Instead, we provide a randomized algorithm that, given parameters $\alpha$ and $z$, either returns an $\alpha$-sparse cut of profit at least $z$, or with high probability correctly determines that every $\alpha^3$-sparse cut has profit at most $8z\log^3n$. This is done in the following theorem, whose proof is delayed for now. We note that~\cite{Saranurak} prove a stronger version of this theorem that obtains better bounds, but their proof relies on existing algorithms for approximate maximum $s$-$t$ flow and minimum $s$-$t$ cut in the standard edge-capacitated setting, which we prefer to avoid.


\begin{theorem}\label{thm: sparse edge cut of large profit or witness}
There is a randomized algorithm, that, given a sub-graph $W'$ of $W$ containing at least half the vertices of $W$, together with parameters $0<\alpha\leq 1/(64\log^9n)$ and $1\leq z\leq \hn$:

\begin{itemize}
\item either returns an $\alpha$-sparse cut in $W'$ of profit at least $z$;

\item or with high probability correctly establishes that every $\alpha^3$-sparse cut has profit at most $8z\log^3n$.
\end{itemize}

The running time of the algorithm is $O((\hn\poly\log n)/\alpha^3)$.
\end{theorem}


If the maximum profit of an $\alpha$-sparse cut and the maximum profit of an $\alpha^3$-sparse cut were close to each other, then we could use the above theorem in every iteration of the procedure for cutting $\alpha$-sparse cuts off of $W'$; the only difference from the previous argument would be that the final graph would be an $\alpha^3$-expander and not an $\alpha$-expander. Unfortunately, the maximum profit of an $\alpha$-sparse cut and the maximum profit of an $\alpha^3$-sparse cut could be very far from each other. In order to overcome this difficulty, we employ the following strategy (which is almost identical to the strategy of~\cite{Saranurak}). Our algorithm consists of a number of phases. At the beginning of the first phase, we set $\alpha_1=1/\poly\log n$ and 
$z_1=\frac{8\hn}{\log n}$. Notice that from Property~(\ref{prop: no small balanced cut}), we are guaranteed that the maximum profit of an $\alpha_1$-sparse cut in the initial graph $W'=W$ is at most $z_1$. We then run our procedure that cuts $\alpha_1$-sparse cuts off of $W'$, by employing Theorem~\ref{thm: sparse edge cut of large profit or witness} in each iteration, with $\alpha=\alpha_1$, and $z=z_1/(8x\log^3n)$, where $x=n^{o(1)}$ is some parameter that we set later. A phase terminates once the theorem establishes that every $\alpha^3$-sparse cut in $W'$ has profit at most $8z\log^3n=z_1/x$. Since the cuts we compute in every iteration have profits at least $z_1/(8x\log^3n)$, while the maximum profit of an $\alpha_1$-sparse cut in $W'$ is at most $z_1$, the number of iterations is bounded by $8x\log^3n=n^{o(1)}$. We then continue to the second phase, setting $\alpha_2=\alpha_1^3$ and $z_2=z_1/x$. Notice that we are now guaranteed that every $\alpha_2$-sparse cut in $W'$ has profit at least $z_2$. Each subsequent phase is executed exactly as before, until some phase $i^*$, when $z_{i^*}\leq x$ holds for the first time. At this point, the maximum profit of each $\alpha_{i^*}$-sparse  is sufficiently small, at most $z_{i^*}\leq x$, and therefore we no longer need to compute $\alpha_{i^*}$-sparse cuts whose profit is large. In the final phase, we will employ the following theorem for computing $\alpha_{i^*}$-sparse cuts, instead of Theorem~\ref{thm: sparse edge cut of large profit or witness}.


\begin{theorem}\label{thm: sparse edge cut or expander}
There is a randomized algorithm, that, given a sub-graph $W'$ of $W$ containing at least half the vertices of $W$, together with a parameter $0<\alpha\leq 1/(64\log^9n)$:

\begin{itemize}
\item either returns an $\alpha$-sparse cut in $W'$;

\item or with high probability correctly establishes that $W'$ is an $\Omega(\alpha^3)$-expander.
\end{itemize}

The running time of the algorithm is $O((\hn\poly\log n)/\alpha^3)$.
\end{theorem}
The last phase is executed exactly like the previous phases, except that Theorem~\ref{thm: sparse edge cut or expander} is employed in each iteration, instead of Theorem~\ref{thm: sparse edge cut of large profit or witness}. The last phase terminates when Theorem~\ref{thm: sparse edge cut or expander} establishes that $W'$ is an $\Omega(\alpha_{i^*}^3)$-expander. The parameter $\alpha^*$ in the definition of a core  structure is chosen to match this value. Recall that we have used a parameter $x=\hn^{o(1)}$, and the number of iterations in every phase is bounded by roughly $x\poly\log n$, so we would like $x$ to be sufficiently small. On the other hand, as the parameters $z_i$ decrease by a factor of $x$ from phase to phase, the number of phases is at most $i^*=\log_x \hn=\log \hn/\log x$. The final bound $\alpha^*$ on the expansion that we obtain depends exponentially on $i^*$, since for all $i$, $\alpha_i=\alpha_{i-1}^3$. The bound $\alpha^*$ in turn influences the running time of the algorithm (the lengths $\ell^*$ of the paths returned inside each core  structure by the algorithm from Theorem~\ref{thm: maintaining a core} depend on $\alpha^*$, and need to be balanced by the threshold $\tau$ that we use in the definition of light and heavy graphs, which in turn affects the running time of the algorithm for the light graph), so we would like to keep $i^*$ small, by letting $x$ be sufficiently large. We choose the parameter $x$ in order to balance these two considerations.

We now provide a formal proof of Theorem~\ref{thm: fast expander construction from large balanced cut}.
The proofs of Theorems~\ref{thm: sparse edge cut of large profit or witness} and~\ref{thm: sparse edge cut or expander} are very similar to the proof of Theorem~\ref{thm: embed almost expander or a balanced cut}, and are provided later, after we complete the proof of Theorem~\ref{thm: fast expander construction from large balanced cut} using them.

We use the following parameters. Let $\alpha_1=\frac{1}{64\log^9n}$, and for $i>1$, let $\alpha_i=\alpha^3_{i-1}$. We also set $z_1=\frac{8}{\log n}\cdot \hn$, and, for $i>1$, $z_{i}=z_{i-1}/x$, where $x=\hn^{8/\log\log n}$.

Our algorithm consists of a number of phases. The input to the $i$th phase is a subgraph $W_i\subseteq W$. 
We denote $V(W_i)=S_i$, $V(W\setminus W_i)=\overline{S}_i$ and $|S_i|=n_i$. 
We will also guarantee that with high probability the following three properties hold:

\begin{properties}{P}
\item The profit of every $\alpha_i$-sparse cut in $W_i$ is at most $z_i$;  \label{prop: sparse cuts have small profit}
\item $|S_i|\geq \hn\left (1-\frac{8}{\log n}\right )$; and \label{prop: many vertices}

\item $|E_W(S_i,\overline{S}_i)|\leq \frac 1 4 |\overline{S}_i|$. \label{prop: sparse cuts}
\end{properties}

The input to the first phase is graph $W_1=W$. From Assumption~(\ref{prop: no small balanced cut}),  Properties~(\ref{prop: sparse cuts have small profit})-(\ref{prop: sparse cuts}) hold for this input. We now describe the execution of the $i$th phase.


\paragraph{Execution of the $i$th Phase.}
The execution of the $i$th phase consists of a number of iterations. The input to the $j$th iteration is a subgraph $C_j\subseteq W_i$, where at the beginning, $C_1=W_i$. We will ensure that the following invariants hold:

\begin{properties}{I}
\item $|V(C_j)|\geq  \hn \left(1-\frac{8}{\log n}\right)$; and \label{inv: large cluster}
\item $|E(C_j,W_i\setminus C_j)|\leq \alpha_i |V(W_i\setminus C_j)|$ \label{inv: small cut}
\end{properties}


We use an additional auxiliary parameter, $z'_i=\frac{z_{i+1}}{8\log^3n}=\frac{z_i}{x\cdot 8\log^3n}$.

In order to execute the $j$th iteration, we apply Theorem~\ref{thm: sparse edge cut of large profit or witness} to graph $C_j$, with parameters $\alpha=\alpha_i$, and $z=z'_i$. We then consider two cases. In the first case, the algorithm from the theorem concludes that every $\alpha^3_i$-sparse cut in $C_j$ has profit at most $8 z'_i\log^3n\leq z_{i+1}$. In this case, we set $W_{i+1}=C_j$, and we terminate the phase. We show that $W_{i+1}$ is a valid input to the $(i+1)$th phase. First, Invariant~(\ref{inv: large cluster}) immediately implies Property~(\ref{prop: many vertices}). Since $\alpha_{i+1}=\alpha_i^3$, we are guaranteed with high probability that every $\alpha_{i+1}$-sparse cut  in $W_{i+1}$ has profit at most $z_{i+1}$, establishing Property~(\ref{prop: sparse cuts have small profit}). Finally, to establish Property~(\ref{prop: sparse cuts}), we need to show that $|E(C_j,W\setminus C_j)|\leq \frac 1 4 |V(W\setminus C_j)|$.
Observe that the set $E(C_j,W\setminus C_j)$ of edges consists of two subsets: the edges of $E(C_j,W_i\setminus C_j)$, whose number is bounded by $
 \alpha_i |V(W_i\setminus C_j)|$ from Invariant~(\ref{inv: small cut}); and the edges that belong to $E(S_i,\overline{S_i})$, whose number is bounded by $\frac{1}{4}|\overline S_i|=\frac{1}{4}|V(W\setminus W_i)|$ from Property~(\ref{prop: sparse cuts}). Therefore, altogether, $|E(C_j,W\setminus C_j)|\leq  \alpha_i |V(W_i\setminus C_j)|+\frac{1}{4}|V(W\setminus W_i)|\leq \frac 1 4 |V(W\setminus C_j)|$. We conclude that Property~(\ref{prop: sparse cuts}) continues to hold, and that $W_{i+1}$ is a valid input to the $(i+1)$th phase.

Assume now that Theorem~\ref{thm: sparse edge cut of large profit or witness} returns an $\alpha_i$-sparse cut $(A_j,B_j)$ of $C_j$, of profit at least $z'_i$. Assume without loss of generality that $|A_j|\geq |B_j|$. We then set $C_{j+1}=W[A_j]$, and continue to the next iteration. We now verify that the invariants~(\ref{inv: large cluster}) and (\ref{inv: small cut}) continue to hold. 

In order to do so, we consider three subsets of vertices of $S_i$: set $X=A_j$, set $Y=B_j$, and set $Z$ containing the remaining vertices of $S_i$. 
From Invariant~(\ref{inv: small cut}), $|E(X\cup Y,Z)|=|E(C_j,W_i\setminus C_j)|\leq \alpha_i|V(W_i\setminus C_j)|=\alpha_i|Z|$, and so:

\[\begin{split}
|E(C_{j+1},W_i\setminus C_{j+1})|&=|E(X,Y\cup Z)|\\
&\leq |E(X,Z)|+|E(X,Y)|\\
&\leq |E(X\cup Y,Z)|+\alpha_i |Y|\\
&\leq \alpha_i|Y\cup Z|\\
&=\alpha_i|V(W_i\setminus C_{j+1})|.\end{split}\]

 Therefore, Invariant~(\ref{inv: small cut}) continues to hold.

From Invariant~(\ref{inv: large cluster}), we are guaranteed that  $|X\cup Y|\geq  \left(1-\frac{8}{\log n}\right)\hn$, and from the choice of $X$ and $Y$, $|X|\geq |V(C_j)|/2\geq \hn/4$. In order to establish Invariant~(\ref{inv: large cluster}) for $C_{j+1}$, from Observation~\ref{obs: large to very large side}, it is enough to show that $|E(X,V(W)\setminus X)|\leq \frac 1 4\min\set{|X|, |V(W\setminus X)|}$. We show this using a similar reasoning to the one we used in establishing Invariant~(\ref{inv: small cut}). Let $Z'=V(W\setminus W_i)$. Then:

\[\begin{split}
|E(X,W\setminus X)|&= |E(X,Y)|+|E(X,Z)|+|E(X,Z')|\\
&\leq  |E(X,Y)|+|E(V(C_j),Z)|+|E(S_i,Z')| \\
&\leq  \alpha_i|Y|+\alpha_i|Z|+\frac 1 4 |Z'|\\
&\leq \frac 1 4 \min\set{|X|, (|Y|+|Z|+|Z'|)}\\
&=\frac 1 4 \min\set{|X|,|V(W\setminus X)|}.
\end{split}\]
(For the second inequality, we have used Invariant~(\ref{inv: small cut}) for $C_j$ and Property~(\ref{prop: sparse cuts}) for $S_i$. For the third inequality, we have used the facts that $|Z'|\leq \hn/(8\log n)$ and $|X|\geq \hn /4$.)
Therefore, Invariant~(\ref{inv: large cluster}) continues to hold.

This concludes the description of an iteration and of a phase. We now bound its running time. 

A single iteration takes time $O(\hn\poly\log n/\alpha_i^3)$. Next, we bound the number of iterations in a phase. Assume that the number of iterations is $r+1$, that is, in the first $r$ iterations we have computed the cuts $(A_j,B_j)$, and in the last iteration we have set $W_{i+1}=C_{r+1}$. Since we are guaranteed that $|V(C_{r+1})|\geq 3n/4$, and that $|E(C_{r+1},W_i\setminus C_{r+1})|\leq \alpha_i |V(W_i\setminus C_{r+1})|$, cut $(V(C_{r+1}),V(W_i\setminus C_{r+1}))$ is an $\alpha_i$-sparse cut of $W_i$, and so its profit is bounded by $z_i$, that is $|V(W_i\setminus C_{r+1})|\leq z_i$. But $V(W_i\setminus C_{r+1})=B_1\cup B_2\cup\cdots\cup B_r$, and for each $1\leq j\leq r$, $|B_j|\geq z'_i=\frac{z_i}{x\cdot 8\log^3n}$. Therefore, the number of iterations  in a phase is bounded by $O(x\log^3n)$, and the total running time of a phase is $O(\hn x\poly\log n/\alpha_i^3)$.

We execute each phase as described above, as long as $z_i\geq 1$. Let $i^*$ be the largest integer for which $z_{i^*}\geq 1$. Then $z_{i^*}\leq x$, and $i^*\leq \frac{\log \hn}{\log x}=\frac{\log \hn}{\log(\hn^{8/\log\log n})}=\frac{\log\log n}{8}$.
Note that $\alpha_{i^*+1}\geq \frac{1}{(\log n)^{10\cdot 3^{i^*}}}$, and:

\[(\log n)^{10\cdot 3^{i^*}}=2^{10\log\log n\cdot 3^{i^*}}\leq 2^{10\log\log n\cdot 3^{\log\log n/8}}\leq 2^{2^{\log \log n/2}/9}\leq 2^{\sqrt{\log n}/9}.\]

Therefore, $\alpha_{i^*+1}\geq 1/2^{\sqrt{\log n}/9}$. So far, we have obtained a subgraph $W_{i^*}\subseteq W$, with $|V(W_{i^*})|\geq  \left(1-\frac{8}{\log n}\right)\hn$, such that every $\alpha_{i^*}$-sparse cut in $W_{i^*}$ has profit at most $x$. The total running time of the algorithm so far is bounded by $O(\hn xi^*\poly\log n/\alpha_{i^*+1}^3)=O(\hn\cdot \hn^{8/\log\log n}\cdot 2^{O(\sqrt{\log n})})=O(\hn^{1+o(1)})$, as $\hn\geq n^{1/\log\log n}/(32\log n)$.

Finally, in order to turn $W_{i^*}$ into a $4\alpha^*$-expander, we run a final phase. The execution of the phase is identical to the execution of all previous phases, with parameter $\alpha_{i^*+1}=\alpha_{i^*}^3$. We will still ensure that Invariants~(\ref{inv: large cluster}) and (\ref{inv: small cut}) hold throughout the phase. The only difference is that in every iteration $j$, we apply Theorem~\ref{thm: sparse edge cut or expander} instead of Theorem~\ref{thm: sparse edge cut of large profit or witness} to $C_j$, with $\alpha=\alpha_{i^*+1}$. 
If the algorithm from Theorem~\ref{thm: sparse edge cut or expander} establishes that $C_j$ is an $\alpha_{i^*+1}^3$-expander, then we terminate the algorithm and return $W^*=C_j$. Since $\alpha^*=1/2^{3\sqrt{\log n}}$, we are now guaranteed that $W^*$ is a $4\alpha^*$-expander with high probability, and, from Invariant~(\ref{inv: large cluster}), $|V(W^*)|\geq  \left(1-\frac{8}{\log n}\right)\hn$. Otherwise, Theorem~\ref{thm: sparse edge cut or expander} returns an $\alpha_{i^*+1}$-sparse cut $(A_j,B_j)$ of $C_j$. As before, we assume without loss of generality that $|A_j|\geq |B_j|$, set $C_{j+1}=W[A_j]$, and continue to the next iteration. Using the same reasoning as before, Invariants~(\ref{inv: large cluster}) and (\ref{inv: small cut}) continue to hold. Repeating the same analysis as before, with $z_{i^*+1}=1$, it is easy to verify that the number of iterations in this final phase is bounded by $x$, and so the running time of the phase is bounded by $O(\hn x\poly\log n/\alpha_{i^*+1}^3)=O(\hn^{1+o(1)})$.
Overall, the algorithm returns a graph $W^*\subseteq W$ with $|V(W^*)|\geq (1-8/\log n)\hn$, such that with probability at least $(1-1/\poly(n))$, graph $W^*$ is a $4\alpha^*$-expander, with total running time $O(\hn^{1+o(1)})$.


In order to complete the proof of Theorem~\ref{thm: fast expander construction from large balanced cut}, it is now enough to prove Theorems~\ref{thm: sparse edge cut of large profit or witness} and~\ref{thm: sparse edge cut or expander}.

\subsection*{Proofs of Theorems~\ref{thm: sparse edge cut of large profit or witness} and~\ref{thm: sparse edge cut or expander}}
The proofs of the two theorems are very similar.
The main tool that both proofs use is the following lemma, that is an analogue of Lemma~\ref{lem: flow or cut} for edge cuts. Its proof is almost identical to the proof of Lemma~\ref{lem: flow or cut} and is delayed to the Appendix.

\begin{lemma}\label{lem: flow or cut edge version}
There is a deterministic algorithm, that, given a subgraph $W'\subseteq W$ containing at least half the vertices of $W$, together with two equal-cardinality subsets $A,B$ of $V(W')$, and parameters $0<z<\hn,\ell>2{\log^{1.5} n}$, computes one of the following:

\begin{itemize}
\item either a collection $\pset$ of more than $|A|-z$ paths in $W'$, where each path connects a distinct vertex of $A$ to a distinct vertex of $B$; every path has length at most $\ell$; and every edge of $W'$ participates in at most $\ell^2/\log^2n$ paths; or

\item a cut $(X,Y)$ in $W'$, with $|E_{W'}(X,Y)|\leq  \frac{4\log^4 n}{\ell}\min\set{|X|,|Y|}$, and $|X|,|Y|\geq z/2$.
\end{itemize}

The running time of the algorithm is $ O(\hn\ell^3\log n)$.
\end{lemma}

\begin{proofof}{Theorem~\ref{thm: sparse edge cut of large profit or witness}}

The idea of the proof is to run the Cut-Matching Game for $\floor{\log^3n}$ iterations in order to embed an ``almost'' expander into $W'$, similarly to the proof of Theorem~\ref{thm: embed almost expander or a balanced cut}.

Let $\ell= \frac{4\log^4n}{\alpha}$ and $z'=2z$. We start with a graph $H$, whose vertex set is $V(W')$, and edge set is empty. We then run the cut-matching game on graph $H$. Recall that in each iteration $1\leq i\leq \floor{\log^3n}$, we are given two disjoint subsets $A_i,B_i$ of $V(H)$ of equal cardinality, and our goal is to return a complete matching $M_i$ between $A_i$ and $B_i$. The edges of $M_i$ are then added to $H$. 

We now describe the execution of the $i$th iteration. We apply Lemma~\ref{lem: flow or cut edge version} to graph $W'$, with the sets $A_i,B_i$ of vertices, and the parameters $\ell,z'$. If the outcome of the lemma is a cut $(X,Y)$ in $W'$, with $|E_{W'}(X,Y)|\leq  \frac{4\log^4 n}{\ell}\min\set{|X|,|Y|}=\alpha\cdot \min\set{|X|,|Y|}$, and $|X|,|Y|\geq z'/2=z$, then we obtain an $\alpha$-sparse cut of profit at least $z$ in $W'$. We terminate the algorithm and return this cut. Therefore, we assume that the algorithm has returned a set $\pset_i$ of paths of cardinality at least $|A_i|-z'$, where each path connects a distinct vertex of $A_i$ to a distinct vertex of $B_i$, and the paths of $\pset_i$ cause edge-congestion at most $\frac{\ell^2}{\log^2n}=\frac{16\log^6n}{\alpha^2}$. We let $M'_i$ be the set of pairs of vertices matched by the paths in $\pset_i$, and we let $F_i$ be an arbitrary matching of the remaining vertices, so that  $M_i=M'_i\cup F_i$ is a complete matching between $A_i$ and $B_i$. We add the edges of $M_i$ to $H$, and we call the edges of $F_i$ \emph{fake edges}. Notice that the number of fake edges is at most $z'$.
This concludes the description of the $i$th iteration. The running time of an iteration is $O(\hn\ell^3\log n)=O((\hn\poly\log n)/\alpha^3)$, plus $O(\hn\poly\log n)$ time required to compute the sets $A_i,B_i$, from Theorem~\ref{thm:cut-matching-game}. 

If, at any time during the algorithm, we find an $\alpha$-sparse cut in $W'$ of profit at least $z$, then  the algorithm terminates and we return this cut. Therefore, we assume from now on that in every iteration $i$, the algorithm computes the matching $M_i$. The final graph $H$ is then a $1/2$-expander with probability at least $(1-1/\poly(n))$. Let $H'$ be the graph obtained from $H$ after we delete all fake edges from it. Notice that the total number of the fake edges in $H$ is at most  $z'\log^3n=2 z\log^3n$. We also obtain a set $\pset=\bigcup_{i=1}^{\floor{\log^3n}}\pset_i$ of paths in $W'$ that contains, for every edge $e\in E(H')$, a path $P(e)$ connecting its endpoints, such that the edge-congestion caused by the paths of $\pset$ in $W'$ is at most $\floor{\log^3n}\cdot \frac{16\log^6n}{\alpha^2}=\frac{16\log^9n}{\alpha^2}$.

Consider now some cut $(X,Y)$ in $W'$, with $|X|,|Y|> 8 z\log^3n$. Assume w.l.o.g. that $|X|\leq |Y|$. Since graph $H$ is a $\half$-expander, $E_{H}(X,Y)\geq |X|/2> 4 z \log^3 n$. 

Since the total number of fake edges in $H$ is at most $2 z\log^3n$, fewer than half the edges of $E_H(X,Y)$ are fake, so $|E_{H'}(X,Y)|> |X|/4$. For every edge $e\in E_{H'}(X,Y)$, there is a path $P(e)\in \pset$ connecting its endpoints. In particular, path $P(e)$ must contain an edge of $E_{W'}(X,Y)$. As the paths in $\pset$ cause edge-congestion at most  $\frac{16\log^9n}{\alpha^2} $, we get that:

\[|E_{W'}(X,Y)|\geq \frac{|E_{H'}(X,Y)|}{(16\log^9n)/(\alpha^2)}> \frac{|X|}{4}\cdot \frac{\alpha^2}{16\log^9n}\geq \alpha^3|X|,\]

since we have assumed that $\alpha\leq \frac{1}{64\log^9n}$. Therefore, every $\alpha^3$-sparse cut in $W'$ has profit at most $8 z\log^3n$.

Finally, we bound the running time of the algorithm. The algorithm has $O(\log^3n)$ iterations, and the running time of every iteration is $O((\hn\poly\log n)/\alpha^3)$. Therefore, the total running time is $O((\hn\poly\log n)/\alpha^3)$.
\end{proofof}

\begin{proofof}{Theorem~\ref{thm: sparse edge cut or expander}}
The proof is almost identical to the proof of Theorem~\ref{thm: sparse edge cut of large profit or witness}. As before, we set  $\ell= \frac{4\log^4n}{\alpha}$, and we run the cut-matching game on graph $H$ for $\floor{\log^3n}$ iterations, where at the beginning $V(H)=V(W')$ and $E(H)=\emptyset$. Consider the $i$th iteration of the game, where we are given two disjoint subsets $A_i,B_i$ of $H$ that have equal cardinality. We apply Lemma~\ref{lem: flow or cut edge version} to graph $W'$, with the sets $A_i,B_i$ of vertices, and the parameters $\ell$ as defined above, with $z=1$. If the outcome of the lemma is a cut $(X,Y)$ in $W'$, with $|E_{W'}(X,Y)|\leq  \frac{4\log^4 n}{\ell}\min\set{|X|,|Y|}=\alpha\cdot \min\set{|X|,|Y|}$, then we obtain an $\alpha$-sparse cut in $W'$, and terminate the algorithm. Otherwise, the algorithm must return a set $\pset_i$ of paths of cardinality $|A_i|$, where each path connects a distinct vertex of $A_i$ to a distinct vertex of $B_i$, and the paths of $\pset_i$ cause edge-congestion at most $\frac{\ell^2}{\log^2n}=\frac{16\log^6n}{\alpha^2}$. The set $\pset_i$ of paths naturally defines a complete matching $M_i$ between $A_i$ and $B_i$.  We add the edges of $M_i$ to $H$, and  continue to the next iteration. As before, the running time of an iteration is $O(\hn\ell^3\log n)=O((\hn\poly\log n)/\alpha^3)$. 

If, at any time during the algorithm, we find an $\alpha$-sparse cut in $W'$, then  the algorithm terminates and we return this cut. Therefore, we assume from now on that in every iteration $i$, the algorithm computes the matching $M_i$. Note that the paths in $\bigcup_i \pset_i$ cause total edge-congestion at most $\frac{16\log^9n}{\alpha^2}$. The final graph $H$ is then a $1/2$-expander w.h.p. Using the same reasoning as in the proof of Theorem~\ref{thm: sparse edge cut of large profit or witness}, it is easy to see that $W'$ is an $\Omega(\alpha^3)$-expander. Indeed, consider any cut $(X,Y)$ in $W'$ and assume w.l.o.g. that $|X|\leq |Y|$. Since graph $H$ is a $\half$-expander, $E_{H}(X,Y)\geq |X|/2$.  For every edge $e\in E_{H}(X,Y)$, there is a path $P(e)\in \bigcup_i\pset_i$ connecting its endpoints. In particular, path $P(e)$ must contain an edge of $E_{W'}(X,Y)$. As the paths in $\bigcup_i\pset_i$ cause edge-congestion at most  $\frac{16\log^9n}{\alpha^2} $, we get that:

\[|E_{W'}(X,Y)|\geq \frac{|E_{H}(X,Y)|}{(16\log^9n)/(\alpha^2)}\geq \frac{|X|}{2}\cdot \frac{\alpha^2}{16\log^9n}> \alpha^3|X|,\]

since we have assumed that $\alpha\leq \frac{1}{64\log^9n}$. Therefore, graph $W'$ does not contain $\alpha^3$-sparse cuts, and so it is an $\alpha^3$-expander.

The total running time of the algorithm is bounded by $O((\hn\poly\log n)/\alpha^3)$ as before.
\end{proofof}

\subsection{Part 3 of the Algorithm.}

Recall that, so far, we have computed a core  structure $(K,U(K),\graph^K,W^K)$ in graph $\graph$, where $V(W^K)=K$, $K\cup U(K)=V(\graph)$, and $\graph^K=\graph$. Moreover, we are guaranteed that $|U(K)|\leq 8\hn/\log n$, and with high probability, $W^K$ is a $4\alpha^*$-expander.
In this part, we modify this core structure to turn it into an $h$-core structure, with the desired properties.

Our first step is to compute an arbitrary maximal matching $M$ between the vertices of $K$ and the vertices of $U(K)$ in graph $\graph$.
Consider now the graph $W'$, obtaining by taking the union of the graph $W^K$, and the matching $M$. It is easy to verify that, if $W^K$ is a $4\alpha^*$-expander, then $W'$ is an $\alpha^*$-expander. Graph $W'$ will be the final witness graph for the core structure that we are constructing.
For every edge $e\in E(W')$, if $e\in W^K$, then its embedding $P(e)$ remains the same; otherwise, $e\in M$, and we embed $e$ into itself.

 The set $K'$ of the core vertices is defined as follows: it contains every vertex $v\in \Upsilon$, such that $v\in V(W')$, and at least $h/(64\log n)$ neighbors of $v$ in graph $\graph$ belong to $V(W')$. Set $U(K')$ contains all vertices of $\graph$ that do not lie in $K'$. The graph $G^{K'}$ associated with the core structure remains $\graph$. By appropriately setting the constant $c^*$ in the definition of our parameters, it is immediate to verify that, provided that $K'\neq \emptyset$ and $|U(K')|\leq |K'|$, we obtain a valid $h$-core structure $\kset'=(K',U(K'),\graph^{K'},W')$, with $K'\cup U(K')=V(\graph)$ and $K'\subseteq \Upsilon$. Moreover, with high probability, $\kset'$ is a perfect $h$-core structure. 
 The following claim is central to the analysis of this part.
 
 \begin{claim}\label{claim: k is good}
 Set $K'$ contains all but at most $16\hn/\log n$ vertices of $\Upsilon$.\end{claim}

\begin{proof}
Let $S=\Upsilon\setminus K'$. We partition $S$ into two subsets: set $S'$ contains all vertices that lied in $U(K)$; in other words, these vertices did not serve as vertices of $W^K$. From Theorem~\ref{thm: fast expander construction from large balanced cut}, $|S'|\leq 8\hn/\log n$. Let $S''=S\setminus S'$ be the set of the remaining vertices of $S$. Recall that each vertex of $\Upsilon$ has at least $h/(32\log n)$ neighbors in $\graph$, but each vertex $v\in S''$ has fewer than $h/(64\log n)$ neighbors in $\graph$ that lie in $V(W')$. Therefore, for each vertex $v\in S''$, there is a set $N'(v)$ of at least $h/(64\log n)$ vertices, such that each vertex $u\in N'(v)$ is a neighbor of $v$ in $\graph$, but $u\not\in V(W')$, and in particular $u\in U(K)$. Since the matching $M$ that we have computed before is a maximal matching, every vertex in $S''$ must have an edge of $M$ incident to it. But then the other endpoint of that edge lies in $U(K)$, and $|U(K)|\leq 8\hn/\log n$, so $|S''|\leq 8\hn/\log n$ must hold. Altogether, $|S|\leq |S'|+|S''|\leq 16\hn/\log n$, as required.	
\end{proof} 

Since we are guaranteed that $|\Gamma|\leq \hn/4$, we get that $K'\neq \emptyset$, and $|U(K')|\leq |\Gamma|+16\hn/\log n<\hn/2\leq |K'|$.

 The running time of the algorithm is dominated by the first two parts of the algorithm, and is bounded by $O\left(\left(|E(\graph)|+|V(\graph)|^{1+o(1)}\right)\poly\log n\right)$. This completes the proof of Theorem~\ref{thm: partition or core}.

%% file: max-flow.tex
\section{Maximum $s$-$t$ Flow and Minimum $s$-$t$ Cut in Undirected Vertex-Capacitated Graphs}\label{sec: max flow}

In this section we present randomized algorithms to compute a $(1+\eps)$-approximate maximum $s$-$t$ flow
and a $(1+\eps)$-approximate minimum $s$-$t$ cut in a simple undirected graph with vertex capacities, whose expected running time is  ${O}(n^{2+o(1)}/\eps^{O(1)})$, thus proving Theorem~\ref{thm: informal max s-t flow and min s-t cut}. 
We first describe in detail the algorithm for computing the $(1+\eps)$-approximate maximum $s$-$t$ flow, and then briefly sketch the ideas for approximately computing minimum $s$-$t$ cut.

\subsection{Maximum Vertex-Capacitated $s$-$t$ Flow}

In the maximum vertex-capacitated $s$-$t$ flow problem, we are given an undirected graph $G=(V,E)$ with capacities $c(v)> 0$ for vertices $v\in V$, together with two special vertices $s$ and $t$. The goal is to compute a maximum flow $f$ from $s$ to $t$, such that every vertex $v$ carries at most $c(v)$ flow units. We will also use a variation of this problem, where the graph is directed, and the capacities are on the edges and not vertices of the graph, which is defined similarly.

We follow the primal-dual framework for fast computation of approximate multicommodity flow of~\cite{GK98, Fleischer00}. Our algorithm, which is an analogue of the algorithms of~\cite{GK98, Fleischer00} for vertex capacities, maintains a length function $\ell$ on the vertices of the graph, that are iteratively updated, as the flow the algorithm computes is augmented. For every path $P$ in the graph, we denote by $\ell(P)$ its length with respect to the current vertex-lengths $\ell(v)$. 

We start by establishing that Algorithm~\ref{VertexFPTAS} below computes a $(1 + O(\eps))$-approximate maximum $s$-$t$ flow, and then use our Decremental SSSP data structure to obtain an efficient implementation of this algorithm. 

\begin{algorithm}[H]
\label{VertexFPTAS}
\SetAlgoLined
{\bf Input:} An undirected graph $G=(V,E)$ with vertex capacities $c(v)>0$, a source $s$, a sink $t$, and an accuracy parameter $0<\epsilon\leq 1$. \\
{\bf Output:} A feasible $s$-$t$ flow $f$.\\
~~~\\
Set $\delta = \frac{1+\eps}{ ((1+\eps)n)^{1/\eps}}$ and $R = \lfloor  \log_{1+\eps} \frac{1+\eps}{\delta} \rfloor$\;
 Initialize $\ell(v) = \delta$ for every vertex $v\in V$, $f \equiv 0$\;
 $P \leftarrow$ a $(1+\eps/3)$-approximate shortest $s$-$t$ path using the length function $\ell$\;
 \While{ $\ell(P) <  \min\{ 1, \delta (1+\eps)^R \}  $} {
 let $c$ be the smallest capacity of an internal vertex of $P$\;
 $f(P) \leftarrow f(P) + c$ ;  ~~~~~/* Augment the flow $f$ along the path $P$.*/ \\
 For each internal vertex $ v \in P, \ell(v) \leftarrow \ell(v)\left(1 + \frac{\eps c}{c(v)} \right)$ \;
  $P \leftarrow$ a $(1+\eps/3)$-approximate shortest $s$-$t$ path using length function $\ell$ on vertices\;
  }
 \smallskip
 \Return{solution $f$ scaled down by a factor of $\log_{1+\eps} \frac{1+\eps}{\delta}$.}
 \medskip
 \caption{An FPTAS for maximum $s$-$t$ flow in simple undirected vertex-capacitated graphs}
\end{algorithm}

\begin{theorem}
\label{thm:vertexFPTAS}
Given an undirected graph $G=(V,E)$ with vertex capacity function $c$, a source $s$, a sink $t$,  and an accuracy parameter $0<\epsilon \leq 1$, Algorithm~\ref{VertexFPTAS} returns a $(1+ 4\eps)$-approximate maximum $s$-$t$ flow. Moreover, the number of augmentation steps is bounded by $\kappa = O(n \log_{1+\eps} \frac{1+\eps}{\delta}) = O((n \log n)/\eps^2)$. 
\end{theorem}

\begin{proof}
We prove the theorem by coupling the execution steps of Algorithm~\ref{VertexFPTAS} to the execution steps of a similar algorithm of Fleischer~\cite{Fleischer00} for approximate multicommodity flow in edge-capacitated {\em directed} graphs. Algorithm~\ref{EdgeFPTAS} presents the algorithm of~\cite{Fleischer00} for maximum $s$-$t$ flow for directed edge-capacitated graphs. We note that the algorithm of~\cite{Fleischer00} is written for the more general multicommodity flow problem, while the version presented here is its restriction  to the single-commodity case. It was shown in~\cite{Fleischer00}  that the algorithm returns a $(1+4\eps)$-approximate $s$-$t$ flow upon termination. 

Given an undirected graph $G=(V,E)$ with capacities on vertices, we use a standard transformation to turn it into a directed graph $G'=(V',E')$ that will serve as input to Algorithm~\ref{EdgeFPTAS}, as follows. Let $C$ denote maximum vertex capacity in $G$, and let $C' = n^3C/\eps$.
For each vertex $v$, we add two vertices $v^-$ and $v^+$ to $V'$. 
For each edge $(u,v)$ in $G$ where $u, v \in V$, we add the edges $(u^+,v^-)$ and $(v^+,u^-)$ to $E'$, each with capacity $C'$ -- we refer to these edges as {\em special edges}. 
Now for each vertex $v \in V$, we add to $E'$ an edge $(v^-,v^+)$ of capacity $c(v)$ -- we refer to these edges as {\em regular edges}. Let $s' = s^{-}$ and $t' = t^{+}$. 

Note that this is the standard reduction from undirected vertex-capacitated maximum flow to the directed edge-capacitated maximum flow, and it is well known that the value of the maximum  (vertex-capacitated) $s$-$t$ flow in $G$ is equal to the value of the maximum (edge-capacitated) $s'$-$t'$ flow in $G'$. In particular, any $s$-$t$ path in $G$, say, $s \rightarrow v_{i_1} \rightarrow v_{i_2} \rightarrow \ldots \rightarrow v_{i_k} \rightarrow t$, can be naturally mapped to the corresponding $s'$-$t'$ path in $G'$, namely, $s' \rightarrow s^{+} \rightarrow v_{i_1}^{-} \rightarrow  v_{i_1}^{+}  \rightarrow v_{i_2}^{-} \rightarrow\ldots \rightarrow v_{i_k}^{+}  \rightarrow t^{-} \rightarrow t'$. Conversely, any $s'$-$t'$ path in $G'$ can be mapped to a corresponding $s$-$t$ path in $G$ in a similar manner. 
Finally, note that each augmentation in $G$ increase the length of some vertex by a factor of $(1+\eps)$. It follows that after $\kappa =O(n \log_{1+\eps} \frac{1+\eps}{\delta}) = O((n \log n)/\eps^2)$ augmentations in $G$, the length of every vertex is at least $1$ and the algorithm terminates. 

We now argue that, using the $1$-$1$ mapping between the $s$-$t$ paths in $G$ and the $s'$-$t'$ paths in $G'$ described above, we can couple together the executions of Algorithm~\ref{VertexFPTAS} and Algorithm~\ref{EdgeFPTAS}. In particular, whenever Algorithm~\ref{VertexFPTAS} updates the $s$-$t$ flow along a path $P$ in $G$, we will apply the update along the corresponding $s'$-$t'$ path $P'$ in $G'$ in Algorithm~\ref{EdgeFPTAS}.  This ensures that at all times during the coupled execution of the two algorithms, for each regular edge $e = (v^-, v^+)$ in $G'$, we have $\ell'(e) = \ell(v)$.

While during the coupled executions of Algorithm~\ref{VertexFPTAS} and Algorithm~\ref{EdgeFPTAS}, we maintain the invariant that the regular edges in $G'$ have the same length as the corresponding vertices in $G$, we additionally have special edges in $G'$ which have positive length and no analog in the graph $G$. This creates a difficulty in directly establishing that any approximate $s$-$t$ shortest path $P$ in $G$ corresponds to an almost equally good  $s'$-$t'$ shortest path $P'$ in $G'$.
 Thus in order to facilitate this coupling, we make slight modifications in the graphs $G$ and $G'$, and analyze the coupled executions on these modified graphs. We attach a path $\Gamma$ of length $3n/\eps$ in $G$ to the vertex $s$ with each vertex on the path assigned a capacity of $C'$, and set the source $s$ to be the endpoint of this path that has degree $1$ in the resulting graph. We make an analogous transformation in the graph $G'$ and attach a directed path $\Gamma'$ of length $3n/\eps$ in $G'$ to the vertex $s^{-}$ where each edge $e$ on $\Gamma'$ is assigned a capacity of $C'$, and set the source $s'$ to be the endpoint of this path that has degree $1$ in the resulting graph. 
It is easy to see that these modifications neither alter the value of the maximum $s$-$t$ vertex-capacitated flow in the graph $G$ nor alter the value of the maximum $s'$-$t'$ edge-capacitated flow in the graph $G'$. Moreover, the $1$-$1$ mapping between the $s$-$t$ paths in $G$ and the $s'$-$t'$ paths in $G'$ described above continues to hold with the prefix $\Gamma$ on each $s$-$t$ path in $G$ mapping to $\Gamma'$ and vice versa.

We are now ready to establish the coupling between the execution of Algorithm~\ref{VertexFPTAS} on (modified) graph $G$ and the execution of Algorithm~\ref{EdgeFPTAS} on (modified) graph $G'$.

To complete the proof of this theorem, we use two simple claims below.

\begin{claim}
\label{claim_maxflow_special_edge}
During the execution of Algorithm~\ref{EdgeFPTAS}, for each special edge $e$ in $G'$, we have $\ell'(e) \le 2\delta$.
\end{claim}
\begin{proof}
Initially, for any special edge $e$ in $G'$, we have $\ell'(e) = \delta$. Since any $s'$-$t'$ path in $G'$ goes through at least one regular edge, and the capacity of each regular edge is bounded by $C$, after any augmentation, $\ell'(e')$ increases by a factor of at most  $\left(1 + \frac{\eps C}{C'} \right)$ which is bounded by $\left(1 + \frac{\eps^2}{n^3} \right)$. Thus after any sequence of $\kappa$ augmentations, 
$$\ell'(e) \le \delta \left(1 + \frac{\eps^2}{n^3} \right)^\kappa \le \delta e^{\frac{\eps^2 \kappa}{n^3}} \le \delta e^{\frac{1}{n}} \le 2\delta,$$
where the penultimate inequality follows from the fact that $\kappa \le n^2/\eps^2$ for sufficiently large $n$.
\end{proof}

We now use Claim~\ref{claim_maxflow_special_edge} to establish the following.

\begin{claim}
During the execution of Algorithm~\ref{VertexFPTAS}, if $P$ is any $(1+\eps/3)$-approximate shortest $s$-$t$ in $G$, and $P'$ is its corresponding  $s'$-$t'$ path in $G'$, then $P'$ is a $(1+\eps)$-approximate shortest $s'$-$t'$ path in $G'$.
\end{claim}
\begin{proof}
Suppose not. Then during the first $\kappa$ augmentations, there exists a $(1+\eps/3)$-approximate $s$-$t$ path $P$ in $G$ such that the corresponding path $P'$ is not a $(1+\eps)$-approximate $s'$-$t'$ path in $G'$. Thus there exists a $s'$-$t'$ path $Q'$ in $G'$ such that $\ell'(P') > (1+\eps) \ell'(Q')$. Let $Q$ be the $s$-$t$ path in $G$ that corresponds to $Q'$. Then $\ell(Q) \le \ell'(Q')$. On the other hand,
$\ell'(P') \le \ell(P) + n(2\delta)$ since any $s'$-$t'$ (simple) path contains at most $n$ special edges, and by Claim~\ref{claim_maxflow_special_edge}, the length of each special edge is bounded by $2\delta$ during the first $\kappa$ augmentations. It the follows that:

$$ \ell(P) + 2\delta n \ge \ell'(P') > (1+\eps) \ell'(Q') \ge (1+\eps) \ell(Q),$$
giving us the inequality $\ell(P) + 2\delta n > (1+\eps) \ell(Q)$. Rearranging terms, we get 

$$ \ell(P) > \left(1 + \frac{\eps}{3} \right)\ell(Q) + \frac{2}{3} \eps \ell(Q) - 2\delta n.$$
But any $s$-$t$ path $Q$ in $G$ has length at least $\delta |\Gamma|$ (recall that $|\Gamma| = 3n/\eps$) so 

$$  \frac{2}{3} \eps \ell(Q) - 2\delta n \ge  \frac{2}{3} \eps \delta (3 n/\eps) - 2\delta n \ge 0,$$
which implies $ \ell(P) > (1 + \eps/3)\ell(Q)$. But this is a contradiction to our assumption that $P$ is a $(1+\eps/3)$-approximate $s$-$t$ path in $G$.
\end{proof}

We are now ready to complete the proof of the theorem. By the above claim, any
$(1+\eps/3)$-approximate $s$-$t$ path $P$ in $G$ corresponds to a $(1+\eps)$-approximate $s'$-$t'$ path $P'$ in $G'$. Thus via this coupling, we are running Algorithm~\ref{EdgeFPTAS} on the graph $G'$, and hence obtain a $(1 + 4\eps)$-approximate flow in $G'$ upon termination. This, in turn, gives us a $(1 + 4\eps)$-approximate flow in $G$.
\end{proof}

\medskip

\begin{algorithm}[H]
\label{EdgeFPTAS}
\SetAlgoLined
{\bf Input:} A directed graph $G'=(V',E')$ with edge capacities $c'(e)$, a source $s'$, a sink $t'$, and an accuracy parameter $0<\epsilon\leq 1$. \\
{\bf Output:} A feasible $s'$--$t'$ flow $f'$.\\
~~~\\
Set $\delta = \frac{1+\eps}{ ((1+\eps)n)^{1/\eps}}$ and $R = \lfloor  \log_{1+\eps} \frac{1+\eps}{\delta} \rfloor$\;
 Initialize $\ell'(e) = \delta ~\forall~e \in E'$, $~f' \equiv 0$\; 
  $P' \leftarrow$ a $(1+\eps)$-approximate shortest $s'$-$t'$ path using the length function $\ell'$ on edges\;
 \While{ $\ell'(P) < \min\{ 1, \delta (1+\eps)^R \} $} {
 $c' \leftarrow \min_{e \in P'} c'(e)$ \;
 $f'(P') \leftarrow f'(P') + c'$ ;  ~~~~~/* Augment the flow $f'$ along the path $P'$.*/ \\
 $\forall e \in P, \ell'(e) \leftarrow \ell'(e)\left(1 + \frac{\eps c'}{c'(e)} \right)$ \;
  $P' \leftarrow$ a $(1+\eps)$-approximate shortest $s$-$t$ path using the length function $\ell'$ on edges\;
      }
 \smallskip
 \Return{solution $f'$ scaled down by a factor of $\log_{1+\eps} \frac{1+\eps}{\delta}$.}
 \medskip
 \caption{An FPTAS for maximum $s'$-$t'$ flow in edge-capacitated directed graphs}
\end{algorithm}

\medskip

We next describe an efficient implementation of Algorithm~\ref{VertexFPTAS}. We note that for any $\eps \in (0,1]$, a $(1+\eps/3)$-approximate shortest $s$-$t$ path in the original graph $G$ gives a $(1+\eps/3)$-approximate shortest $s$-$t$ path in the modified graph $G$ as well; so it is enough to focus on computing $(1+\eps/3)$-approximate shortest $s$-$t$ path in the original graph $G$. 

\begin{theorem}
\label{thm: max s-t flow}
There is a randomized algorithm, that, given a simple undirected graph $G=(V,E)$ with vertex capacities $c(v)>0$, a source $s$, a sink $t$,  and an accuracy parameter $\epsilon \in (0,1]$, computes 
 a $(1+ 4 \eps)$-approximate maximum  $s$-$t$ flow in ${O}(n^{2+o(1)}/\eps^{O(1)})$ expected time.
\end{theorem}

\begin{proof}
By Theorem~\ref{thm:vertexFPTAS}, it suffices to show that Algorithm~\ref{VertexFPTAS} can be implemented in 
${O}(n^{2+o(1)}/\eps^{O(1)})$ expected time. 
Let $\delta$ be as defined in Algorithm~\ref{VertexFPTAS}, and let $K = \lfloor \log_{1+ \frac{\eps}{9}} \frac{1+\eps}{\delta} \rfloor$.
We start by creating a new graph $H_1 =(V_1,E_1)$ where $V_1$ contains vertices $s, t$, and for each vertex $v \in V \setminus \{s,t\}$, we add a copy $(v,i)$ for $i \in \{ 0, \ldots, K \}$ to $V_1$. 
For each edge $(u,v)\in E$ with $u, v \in V \setminus \{ s,t \}$, we add an edge between the vertices 
$(u,i)$ and $(v,j)$ for all $i, j \in \{ 0, \ldots, K \}$ to $H_1$. For each edge $(s,u)$ in $G$, we add to $H_1$ an edge between vertices $s$ and $(u,i)$ for all $i \in \{ 0, \ldots, K \}$.
Finally, for each edge $(v,t)$ to $G$, we add in $H_1$ an edge between vertices $(v,i)$ and $t$ for all 
$i \in \{ 0, \ldots, K \}$.

 We next define a length function $\ell_1$ on the vertices of $H_1$ as follows. We let $\ell_1(s) = \ell_1(t) = 0$, and for each vertex $(v, i)$, we let $\ell_1((v,i)) = \delta \left(1 + \frac{\eps}{9} \right)^i$. Note that the smallest non-zero  vertex-length is $\delta$, and the largest non-zero vertex-length is $(1 + \eps)$.

We now show that we can execute Algorithm~\ref{VertexFPTAS} on the graph $G$ using the Decremental SSSP data structure on the graph $H_1$ with source $s$ and accuracy parameter $\eps/9$.
Our implementation will use $O((n \log n)/\eps^2)$ path queries (same as the number of augmentations needed by Algorithm~\ref{VertexFPTAS}). 
In the remainder of the proof, we describe how the Decremental SSSP data structure allows us to implement an oracle that returns a $(1+\eps)$-approximate shortest $s$-$t$ path in $G$ in $\tilde{O}(n/\eps^{O(1)})$ expected time. Since Theorem~\ref{thm: main for SSSP} works on graphs with a length function defined over edges instead of vertices, we first create a new graph $H_2 =(V_2,E_2)$ which is identical to $H_1$, that is, $V_2 = V_1$ and $E_2 = E_1$, but has a length function $\ell_2$ that is defined over its edges instead of vertices.
 For each edge $(u,v)$ in $E_2$ we define $\ell_2(u,v) = (\ell_1(u) + \ell_1(v))/2$. Since $\ell_1(s) = \ell_1(t) = 0$, it is easy to verify that for any $s$-$t$ path $P$ in $H_1$, its length under the vertex-length function $\ell_1$ is same as its length under the edge-length function $\ell_2$. 
 
Consider the step requiring the computation of the $(1+\eps/3)$-approximate shortest $s$-$t$ path in the original $G$ during some iteration of the 
Algorithm~\ref{VertexFPTAS}. We invoke $\pquery(t)$ in $H_2$, and let $P_2$ be the path returned. If length of $P_2$ is at least $\min\{ 1, \delta (1+\eps)^R \}$, we terminate the execution of the algorithm. Otherwise, we use $P_2$ to compute an $s$-$t$ path $P$ in $G$ as follows: we replace each vertex of the form $(v,i)$ on the path $P_2$ with the vertex $v$. 
We can assume here w.l.o.g. that for any vertex $(v,i)$ on the path $P_2$, $i$ is the smallest integer for which the graph $H_2$ contains a copy of the vertex $v$ (otherwise, we can replace with another copy of $v$ to ensure this property). 
We now update in $G$ the length function $\ell(v)$ for each internal vertex on the path $P$ using the update rule of Algorithm~\ref{VertexFPTAS}. For any internal vertex $(v,i)$ on the path $P_2$, let $j$ be the smallest integer such that the updated length $\ell(v) \le \delta \left(1 + \frac{\eps}{9} \right)^j$. If $i<j$, then we delete from $H_2$ vertices $(v,i), (v,i+1), \ldots, (v,j-1)$. In other words, we maintain the invariant that, when $\ell(v) \in [ \delta \left(1 + \frac{\eps}{9} \right)^{j-1}, \delta \left(1 + \frac{\eps}{9} \right)^{j} )$, the smallest available length for a copy of vertex $v$ in $H_2$ is $ \delta \left(1 + \frac{\eps}{9} \right)^j$. Consequently, the graph $H_2$ inflates $s$-$t$ path lengths in $G$ by at most a factor of $(1+ \frac{\eps}{9})$. Thus a $(1 + \frac{\eps}{9})$-approximate $s$-$t$ path in $H_2$ is guaranteed to be a $(1+\frac{\eps}{9})^2$-approximate shortest path in $G$. Note that, since $\eps \le 1$, we have $(1+ \frac{\eps}{9})^2 \le (1 + \eps/3)$, and hence this is equivalent to running Algorithm~\ref{VertexFPTAS} with accuracy parameter $\eps/3$. The resulting  maximum $s$-$t$ flow is thus $(1 + 4 \eps)$-approximate as desired.

We conclude by analyzing the total time spent on maintaining the Decremental SSSP data structure and answering all path queries. 
Since the graph $H_2$ contains $\tilde{O}(n/\eps^2)$ vertices and $L = O(1/\delta)$, by Theorem~\ref{thm: main for SSSP}, total expected time spent on processing all vertex deletions can be bounded by ${O}(n^{2+o(1)} /\eps^{O(1)})$. 
Similarly, total  expected time spent on answering all $O((n \log n)/\eps^2)$ path queries is bounded by ${O}(n^{2+o(1)} /\eps^{O(1)})$, completing the proof.
\end{proof}

\subsection{Minimum $s$-$t$ Cut}

In the minimum vertex-capacitated $s$-$t$ cut problem, we are given a simple undirected graph $G=(V,E)$ with capacities $c(v)> 0$ for vertices $v\in V$, together with two special vertices $s$ and $t$. The goal is to find a smallest capacity subset $X \subseteq V \setminus \{s,t\}$ of vertices whose deletion disconnects $s$ from $t$. 

The fractional relaxation of this problem asks for a length function $\ell$ defined over the  vertices of $G$ such that every $s$-$t$ path has length at least $1$ under the function $\ell$, and the cost $\sum_v c(v) \ell(v)$ is minimized.
The fractional minimum $s$-$t$ cut problem is the dual of the maximum $s$-$t$ flow problem, and we start by
observing that Algorithm~\ref{VertexFPTAS} can also be used to compute a near-optimal fractional solution for minimum $s$-$t$ problem. 

We once again rely on Fleischer's analysis of Algorithm~\ref{EdgeFPTAS}~\cite{Fleischer00}, and the coupling between Algorithm~\ref{VertexFPTAS} and Algorithm~\ref{EdgeFPTAS} as described in the proof of Theorem~\ref{thm:vertexFPTAS}.
Let $\ell_i$ denote the vertex-length function in iteration $i$ of the while loop of Algorithm~\ref{VertexFPTAS}, and let $\alpha(i)$ denote the shortest $s$-$t$ path length under the length function $\ell_i$. Then the fractional solution $ \ell_i/\alpha(i)$ clearly has the property that the shortest $s$-$t$ path length is at least $1$. 
Similarly, let $\ell'_i$ denote the edge-length function in iteration $i$ of the while loop of Algorithm~\ref{EdgeFPTAS}, and let $\alpha'(i)$ denote the shortest $s'$-$t'$ path length under the length function $\ell'_i$. Then the fractional solution $ \ell'_i/\alpha'(i)$ clearly has the property that the shortest $s'$-$t'$ path length is at least $1$.
  The analysis of~\cite{Fleischer00} shows that the solution $\min_{i} \set{\ell'_i/\alpha'(i)}$ is within a factor $(1 + 4\eps)$ of the optimal fractional $s'$-$t'$ cut in $G'$. Using the coupling between Algorithm~\ref{VertexFPTAS} and Algorithm~\ref{EdgeFPTAS}, we can also conclude that $\min_{i}\set{ \ell_i/\alpha(i)}$ is a $(1 + 4\eps)$-approximate fractional $s$-$t$ cut in $G$.
  
  Thus to compute a $(1 + 4\eps)$-approximate fractional $s$-$t$ cut in $G$, it suffices to take the solution with
 a minimum $( \sum_v c(v) \ell_i(v) )/\alpha(i)$ ratio over all iterations of the while loop in Algorithm~\ref{VertexFPTAS}. We now show how to efficiently track this value during the execution of  Algorithm~\ref{VertexFPTAS}, building on the implementation given in Theorem~\ref{thm: max s-t flow}.
  Let $D(i) =  \sum_v c(v) \ell_i(v) $. The quantity $D(0) = n \delta$. After iteration $i$,  
the quantity $D(i)$ is same as $D(i-1)$ except for the contribution of the vertices on the augmenting path $P$ used in iteration $i$. We can update $D(i)$ to reflect this change in $O(n)$ time. Thus total time taken to maintain $D(i)$ over all $O(n \log n/\eps^2)$ iterations is $O(n^2 \log n/\eps^2)$. Finally, the quantity $\alpha(i)$ is computed approximately to within a factor of $(1 + \eps)$ in every iteration. So we can compute a solution that minimizes $D(i)/\alpha(i)$ to within a factor of $(1+\eps)$, giving us a fractional $s$-$t$ cut solution that is within a factor $(1+\eps)(1+4\eps) \le (1 + 6\eps)$ (assuming $\eps \le 1$) of the optimal fractional cut.

Let $\ell^*$ be the length function chosen by the above process. The final step of our algorithm is to convert the fractional $s$-$t$ cut solution defined by $\ell^*$ to an integral solution.
We can use the standard random threshold rounding where we choose a random radius $r \in [0,1)$, and grow a ball of radius $r$ around $s$ using the length function $\ell^*$. Any vertices that intersect the boundary of this ball are placed in the separator $X$. It is easy to see that the expected cost of this solution is 
$\sum_v c(v) \ell^*(v)$. 

Putting everything together, we obtain the following theorem.

\begin{theorem}
\label{thm: min s-t cut}
There is a randomized algorithm, that, given a simple undirected graph $G=(V,E)$ with vertex capacities $c(v)>0$, a source $s$, a sink $t$,  and an accuracy parameter $\epsilon \in (0,1]$,  computes 
 a $(1+ 6 \eps)$-approximate  $s$-$t$ cut in ${O}(n^{2+o(1)}/\eps^{O(1)})$ expected time.
\end{theorem}

%% file: sparsest-cut.tex
\section{Vertex Sparsest Cut in Undirected Graphs}
\label{sec: sparsest cut}

In this section we present a proof of Theorem~\ref{thm: informal sparsest cut}. We use the standard definition of vertex sparsest cut.
A vertex cut in a graph $G$ is a partition $(A,X,B)$ of its vertices, so that there is no edge connecting $A$ to $B$ (note that we allow $A=\emptyset$ and $B=\emptyset$). The sparsity of the cut $(A,X,B)$ is $\psi(A,X,B)= \frac{|X|}{\min\set{|A|,|B|} + |X|}$. The goal of the vertex sparsest cut problem is to compute a vertex cut of minimum sparsity in $G$. We denote by $\psi(G)$ the minimum sparsity of any vertex cut in $G$. Note that, if $G$ is a connected graph, then $\psi(G)$ is at least $\Omega(1/n)$ and at most $1$. 

We design an algorithm that either produces a cut $(A,X,B)$ with $\psi(A,X,B) =O(\log^4n)\cdot \psi(G)$, or determines that $\psi(G)=\Omega(1/\log^4n)$. In the latter case, we output the vertex cut $(A,X,B)=(\emptyset, V , \emptyset)$  of sparsity $\psi(A,X,B)=1$ as a trivial solution, thus obtaining an $O(\log^4 n)$-approximate solution.  The expected running time of the algorithm is $n^{2+o(1)}$. The following theorem summarizes our main subroutine.

\begin{theorem}\label{thm: sparsest w alpha given}
There is a randomized algorithm, that, given a simple undirected $n$-vertex graph  $G=(V,E)$ and a target sparsity value  $0<\alpha \leq 1$, either computes a vertex cut $(A,X,B)$ with $\psi(A,X,B) = O(\alpha)$, or with high probability correctly certifies that $\psi(G) = \Omega(\alpha/\log^4 n)$. The expected running time of the algorithm is $n^{2+o(1)}$.
\end{theorem}

We prove Theorem~\ref{thm: sparsest w alpha given} below, after we complete the proof of Theorem~\ref{thm: informal sparsest cut}  using it.
We run the algorithm from Theorem~\ref{thm: sparsest w alpha given} for all sparsity values $\alpha_i=2^i/n$, for $1\leq i\leq \log n$. If, for any $1\leq i\leq \log n$, the algorithm computes a cut $(A,X,B)$ with $\psi(A,X,B)\leq O(\alpha_i)$, then we let $i^*$ be the smallest index for which the algorithm returns such a cut, and we output the corresponding cut $(A,X,B)$.
Otherwise, we are guaranteed that $\psi(G)\geq \Omega(1/\log^4n)$, and we output the cut $(\emptyset, V , \emptyset)$.

From now on we focus on proving Theorem~\ref{thm: sparsest w alpha given}. We assume that we are given the value $0<\alpha\leq 1$.
%
Our algorithm implements the cut-matching game of~\cite{KRV} (see Theorem~\ref{thm:cut-matching-game}). The main tool that we will use is the following lemma.

\begin{lemma}\label{lem: KRV_routing}
There is a randomized algorithm, that, given a simple undirected $n$-vertex graph $G=(V, E)$, a parameter $0<\alpha\leq 1$, and two equal-cardinality disjoint subsets $A, B \subseteq V$ of its vertices, returns one of the following:

\begin{itemize}
\item  either a vertex cut  $(Y,X,Z)$ in $G$ with  $\psi(Y,X,Z)=O(\alpha)$; 

\item or a set $\qset$ of at least $|A|/100$ paths in $G$, connecting vertices of $A$ to vertices of $B$, such that the paths in $\qset$ have distinct endpoints, and every vertex in $G$ participates in at most 
$O(\log n/\alpha)$ paths in $\qset$. \end{itemize}
The expected running time of the algorithm is  ${O}(n^{2+o(1)})$.
\end{lemma}

We prove Lemma~\ref{lem: KRV_routing} below, after completing the proof of Theorem~\ref{thm: sparsest w alpha given} using it. We employ the cut-matching game from Theorem~\ref{thm:cut-matching-game}. Recall that the game lasts for $O(\log^2n)$ iterations, that we call phases. We start with a graph $W$, whose vertex set is $V(W)=V(G)$, and $E(W)=\emptyset$. For all $1\leq i\leq O(\log n)$,  in the $i$th phase, we use the algorithm from Theorem~\ref{thm:cut-matching-game} to compute two disjoint equal-cardinality subsets $(A_i,B_i)$ of $V(G)$ in time $O(n\poly\log n)$. We then attempt to find a collection $\pset_i$ of paths that connect every vertex of $A_i$ to a distinct vertex of $B_i$ in $G$. In order to do so, we start with $\pset_i=\emptyset$. While $A_i,B_i\neq \emptyset$, we iteratively apply Lemma~\ref{lem: KRV_routing} to graph $G$ and vertex sets $A_i,B_i$. If the outcome of the lemma is a vertex cut  $(Y,X,Z)$ in $G$ with  $\psi(Y,X,Z)=O(\alpha)$, then we terminate the algorithm and return this cut. Otherwise, we obtain a set $\qset$ of at least $|A_i|/100$ paths in $G$, connecting vertices of $A_i$ to vertices of $B_i$, such that the paths in $\qset$ have distinct endpoints, and every vertex in $G$ participates in at most 
$O(\log n/\alpha)$ paths in $\qset$. We add the paths of $\qset$ to $\pset_i$, and we discard from $A_i$ and $B_i$ the endpoints of the paths in $\qset$. Notice that after $O(\log n)$ iterations, if the algorithm does not terminate with a vertex cut of sparsity $O(\alpha)$, then we obtain a set $\pset_i$ of paths connecting every vertex of $A_i$ to a distinct vertex of $B_i$, such that every vertex of $G$ participates in at most $O(\log^2n/\alpha)$ paths in $\pset_i$.  The paths in $\pset_i$ naturally define a matching $M_i$ between the vertices of $A_i$ and the vertices of $B_i$. We add the edges of $M_i$ to $W$. This completes the description of the $i$th phase. If our algorithm does not terminate with a vertex cut of sparsity $O(\alpha)$, then the final graph $W$ is a $\half$-expander with high probability. Moreover, for every edge $e\in E(W)$, we have computed a path $P_e$ connecting its endpoints in $G$, such that the paths in $\set{P_e\mid e\in E(W)}$ cause vertex-congestion $O(\log^4n/\alpha)$ in $G$.

We claim that, if $W$ is a $\half$-expander, then $\psi(G)=\Omega(\alpha/\log^4n)$. Indeed, consider any vertex cut $(Y,X,Z)$ in $G$, and assume w.l.o.g. that $|Y|\leq |Z|$. 
It is enough to show that $\psi(Y,X,Z)=\Omega(\alpha/\log^4n)$.
Note that if $Y = \emptyset$, then $\psi(Y,X,Z) = 1=\Omega(\alpha/\log^4n)$ since $\alpha\leq 1$, so it suffices to focus on the case where $Y, Z \neq \emptyset$. 
Let $E'$ be the set of edges leaving the set $Y$ in graph $W$, so $|E'|\geq |Y|/2$. We further partition $E'$ into two subsets: set $E_1=E_W(Y,X)$ and set $E_2=E_W(Y,Z)$. Since all vertex degrees in $W$ are $O(\log^2n)$, $|X|\geq \Omega(|E_1|/\log^2n)$ must hold. Since all paths in set $\set{P_e\mid e\in E_2}$ must contain a vertex of $X$, and the paths in $\set{P_e\mid e\in E(W)}$ cause vertex-congestion $O(\log^4n/\alpha)$ in $G$, we get that $|X|=\Omega(\alpha |E_2|/\log^4n)$. Altogether, we get that $|X|=\Omega(\alpha |E'|/\log^4n)=\Omega(\alpha|Y|/\log^4n)$, and so $\psi(Y,X,Z)=\Omega(\alpha/\log^4n)$.

It now remains to prove Lemma~\ref{lem: KRV_routing}. We construct a new simple undirected  vertex-capacitated graph $G'=(V',E')$ as follows. The vertex set $V'$ contains a copy of every vertex $v \in V$, called a {\em regular copy}, that is assigned capacity $c(v) = 1/\alpha$. For each vertex $v \in A \cup B$, we add a {\em special copy}, $v'$ that is assigned capacity $c(v) = 1$. Let $A'$ denote the set of special vertices corresponding to the vertices of $A$, and let  $B'$ denote the set of special vertices corresponding to the vertices of $B$.
The edge set $E'$ is constructed as follows. For every edge $(u,v) \in E$, we add the edge $(u,v)$ to $E'$. Additionally, for every vertex $u \in A$, we add the edge $(u',u)$ to $E'$, and similarly, for every vertex $v \in B$, we add the $(v,v')$ to $E'$.
This completes the description of the graph $G'$. 
Our goal now is to either compute a collection $\qset$ of $A'$-$B'$ paths in $G'$, whose endpoints are disjoint, and that cause vertex-congestion $O(\log n/\alpha)$, such that $|\qset|\geq |A|/100$, or find a cut of sparsity $O(\alpha)$ in the original graph $G$. We do so by employing the standard primal-dual approach for computing maximum multicommodity flow.

\paragraph{The Primal-Dual Framework.}

We now describe a primal-dual framework for computing integral $A'$-$B'$ flow in graph $G'$. This is similar in spirit to the primal-dual framework of~\cite{GK98, Fleischer00} used in proving Theorem~\ref{thm: informal max s-t flow and min s-t cut}, but it is better suited for our integral routing application (see also~\cite{AAP93}).

Let $\pset$ be the set of all paths connecting vertices of $A'$ to vertices of $B'$ in $G'$.
We use the following linear program and its dual.

\begin{tabular}[t]{|l|l|}\hline &\\
$\begin{array}{lll}
\text{\underline{Primal}}&&\\
\text{Max}&\sum_{P\in \pset} f(P)&\\
\text{s.t.}&&\\
&\sum_{\stackrel{P\in\pset:}{v\in P}}f(P)\leq
c(v)&\forall v\in V(G')\\
&f(P)\geq 0&\forall P\in \pset\\
\end{array}$
&$
\begin{array}{lll}
\text{\underline{Dual}}&&\\
\text{Min}&\sum_{v\in V(G')}c(v)x_v\\
\text{s.t.}&&\\
&\sum_{v\in P}x_v\geq 1&\forall P\in \pset\\
&x_v\geq 0&\forall v\in V(G')\\
&&\\
\end{array}$\\ &\\ \hline
\end{tabular}

Recall that, if we denote by $\opt$ the value of the optimal integral solution to the primal linear program, and by $\optp$ and $\optd$ the values of the optimal (fractional) solutions to the primal and the dual LP's, respectively, then for any feasible dual solution of cost $\costd$,

\[\opt\leq \optp \leq \optd \leq \costd.\]


The algorithm starts with an infeasible dual solution, where for every vertex $v\in V'$, we set $x_v=0$, and a path set $\qset=\emptyset$. Throughout the algorithm, we increase the vertex lengths $x_v$ and add paths to $\qset$. At every point in the algorithm, for every path $P\in \pset$, we define its length to be $\sum_{v\in P}x_v$.

We assume for now that  we are given an oracle, that, in every iteration, either produces a path $P\in \pset$, whose length is less than $1$, or certifies that every path in $P\in \pset$ has length at least $(1-\eps)$ for some $\eps \in (0,1/2)$. An iteration is executed as follows. If the oracle returns a path $P\in \pset$ of length less than $1$, then we add $P$ to $\qset$. Assume that $a\in A'$ and $b\in B'$ are the endpoints of $P$. We set $x_a=x_b=1$, and, for every regular vertex $v\in V'$ that lies on $P$, we update the value $x_v$ as follows. If $x_v=0$, then we set $x_v=1/n$; otherwise, we set $x_v = (1 + \alpha)x_v$. 
Notice that in every iteration, the number of the paths we route increases by $1$, while the cost of the dual solution value increases by at most $4$ (since $\alpha\leq 1$).  Therefore, if $\qset$ if the current set of paths that we have routed, then $|\qset|\geq  (\sum_{v \in V'} c(v)x_v)/4$ holds throughout the algorithm. Notice that all paths in $\qset$ are guaranteed to have distinct endpoints. The algorithm terminates once the oracle reports that every path in $\pset$ has length at least $(1-\eps)$. We now consider two cases.

The first case happens if $|\qset|\geq |A|/100$ at the end of the algorithm. In this case, we return a set $\qset'$ of paths, obtained from the set $\qset$ of paths by deleting, for every path $Q\in \qset$, its first and last vertex. This ensures that the paths of $\qset'$ are contained in the original graph $G$. Notice that their endpoints still remain disjoint. Moreover, it is easy to see that the paths in $\qset'$ cause vertex-congestion at most $O(\log n/\alpha)$ in $G$, since after $O(\log n/\alpha)$ updates of the form $x_v = (1 + \alpha)x_v$, the value $x_v > 1$,  and hence $v$ cannot be on any path $P\in \pset$ of length less than $1$.

 Assume now that $|\qset|<|A|/100$. We show a near-linear algorithm to compute a cut of sparsity at most $O(\alpha)$ in $G$. Recall that, upon the termination of the algorithm, every $A'$-$B'$ path has length at least $(1 - \eps)$. So if we scale up the value $x_v$ for every regular vertex $v$ by $1/(1-\eps) \le 2$, we obtain a feasible dual solution of value at most $8|\qset|$.
  
Let $A''\subseteq A'$ and $B''\subseteq B$ be the sets of vertices that do not serve as endpoints for paths in $\qset$. Notice that $|A''|,|B''|>  99|A|/100$,  and for all $v\in A''\cup B''$, $x_v=0$. Consider a new graph $G''$, obtained from $G'$, after we delete all vertices of $A'\setminus A''$ and $B'\setminus B''$ from it;  unify all vertices of $A''$ into a new source $s'$; and  unify all vertices of $B''$ into a new sink $t'$. The current values $x_v$ for regular vertices $v\in V'$ now define a feasible solution to the minimum $s'$-$t'$ cut LP in this new graph $G''$. The value of this $s'$-$t'$ cut LP solution is at most $8|\qset|\leq 8 |A|/100$, since $|\qset|\leq |A|/100$. 
  We can now once again use the standard random threshold rounding to recover, in  expected time $O(|V| + |E|)$, an integral $s'$-$t'$ vertex cut in $G''$ of cost at most $ |A|/10$. Let $X'$ be the set of vertices deleted in this $s'$-$t'$ cut -- note that all vertices in $X'$ are regular vertices. Since each vertex in $X'$ has capacity of $1/\alpha$, it must be that $|X'| \le  \alpha  |A|/10 $.
  
  We now construct a vertex cut of sparsity $O(\alpha)$ in the original graph $G$.  Let $\hat A\subseteq A$ be the set of vertices of $G$ corresponding to the vertex set $A''$, that is: $\hat A=\set{a\mid a'\in A''}$, and let $\hat B$ be defined similarly for $B''$. 
  By our construction of $G''$, graph $G\setminus X'$ contains no path connecting a vertex of $\hat A$ to a vertex of $\hat B$ (but it is possible that $\hat A\cap X',\hat B\cap X'\neq \emptyset$).  We let $Y$ be the union of all connected components of $G\setminus X'$ containing  vertices of $\hat A$, and we let $Z=V(G)\setminus(X'\cup Y)$. Our algorithm returns the cut $(Y,X',Z)$. We now analyze its sparsity. Recall that $|X'|\leq \alpha |A|/10$, while $|\hat A|,|\hat B|\geq 99|A|/100$. Since $\alpha\leq 1$, at least $0.8|A|$ vertices of $\hat A$ lie in $Y$, and similarly, at least $0.8|A|$ vertices of $\hat B$ lie in $Z$. Therefore, $|Y|,|Z|\geq 0.8|A|$, and the sparsity $\psi(Y,X',Z)\leq O(\alpha)$.
 
 Before we discuss the implementation of the oracle, we analyze the running time of the algorithm so far. The running time of the primal-dual part is bounded by the number of updates to the vertex lengths $x_v$. It is easy to verify that each such variable $x_v$ is updated at most $O(\log n/\alpha)$ times, and so the total running time of this part, excluding the time needed to respond to the oracle queries, is $O(n\log n/\alpha)=O(n^2\log n)$. The final step of computing the cut via the random threshold algorithm has expected running time $O(n+|E(G)|)$. 
 It now remains to show how to implement the oracle.
 
\paragraph{Implementing The Oracle using Vertex-Decremental SSSP}
We fix the parameter $\eps = 1/2$. The oracle is implemented using vertex-decremental SSSP in essentially identical fashion to the one given in the proof of Theorem~\ref{thm: max s-t flow}; and we omit repeating the details here. Recall that the total expected update time of the oracle is $O(n^{2+o(1)})$, and the expected query time is $O(n\poly\log n)$ per query.
The primal-dual algorithm above uses at most $|A| = O(n)$ path queries. Thus, as in the proof of Theorem~\ref{thm: max s-t flow}, total expected time taken to maintain the decremental SSSP data structure and answer all path queries is bounded by $n^{2+o(1)}$.
This completes the proof of Theorem~\ref{thm: informal sparsest cut}.

%% file: appendix.tex
\section{Proof of Theorem~\ref{thm: main for maintaining a light graph}} \label{apx: maintaining the light graph}
The proof uses arguments almost identical to those from \cite{Bernstein}. Our starting point is a data structure similar to that defined in~\cite{Bernstein}, that is called \WSES (weight-sensitive Even-Shiloah). Its input is the initial extended light graph $\hat G^L$, a source vertex $s$, a distance parameter $D$, and an error parameter $\eps$. Recall that every special vertex $v_C$ of $\hat G^L$ represents some connected component $C$ in one of the heavy graphs $G^H_i$. For convenience, abusing the notation, in this section we view $C$ as the set of vertices that belong to the connected component. 

The data structure starts from the initial extended light graph $\hat G^L$, which then undergoes a number of transformations, described below.
As the graph evolves due to these transformations, it may no longer coincide with the current extended light graph $\hat G^L$, and so we denote the current graph obtained over the course of the sequence of transformations by $\tG$. The algorithm maintains a tree $T\subseteq \tG$, whose root is $s$, such that every vertex $u$ with $\dist_{\tG}(s,u)\leq D$ lies in $T$. Moreover, if $u\in V(T)$, then $\dist_{\tG}(s,u)\leq \dist_T(s,u)\leq (1+\eps)\dist_{\tG}(s,u)$. The data structure supports the following three operations:

\begin{itemize}
\item Delete an edge $e$ from the graph $\tG$, denoted by $\WSESDEL(e)$;

\item Insert an \emph{eligible} edge $e$ into the graph $\tG$, denoted by $\WSESIN(e)$; a new edge $e=(u,u')$ of length $1\leq \ell(e)\leq D$ is eligible for insertion iff $\ell(e)$ is an integer, $u$ and $u'$ are regular vertices, and there is some special vertex $v_C$ to which they are both currently connected.

\item Create a twin $v_{C'}$ of a special vertex $v_C$, denoted by $\twin(v_C,C')$. Here, we are given a special vertex $v_C$, and a subset $C'\subseteq C$ containing at most half the vertices of $C$. We need to insert a new special vertex $v_{C'}$ into the graph, and to connect it to every vertex of $C'$ with a special edge of length $1/4$. Additionally, if $v_C$ lies in the current tree $T$ and $p$ is its parent in the tree, but $p\not\in C'$, then we add the special edge $(v_{C'},p)$ of length $1/4$ to the graph. We will exploit this operation in order to split special vertices that represent connected components of the graphs $G^H_i$.
\end{itemize}

The following theorem follows from Lemma 4.3 from~\cite{Bernstein} with slight changes; for completeness we provide its proof in Section~\ref{sec: old stuff}.

\begin{theorem}\label{thm: WSES}
There is a deterministic algorithm, that, given the initial extended light graph $\hat G^L$, undergoing operations $\WSESDEL$, $\twin$, and $\WSESIN$ for eligible edges, and parameters $D\geq 1$, $0<\eps<1$, 
maintains a  tree $T\subseteq \tG$ (where $\tG$ is the current graph obtained from $\hat G^L$ after applying a sequence of the above operations), rooted at $s$, such that for every vertex $u\in V(\tG)$ with $\dist_{\tG}(s,u)\leq D$, $u\in V(T)$, and $ \dist_T(s,u)\leq (1+\eps)\dist_{\tG}(s,u)$. Additionally, for each such vertex $u\in V(T)$, a value $\delta(u)\geq \dist_T(s,u)$ is stored with $u$, such that $\delta(u)\leq (1+\eps)\dist_{\tG}(s,u)$.
The total update time of the algorithm is $O\left (\frac{nD\log n}{\eps}\right )+O\left (\sum_{e\in E}\frac{D\log n}{\eps \ell(e)}\right )$, where $E$ is the set of all edges that were ever present in graph $\tG$, 
and $n$ is the total number of vertices that were ever present in graph $\tG$. 
\end{theorem}



Recall that the extended light graph $\hat G^L$ consists of two types of vertices: the regular vertices are the vertices of $V(G)$, and the special vertices, defined as follows: for every $1\leq i\leq \lambda$, for every connected component $C\in G^H_i$, there is a special vertex $v_C$, that connects to every (regular) vertex of $C$ with an edge of length $1/4$. We call the edges incident to special vertices \emph{special edges}, and the remaining edges of $\hat G^L$ \emph{regular edges}. We  maintain the $\WSES$ data structure from Theorem~\ref{thm: WSES}, starting from the initial graph $\hat G^L$, with the original source vertex $s$, error parameter $\eps/2$, and the distance bound $8D$. Observe that, over the course of the algorithm, graph $\hat G^L$ undergoes the following changes.

First, when some vertex $v^*$ is deleted from $G$, then we need to delete every edge incident to $v^*$ from $\hat G^L$. This can be implemented through the $\WSESDEL$ procedure.

Second, when some vertex $v$ that was heavy for some class $i$, becomes light for that class, we insert all edges that are incident to $v$ in $G^H_i$ into $\hat G^L$. We claim that these insertions can be implemented by using the $\WSESIN$ operation, as each inserted edge is eligible for insertion. Indeed, consider any such edge $e=(u,u')$. Since $e\in G^{H}_i$ before its insertion into $\hat G^L$, vertices $u$ and $u'$ lie in the same connected component of $G^H_i$, that we denote by $C$. But then there is a special vertex $v_C$ in $\hat G^L$, that is connected to both $u$ and $u'$, so $e$ is an eligible edge.

The third type of changes is when for some $1\leq i\leq \lambda$, some connected component $C$ of $G^H_i$ splits into two connected components, $C_1$ and $C_2$. We assume that $|V(C_1)|\leq |V(C_2)|$. In this case, we need to delete $v_C$ from $\hat G^L$, and add $v_{C_1}$ and $v_{C_2}$, together with length-$1/4$ edges that connect every vertex $u\in V(C_1)$ to $v_{C_1}$ and every vertex $u'\in V(C_2)$ to $v_{C_2}$. We denote the required update operation to the \WSES data structure by $\WSESSPLIT(C,C_1,C_2)$. As input, this operation receives the names $C,C_1,C_2$ of the corresponding components, and a list of all vertices in $C_1$ (recall that $|V(C_1)|\leq |V(C_2)|$). We show that the required update can be implemented by executing a suitable sequence of $\twin$ and $\WSESDEL$ operations. The algorithm appears in Figure~\ref{fig: WSESSPLIT}.


\begin{figure}
\program{$\WSESSPLIT(C,C_1,C_2)$}{
Input: names $C,C_1,C_2$ of clusters with $v_C\in V(\hat G^L)$ and a list of vertices in $C_1$.

\begin{enumerate}
\item  Perform operation $\twin(v_C,V(C_1))$ on $\hat G^L$. 
Rename $v_{C}$ as $v_{C_2}$. 

\item Let $p$ be the parent of $v_{C_2}$ in the tree $T$ maintained by \WSES data structure. If $p\not\in C_1$, execute $\WSESDEL((p,v_{C_1}))$. 


\item For each vertex $z\in V(C_1)$, delete the edge $e'_z=(v_{C_2},z)$ from $\hat G^L$ by running $\WSESDEL(e'_z)$. \label{cluster splitting - delete extra edges}
\end{enumerate}
}
\caption{Procedure \WSESSPLIT. \label{fig: WSESSPLIT}}
\end{figure}

It is immediate to see that the algorithm for \WSESSPLIT updates the graph $\hat G^L$ correctly. 

We now bound the total number of vertices and edges that were inserted into the graph as a part of the \WSESSPLIT operation. Notice that when a cluster $C$ is split into $C_1$ and $C_2$, with $|C_1|\leq |C_2|$, we insert one new vertex and at most $|C_1|+1$ new edges into the graph: one edge incident to every vertex of $C_1$. We say that every vertex of $C_1$ is \emph{responsible} for the unique new edge that is inserted into the graph and is incident to it, and one of these vertices is responsible for the one additional edge -- the edge connecting $v_{C_1}$ to $p$. Since $|C_1|\leq |C|/2$, every regular vertex of $\hat G^L$ will be responsible for at most $2\lambda \log n$ edge insertions -- at most $2\log n$ edge insertions due to the splitting of clusters of $G^H_i$ for each $1\leq i\leq \lambda$. Therefore, the total number of special edges inserted into $\hat G^L$ due to cluster splitting is $O(n \lambda\log n)=O(n\log D\log n)$. It is easy to verify that the total number of special vertices inserted into the graph due to cluster splitting is $O(n\lambda)=O(n\log D)$.


Next, we bound the total update time needed to maintain the $\WSES$ data structure  for $\hat G^L$. Recall that, from Theorem~\ref{thm: WSES}, the total running time is bounded by $O\left(\frac{n'D\log n'}{\eps}\right)+O\left (\sum_{e\in E}\frac{D\log n'}{\eps \ell(e)}\right )$, where $E$ is the set of all edges ever present in the graph $\tilde G$, and $n'$ is the total number of vertices ever present in the graph. 
Recall that $n'=O(n\log D)$, and that $(E^R,E^S)$ is the partition of $E'$ into sets containing regular and special edges, respectively, where $|E^S|=O(n\log D\log n)$.

The first term in the bound on the running time becomes:

\[O\left(\frac{n'D\log n'}{\eps}\right)=O\left(\frac{nD\log D\log(n\log D)}{\eps}\right)=O\left(\frac{nD\log n\log^2 D}{\eps}\right).\]

The contribution of the special edges to the second term is bounded by:

\[O\left (|E^S|\cdot \frac{D\log n'}{\eps}\right )= O\left(n\log D\log n\cdot \frac{D\log(n\log D)}{\eps}\right)=O\left(\frac{nD\log^2 n\log^2D}{\eps}\right).\]

Lastly, the contribution of the regular edges to the second term is bounded by:

\[O\left (\sum_{e\in E^R}\frac{D\log n'}{\eps \ell(e)}\right )=O\left (\sum_{e\in E^R}\frac{D\log n\log D}{\eps \ell(e)}\right ).\]

In total, the running time is bounded by:

\[O\left(\frac{nD\log^2 n\log^2D}{\eps}\right)+O\left (\sum_{e\in E^R}\frac{D\log n\log D}{\eps \ell(e)}\right ).  \]

In addition to maintaining the data structure $\WSES$ from Theorem~\ref{thm: WSES}, our algorithm maintains, for each $1\leq i\leq \lambda$, a connectivity/spanning forest data structure $\CONNSF(G^H_i)$. Recall that the total time required to maintain each such structure is $O\left((|E(G^H_i)|+n\right )\log^2n)\leq O(n^2\log^2n)$, and the total time required to maintain all such structures for all $1\leq i\leq \lambda$ is at most $O(n^2\lambda\log^2n)=O(n^2\log D\log^2n)$. Finally, for every vertex $v\in V(G)$ and every index $1\leq i\leq \lambda$, if $v\in G^H_i$, we maintain the degree $d_i(v)$ of $v$ in $G^H_i$.

It now remains to describe the algorithm for handling a deletion of a vertex $v^*$ from the graph $G$. Over the course of the update step, for each $1\leq i\leq \lambda$, we will maintain a set $Q_i$ of vertices that need to be deleted from $G^H_i$, and the set $\hat E_i$ of edges that are incident to all these vertices. 

The algorithm proceeds as follows. We consider the indices $1\leq i\leq \lambda$ one-by-one. Consider the current index $i$. We start with $\hat E_i=\emptyset$. If $v^*\not\in G^H_i$, then we terminate the algorithm and continue to the next index $i$. Otherwise, we initialize $Q_i=\set{v^*}$. While $Q_i\neq \emptyset$, let $v$ be any vertex in $Q_i$. We delete $v$ from $G^H_i$, and add every edge incident to $v$ in $ G^H_i$ to $\hat E_i$. For every neighbor $u$ of $v$ in $G^H_i$, we decrease $d_i(u)$ by $1$. If $d_i(u)$ falls below $\tau_i$ and $u\not\in Q_i$, we add $u$ to $Q_i$. 

Once $Q_i=\emptyset$, we start processing all edges of $\hat E_i$, one-by-one. While $\hat E_i\neq \emptyset$, let $e$ be any edge in $\hat E_i$. We run Procedure $\deledge$ for edge $e$ and index $i$; the procedure is described in Figure~\ref{fig: procedure delete edge}.

\begin{figure}[h]
\program{Procedure \deledge}{

Input: an integer $1\leq i\leq \lambda$ and an edge $e=(u,w)\in  E(G^H_i)$.

\begin{enumerate}
\item If $e$ is not incident to the original vertex $v^*$, insert edge $e$ into the \WSES data structure and into $\hat G^L$,  using operation $\WSESIN(e)$ (as discussed above, $e$ is an eligible edge).
\item Delete edge $e$ from $\CONNSF(G^H_i)$.
\item Check whether $u,w$ remain connected in $\CONNSF(G^H_i)$ in time $O(\log n/\log\log n)$. If so, terminate the procedure.

We assume from now on that $u$ and $w$ are no longer connected in $\CONNSF(G^H_i)$. We denote by $C$ the original connected component to which they belonged before the deletion, and by $C_1,C_2$ the two new components. Next, we will try to establish which of the two components is smaller.

\item Run two BFS searches in parallel: one in the tree of $\CONNSF(G^H_i)$ to which $u$ belongs, and one in the tree containing $w$, so that both searches explore the same number of vertices at each time step; store all vertices explored. Terminate the algorithm once one of the two trees is completely explored. We assume w.l.o.g. that it is the tree corresponding to $C_1$.

\item Run $\WSESSPLIT(C,C_1,C_2)$ with the list of vertices of $C_1$.
\end{enumerate}
}
\caption{Procedure \deledge \label{fig: procedure delete edge}}
\end{figure}

Once $\hat E_i=\emptyset$, we continue to the next index $i+1$. Once we finish processing all indices, we delete from $\hat G^L$ all edges incident to the original vertex $v^*$ one-by-one, using procedure $\WSESDEL(\hat G^L,e)$. This finishes the procedure for the deletion of a vertex $v^*$ from $G$. We have already accounted for the total time needed to maintain data structure $\WSES$, and this includes the running time needed for cluster splitting. The BFS searches on the trees corresponding to clusters $C_1$ and $C_2$ terminate in time $O\left(\min\set{|C_1|,|C_2|}\right )$, and can be charged to the cluster splitting procedure $\WSESSPLIT(C,C_1,C_2)$ -- this will not increase its asymptotic running time. The total time needed to maintain the $\CONNSF(G^H_i)$ data structures for all $1\leq i\leq \lambda$ is $O(n^2\log^2n\log D)$, and this includes the time required to answer connectivity queries, as we only ask one connectivity query for each edge deleted from $G^H_i$. The total running time that is needed to maintain all data structures is therefore bounded by $O(n^2\log D\log^2n)+O\left(\frac{nD\log^2 n\log^2D}{\eps}\right)+O\left (\sum_{e\in E^R}\frac{D\log n\log D}{\eps \ell(e)}\right )$, as required.

\section{Proof of Theorem~\ref{thm: WSES}}\label{sec: old stuff}

Throughout the proof, the notation $\dist(u,u')$ refers to the distance from $u$ to $u'$ in graph $\tilde G$.
The theorem is proved by modifying the standard \EST algorithm; we follow the proof of Lemma 4.3 from~\cite{Bernstein} almost exactly.
\input{appendix-Bernstein}


\section{Proof of Lemma~\ref{lem: flow or cut edge version}}\label{appx: sec: flow or cut edge version}

The proof is almost identical to the proof of Lemma~\ref{lem: flow or cut}, except that now we deal with edge-congestion instead of vertex-congestion, and with regular cuts instead of vertex cuts. As before, the algorithm is partitioned into phases, where the input to phase $i$ is a pair of subsets $A_i\subseteq A$, $B_i\subseteq B$ of vertices that were not routed yet, with $|A_i|=|B_i|$. We will ensure that during the $i$th phase, we either compute a set $\pset_i$ of at least $\frac{|A_i|\log^3 n}{\ell^2}$ edge-disjoint paths,  where every path connects a distinct vertex of $A_i$ to a distinct vertex of $B_i$, such that the length of every path in $\pset_i$ is at most $\ell$; or we will return a cut $(X,Y)$ with the required properties. The algorithm terminates once $|A_i|< z$, and so $\bigcup_{i'=1}^{i-1}\pset_{i'}$ contains more than $|A|-z$ paths.
Since we are guaranteed that for every $i$, $|A_i|\leq |A_{i-1}|(1-\log^3 n/\ell^2)$, the number of phases is bounded by $\ell^2/\log^2n$. The final set of paths is $\pset=\bigcup_i\pset_i$, and, since the paths in every set $\pset_i$ are edge-disjoint, the paths in $\pset$ cause edge-congestion at most $\ell^2/\log^2n$. We will also ensure that every phase runs in time $O(\hn\ell\log^3n)$, which will ensure that the total running time is  $ O(\hn\ell^3\log n)$, as required. The input to the first phase is $A_1=A$ and $B_1=B$.
It is now enough to describe the execution of a single phase. The next claim, which is an analogue of Claim~\ref{claim: flow or cut one phase} for edge-disjoint routing, and whose proof is almost identical, will finish the proof of the lemma.

\begin{claim}\label{claim: flow or cut one phase edge version}
There is a deterministic algorithm, that, given a subgraph $W'\subseteq W$ containing at least half the vertices of $W$, together with two equal-cardinality subsets  $A',B'$ of $V(W')$, and with a parameter $\ell>2{\log^{1.5} n}$, computes one of the following:

\begin{itemize}
\item either a collection $\pset'$ of at least $\frac{|A'|\log^3 n}{\ell^2}$ edge-disjoint paths in $W'$, where each path connects a distinct vertex of $A'$ to a distinct vertex of $B'$ and has length at most $\ell$; or

\item a cut $(X,Y)$ in $W'$, with $|E_{W'}(X,Y)|\leq \frac{4\log^4 n}{\ell}\min\set{|X|,|Y|}$, and $|X|,|Y|\geq |A'|/2$.
\end{itemize}

The running time of the algorithm is $ O(\hn\ell\log^3n)$.
\end{claim}

\begin{proof}
We build a new graph $H$: start with  graph $W'$, and add a source vertex $s$ that connects to every vertex in $A'$ with an edge; similarly, add a destination vertex $t$, that connects to every vertex in $B'$ with an edge. Set up single-source shortest path data structure $\EST(H,s,\ell+1)$, up to distance $(\ell+1)$ in $H$, with $s$ being the source. While the distance from $s$ to $t$ is less than $(\ell+1)$, choose any path $P$ in $H$ connecting $s$ to $t$, that has at most $(\ell+1)$ edges. Let $P'$ be the path obtained from $P$, after we delete its endpoints $s$ and $t$, so that $P'$ now connects some vertex $a\in A'$ to some vertex $b\in B'$. Add $P'$ to $\pset'$, and delete all edges of $P$ from $H$. Notice that, in particular, the edges $(s,a)$ and $(b,t)$ are deleted from $H$ -- this ensures that the paths in the final set $\pset'$ have distinct endpoints. As before, finding the path $P$ takes time $O(\ell)$, since we simply follow the shortest-path tree that the \EST data structure maintains. The total update time of the data structure is $ O(|E(H)|\ell)$, and the total running time of the algorithm, that includes selecting the paths and deleting their edges from $H$, is bounded by $O(|E(H)|\cdot \ell)$. We now consider two cases. First, if we have managed to route at least $\frac{|A'|\log^3 n}{\ell^2}$ paths, then we terminate the algorithm, and return the set $\pset'$ of paths.

Otherwise, consider the current graph $H'$, that is obtained from $H$ after all edges participating in the paths in $\set{P\mid P'\in \pset'}$ were deleted. We perform a BFS from the vertices of $A'$ in this graph: start from a set $S_0$ containing all vertices of $A'$ that are connected to $s$ with an edge in $H'$. Given the current vertex set $S_j$, let $S_{j+1}$ contain all vertices of $S_{j}$ and all neighbors of $S_j$ in $H'\setminus\set{s,t}$. 
Similarly, we perform a BFS from the vertices of $B'$ in $H'$: 
start from a set $T_0$ containing all vertices of $B'$ that are connected to $t$ with an edge in $H'$. Given the current vertex set $T_j$, let $T_{j+1}$ contain all vertices of $T_{j}$ and all neighbors of $T_j$ in $H'\setminus\set{s,t}$.

Using exactly the same reasoning as in the proof of Claim~\ref{claim: flow or cut one phase}, there is an index $j<\ell/2$, such that
either (i) $|S_{j+1}|\leq n/2$ and $|S_{j+1}|<|S_j|\left (1+\frac{2\log n}{\ell}\right)$; or (ii) $|T_{j+1}|\leq n/2$ and $|T_{j+1}|<|T_j|\left (1+\frac{2\log n}{\ell}\right)$.

 We assume w.l.o.g. that $|S_{j+1}|\leq n/2$ and $|S_{j+1}|<|S_j|\left (1+\frac{2\log n}{\ell}\right)$, and we define the cut $(X,Y)$, by setting $X=S_j$ and $Y=V(W')\setminus X$. Clearly, $X$ contains all vertices of $A'$ that still need to be routed, so $|X|\geq |A'|\left(1-\frac{\log^3 n}{\ell^2}\right )\geq |A'|/2$. Since $|S_{j+1}|\leq n/2$, $|Y|\geq n/2$ and in particular $|Y|\geq |X|$. We now bound $|E_{W'}(X,Y)|$.
 
 The set $E_{W'}(X,Y)$ of edges consists of two subsets: edges that lie on paths in $\pset'$, and the remaining edges, that belong to the graph $H'$. The cardinality of the former set of edges is bounded by $ |\pset'|\cdot \ell\leq \frac{|A'|\log^3 n}{\ell^2}\cdot \ell\leq \frac{|A'|\log^3 n}{\ell}\leq \frac{2|S_j|\log^3 n}{\ell}$. In order to bound the cardinality of the second set of edges, let $Z=S_{j+1}\setminus S_j$, so $|Z|< \frac{2\log n}{\ell}|S_j|$. Every edge in $E_{H'}(X,Y)$ connects a vertex of $S_j$ to a vertex of $Z$. Since the maximum vertex degree in $W$ is bounded by $\log^3n$, the number of such edges is at most $|Z|\cdot \log^3n\leq \frac{2\log^4n}{\ell}|S_j|$. We conclude that $|E_{W'}(A,B)|\leq \frac{4\log^4n}{\ell}\min\set{|X|,|Y|}$, and that $|X|,|Y|\geq |A'|/2$.

 The running time of the first part of the algorithm, when the paths of $\pset'$ are computed is $O(|E(H)|\cdot \ell)$, as discussed above. The second part only involves performing two BFS searches in graph $H'$ and computing the final cut, and has running time $O(|E(H)|)$. Since $|E(H)|=O(\hn\log^3n)$, the total running time is $O(\hn\ell\log^3n)$.
\end{proof}

%% file: appendix-Bernstein.tex
We scale all distances up by factor $4$, so that all special edges have length $1$, and the length of every regular edge is an integer greater than $3$. 
For simplicity, we assume that $\eps=1/k$ for some integer $k>4$; this can be done so that the value of $\eps$ decreases by at most factor $4$. Recall that we start with the graph $\hat G^L$ that undergoes $\WSESDEL$, $\WSESIN$ and $\twin$ operations, and that the current graph is denoted by $\tG$. We denote $V(\tG)$ by $V$. We also denote by $E$ the set of all edges that were ever present in $\tG$, and $m=|E|$. We let $n$ be the total number of vertices that were ever present in $\tG$.
We will use the following definition.

\begin{definition}
Given an edge $e$ and a number $x>0$, $\round_e(x)$ is the smallest number $y>x$ that is an integral multiple of $\eps \ell(e)$. Note that $x<\round_e(x)\leq x+\eps \ell(e)$, and, since all edge lengths are integral, $\round_e(x)$ is an integral multiple of $\eps$.
\end{definition}

The algorithm maintains a tree $T\subseteq \tG$, that is rooted at the vertex $s$, and contains a subset of the vertices of $V$.
For every vertex $u\in V$, we maintain a value $\delta(u)$, which is our estimate on $\dist(s,u)$. For every vertex $u\in V(T)$, we also maintain a heap $H_u$ that contains all its neighbors. However, the key associated with every neighbor is computed differently than in the \EST algorithm. For every neighbor $w$ of $u$, we maintain a \emph{local copy} $\delta_{u}(w)$, which is the local estimate of $u$ on the value $\delta(w)$. 
 We will ensure that $\delta(w)\leq \delta_u(w)\leq \round_e(\delta(w))\leq \delta(w)+\eps \ell(e)$, where $e=(u,w)$. 
For every neighbor $w$ of $u$, we store the vertex $w$ in $H_u$, with the key $\delta_u(w)+\ell(u,w)$. 
We ensure that the following invariants hold for every vertex $u$ in $V(T)$ throughout the algorithm (some of these invariants repeat the properties stated above).

\begin{properties}{J}
\item If $u=s$ then $\delta(u)=0$. Otherwise, if $p$ is the parent of $u$ in $T$, then $\delta(u)=\delta_u(p)+\ell(u,p)$, and $p$ is the element with the smallest key in $H_u$. \label{inv: parent}


\item If $e=(u,w)$ is an edge incident to $u$, then $\delta(w)\leq \delta_u(w)\leq \round_e(\delta(w))\leq \delta(w)+\eps \ell(e)$; and\label{inv: regular edge}

\item Value $\delta(u)$ does not decrease over the course of the algorithm, and it is always an integral multiple of $\eps$ between $0$ and $(1+\eps)D$ (if vertex $u$ is inserted due to a call to $\twin$ operation, this holds from the moment of insertion). Similarly, for every neighbor $w$ of $u$, value $\delta_u(w)$ does not decrease over the course of the algorithm, and it is always an integral multiple of $\eps$ (if edge $(u,w)$ is inserted due to a $\twin$ or $\WSESIN$ operations, this holds only from the moment the edge is inserted). \label{inv: monotonicity and bounded fractionality} 
\end{properties}

We also ensure the following invariant.

\begin{properties}[3]{J}
\item If $u\in V\setminus V(T)$, then $\delta(u)> (1+\eps)D$ and $\dist(s,u)> D$. From the moment $u$ is deleted from $T$, $\delta(u)$ does not change, and $\delta(u)<2D$. \label{inv: large distances}
\end{properties}

\begin{claim}\label{claim: approx}
Assume that the invariants~(\ref{inv: parent})--(\ref{inv: large distances}) hold throughout the algorithm. Then, throughout the algorithm, for every vertex $u\in V(T)$,
 $\dist(s,u)\leq \delta(u)\leq (1+\eps)\dist(s,u)$. Moreover, if $P_u$ is the path connecting $u$ to $s$ in $T$,  then the length of $P_u$ is at most $\delta(u)$. 
\end{claim}

\begin{proof}
We first show that for every vertex $u\in V(T)$, $\delta(u)\geq \dist(s,u)$. The proof is by induction on the number of edges on the path $P_u$, connecting $u$ to $s$ in the tree $T$. If $P_u$ contains $0$ edges, then $u=s$, and, from Invariant~(\ref{inv: parent}), $\delta(u)=0=\dist(s,u)$. Assume now that $u\neq s$, and let $p$ be the parent of $u$ in $T$. From the induction hypothesis, $\delta(p)\geq \dist(s,p)$. From Invariant~(\ref{inv: parent}), $\delta(u)=\delta_u(p)+\ell(u,p)$; from Invariant (\ref{inv: regular edge}), $\delta_u(p)\geq \delta(p)$, and therefore, altogether, $\delta(u)\geq \delta(p)+\ell(u,p)\geq \dist(s,p)+\ell(p,u)\geq \dist(s,u)$.

We now turn to prove that for every vertex $u\in V(T)$, $\delta(u)\leq (1+\eps)\dist(s,u)$. 
Let $P^*_u$ be the shortest path connecting $s$ to $u$ in $\tG$, and among all such paths, choose the one with fewest edges, breaking ties arbitrarily.
The proof is by induction on the number of edges on $P^*_u$. If $P^*_u$ contains no edges, then $u=s$, and $\delta(u)=0=\dist(s,u)$ form Invariant~(\ref{inv: parent}), so the claim is true. Assume now that $P^*_u$ contains $i$ edges, and that the claim is true for all vertices $w\in V(T)$, for which $P^*_w$ contains fewer than $i$ edges. Let $w$ be the penultimate vertex on $P^*_u$. Observe that, if we let $P'$ be the path obtained from $P^*_u$, by deleting the vertex $u$ from it, then $P'$ must be the shortest $s$--$w$ path in $\tG$, and among all such paths, it contains the smallest number of edges. Assume first that $w\in V(T)$. Then, from the induction hypothesis, $\delta(w)\leq (1+\eps)\dist(s,w)$.
The key of $w$ in the heap $H_u$ is $\delta_u(w)+\ell(u,w)\leq \delta(w)+\eps \ell(u,w)+\ell(u,w)\leq (1+\eps)(\dist(s,w)+\ell(u,w))\leq (1+\eps)\dist(s,u)$ (we have used Invariant~(\ref{inv: regular edge}) for the first inequality.) Since $\delta(u)$ is equal to the smallest key in $H_u$, we get that $\delta(u)\leq (1+\eps)\dist(s,u)$.
Finally, if $w\not \in V(T)$, then, from Invariant~(\ref{inv: large distances}), $\dist(s,w)>D$, so $\dist(s,u)>D$ must hold as well. However, since $u\in V(T)$, from Invariant~(\ref{inv: monotonicity and bounded fractionality}), $\delta(u)\leq (1+\eps)D\leq (1+\eps)\dist(s,u)$.

We now turn to prove the last assertion. The proof is by induction on the number of edges on the path $P_u$, connecting $u$ to $s$ in the tree $T$. If $P_u$ contains $0$ edges, then the claim is trivially true. Assume now that $P_u$ contains $i>0$ edges, and let $p$ be the parent of $u$ in $T$. From the induction hypothesis, the length of the path connecting $p$ to $s$ in $T$ is at most $\delta(p)$. Therefore, the length of the path $P_u$ is at most $\delta(p)+\ell(p,u)\leq \delta_u(P)+\ell(p,u)=\delta(u)$.
\end{proof}

As in the $\EST$ algorithm, our algorithm will perform inspection of edges, where each inspection will involve a constant number of standard heap operations and will take $O(\log m)$ time. Recall that in the original $\EST$ algorithm, an edge $e=(u,w)$ is inspected whenever $\delta(u)$ or $\delta(w)$ increase. In the former case, we say that $e$ is inspected due to $u$, and in the latter case, it is inspected due to $w$. We will have different rules for edge inspection, which will allow us to save on the running time for inspecting edges that have large lengths.
Let $e=(u,w)$ be an edge of $\tG$ (which may either belong to the original graph $\hat G^L$, or may have been inserted over the course of the algorithm).

\begin{properties}{R}

\item Edge $e$ is inspected due to $u$ only when $\round_e(\delta(u))$ increases.
Similarly, edge $e$  is inspected due to $w$ only when $\round_e(\delta(w))$ increases. \label{prop: edge inspection first}
\item Additionally, edge $e$ is inspected when it is added to or is deleted from $\tG$. \label{prop: edge inspection last}
\end{properties}

We now bound the total update time of the algorithm due to edge inspections. First, every edge may be inspected at most once when it is added to $\tG$ and at most once when it is deleted from $\tG$. 
Consider now some edge $e=(u,w)$. Since, throughout the algorithm, $0\leq \delta(u)\leq 2D$, and $\round_e(\delta(u))$ is an integral multiple of $\eps\ell(e)$, it may increase at most $O\left(\frac{D}{\eps\ell(e)}\right )$ times over the course of the algorithm. So edge $e$ may be inspected due to $u$ at most $O\left(\frac{D}{\eps\ell(e)}\right )$ times, and similarly it can be inspected due to $w$ at most $O\left(\frac{D}{\eps\ell(e)}\right )$ times.  Let $E$ denote the set of all edges that are ever present in the graph $\tG$. Since every inspection of an edge takes $O(\log n)$ time,  the total update time of the algorithm due to edge inspections is $O\left (\sum_{e\in E}\frac{D\log n}{\eps\ell(e)}\right )$.

\paragraph{Data Structures.}
We maintain a tree $T\subseteq \tG$ rooted at $s$, the values $\delta(u)$ and the heaps $H_u$  for all $u\in V$, as described above. Additionally, for every vertex $u\in V(T)$, 
whenever $\delta(u)$ increases, we need to be able to quickly identify the edges that are incident to $u$, which need to be inspected. In order to do so, for every vertex $u\in V(T)$,
we maintain $2D/\eps$ buckets $B_1(u),\ldots,B_{2D/\eps}(u)$. For every neighbor $w$ of $u$, we add $w$, together with a pointer to the copy of $u$ in $H_w$,  to one of the buckets. We also add a pointer from $w$ to its copy in a bucket of $u$. If vertex $w$ lies in bucket $B_i(u)$, then we need to inspect the edge $(u,w)$, and to update $\delta_w(u)$ when $\delta(u)$ becomes greater than $\eps i$. We will use the pointer to the copy of $u$ in $H_w$, that we store together with $w$ in $B_i(u)$, in order to do it efficiently. We add $w$ to the bucket $B_{\round_e(\delta(u))/\eps}$. 

\paragraph{Initialization.}
We construct a shortest-path tree $T$ of $\tG$ rooted at $s$, using the algorithm of Thorup~\cite{linear-sssp} in $O(m)$ time. We then compute, for every vertex $u\in V$, the initial value $\delta(u)=\dist(s,u)$; notice that this value is an integer. We delete from $T$ all vertices $u$ with $\delta(u)>(1+\eps)D$. For every vertex $u$, for every neighbor $w$ of $u$, we set $\delta_u(w)=\delta(w)$. We add $w$ to the heap $H_u$ with the key $\delta_u(w)+\ell(u,w)$, and we add $w$ to the bucket $B_{\round_e(\delta(u))/\eps}$, where $e=(u,w)$. All this can be done in time $O\left(m+\frac{Dn\log m}{\eps}\right)$. Observe that Invariants~(\ref{inv: parent})--(\ref{inv: large distances}) hold after the initialization.

\paragraph{Edge Insertion.}
Suppose we need to insert an eligible edge $e=(u,w)$ into $\tG$.  Recall that $\ell(e)\geq 4$ is an integer, and that there is some special vertex $v=v_C$, and special edges $(u,v)$, $(w,v)$ in $\tG$, each of length $1$. We assume that $u,w\in V(T)$, since otherwise we can ignore this edge (it is easy to verify that $\dist(s,w)$ and $\dist(s,u)$ do not decrease). From our invariants, $\delta(u)\leq \delta_{u}(v)+\ell(u,v)\leq  \delta(v)+\eps+1\leq \delta_v(w)+\ell(v,w)+\eps+1\leq \delta(w)+2+2\eps \leq \delta(w)+3$. Similarly, $\delta(w)\leq \delta(u)+3$. Therefore, if we let $\delta_u(w)=\delta(w)$, and insert $w$ into $H_u$ with the key $\delta_u(w)+\ell(e)$, we will still maintain the invariant that $\delta(u)$ is the value of the smallest key in $H_u$, and in particular it does not decrease. Similarly, we can insert $u$ into $H_w$ with the key $\delta_w(u)+\ell(e)$, where $\delta_w(u)=\delta(u)$, without violating any invariants. We also insert $w$ into the appropriate bucket $B_{\round_e(\delta(u))/\eps}(u)$, together with a pointer to the copy of $u$ in $H_w$, and similarly, we insert $u$ into the appropriate bucket $B_{\round_e(\delta(w))/\eps}(w)$, together with a pointer to the copy of $w$ in $H_u$. All this takes $O(\log n)$ time, and is included in a single inspection of the edge $e$.

\paragraph{Twin Operation.}
Recall that in the $\twin$ operation, we are given a special vertex $v_C$, together with a subset $C'\subseteq C$ of regular vertices. Our goal is to insert a ``twin'' vertex $v_{C'}$ for $v_C$ into $\tG$, and to connect it to every vertex in $C'$, and to the parent $p$ of $v_C$ in the current tree $T$. We add the new vertex $v_{C'}$ into $T$ as a child of $p$, setting $\delta_p(v_{C'})=\delta_p(v_C)$, adding $v_{C'}$ to the heap $H_p$ with the same key as $v_C$, and to the bucket of $p$ to which $v_C$ currently belongs. For every vertex $u\in C'\setminus \set{p}$, we similarly set $\delta_u(v_{C'})=\delta_u(v_C)$, add $v_{C'}$ to $H_u$ with the same key as $v_C$, and to the bucket of $u$ to which $v_C$ currently belongs. We initialize $H_{v_{C'}}$ using the copies of the neighbors of $v_{C'}$ in the heap $H_{v_C}$, and we initialize the buckets of $v_{C'}$ similarly. We set $\delta(v_{C'})=\delta(v_C)$. It is easy to verify that all invariants continue to hold. The running time of the operation is $O(D/\eps)$ plus the time needed to inspect every edge incident to $v_{C'}$, which is already accounted for in the analysis of the total running time for edge inspection.

\paragraph{Edge Deletion.}
Consider a call to procedure $\WSESDEL$ with an edge $e=(u,w)$ that needs to be deleted from the graph $\tilde G$. Assume first that edge $e$ does not belong to the current tree $T$. Then we simply delete  $u$ from $H_w$ and the corresponding bucket of $w$, delete $w$ from $H_u$ and the corresponding bucket of $u$, and terminate the procedure. It is immediate to verify that all invariants continue to hold. Therefore, we assume from now on that edge $(u,w)$ belongs to the current tree $T$, and we assume w.l.o.g. that $w$ is the parent of $u$ in $T$. Let $T_u$ be the subtree of $T$ rooted at $u$. Throughout the update procedure, we view the tree $T_u$ as fixed, while the tree $T$ is changing when vertices of $T_u$ are attached to it.

Throughout the update procedure,  we maintain a heap $H$, containing the vertices of $T_u$ that we need to inspect. The key stored with each vertex $x$ in $H$ is the current value $\delta(x)$. As the update procedure progresses, $\delta(x)$ may grow.
We ensure that throughout the update procedure, Invariants~(\ref{inv: parent})--(\ref{inv: large distances}) hold (where Invariant (\ref{inv: parent}) is only guaranteed to hold for vertices that are currently attached to $T$). Additionally, we ensure that the following invariant holds throughout the update procedure:

\begin{properties}[4]{J}
\item For every vertex $x$ in $T_u$, if $y$ is the vertex in $H_x$ with the smallest key $\delta_x(y)+\ell(x,y)$, then $\delta(x)\leq \delta_x(y)+\ell(x,y)$.\label{prop: extra}
\end{properties}

Note that this invariant holds at the beginning of the update procedure due to Invariant~(\ref{inv: parent}).
 
Initially, $H$ contains a single vertex, the vertex $u$. Consider some vertex $x\in T_u$, and let $y$ be the parent of $x$ in $T_u$, with the corresponding edge denoted by $e'=(x,y)$.  We will only add $x$ to $H$ if $\round_{e'}(\delta(y))$ increases. Moreover, for every vertex $y$ in $H$, and every child vertex $x$ of $y$ in $T_u$, for which $\round_{e'}(\delta(y))$ increased (where $e'=(x,y)$), vertex $x$ is added to $H$ the moment $\round_{e'}(\delta(y))$ increases.

Over the course of the algorithm, every vertex $x$ of $T_u$ is in one of the following four states:

\begin{itemize}
\item {\bf Untouched:} we have never inspected $x$, and it is not currently attached to the tree $T$. If $x$ is untouched, then it does not belong to $H$, and the value $\delta(x)$ did not change during the current update procedure yet. All descendants of $x$ in $T_u$ are also untouched.

\item {\bf Settled:} we attached $x$ to the tree $T$. A settled vertex does not belong to the heap $H$, and Invariant~(\ref{inv: parent}) holds for it. When a vertex $x$ becomes settled, then every child vertex $y$ of $x$ in $T_u$ that is currently untouched also becomes settled, as do all descendants of $y$ in $T_u$ (note that they are also untouched).

\item {\bf Suspicious:} we have added $x$ to the heap $H$ but we did not inspect it yet, and $\delta(x)$ has not changed yet. However, if $y$ is the parent of $x$ in $T_v$, and $e'=(x,y)$ is the corresponding edge, then $\round_{e'}(\delta(y))$ has increased. A suspicious vertex is not attached to the tree $T$; and

\item {\bf Changed:} we have increased $\delta(x)$ in the current update procedure, but $x$ is not yet settled. All changed vertices belong to $H$, and a changed vertex is not attached to the tree. Whenever $\round_{(x,y)}(\delta(x))$ increases for any child vertex $y$ of $x$, we add $y$ to $H$.\end{itemize}

Before we describe our algorithm, we need the following three observations.

\begin{observation}\label{obs: untouched}
Let $y$ be an untouched vertex. Then there is some vertex $y'$, that currently belongs to $H$, such that $y'$ is an ancestor of $y$ in $T_u$, and $\delta(y)\geq \delta(y')$.
\end{observation}

\begin{proof}
Notice that $u$ is an ancestor of $y$ in $T_u$, and it belonged to $H$ at the beginning of the algorithm. Let $y'$ be an ancestor of $y$ that belonged to $H$ at any time during the algorithm, such that $y'$ is closest ancestor to $y$ among all such vertices. Then all vertices lying on the path $P$ connecting $y'$ to $y$ in $T_u$ are either untouched or settled (excluding the vertex $y'$). We claim that $y'$ must currently belong to $H$. Otherwise, it is a settled vertex, and so all other vertices on the path $P$ must be settled, including $y$. Assume that $P=(y=y_0,y_1,\ldots,y_r=y')$. Then for all $0\leq i<r$, value $\delta(y_i)$ did not change yet, and so, from Invariants~(\ref{inv: parent}) and (\ref{inv: regular edge}), for $0\leq i<r-1$, $\delta(y_{i-1})\geq \delta(y_{i})$, and in particular $\delta(y_0)\geq \delta(y_{r-1})$. Since $y_{r-1}$ is untouched, if we denote $e'=(y_{r-1},y')$, then $\round_{e'}(\delta(y'))$ did not change yet, so $\delta(y_{r-1})=\round_{e'}(\delta(y'))+\ell(e')\geq \delta(y')$. Therefore, $\delta(y_0)\geq \delta(y')$.
\end{proof}

\begin{observation}\label{obs: smallest key}
Let $x$ be a vertex in $H$, with smallest value $\delta(x)$, and let $y\in H_x$ be the vertex minimizing $\delta_x(y)+\ell(x,y)$. Assume further that $ \delta_x(y)+\ell(x,y)\leq \delta(x)$. Then either $y\in T\setminus T_u$, or $y$ is settled; in other words, $y$ is currently attached to $T$. Moreover, $\delta(x)=\delta_x(y)+\ell(x,y)$.
\end{observation}
\begin{proof}
Assume that $y$ is not currently attached to $T$. Then either $y\in H$, or $y$ is untouched. If $y\in H$, then $\delta(y)\geq \delta(x)$ from the choice of $x$, and, since $\ell(x,y)\geq 1$ and $\delta_x(y)\geq \delta(y)$, we get that $\delta_x(y)+\ell(x,y)> \delta(y)\geq \delta(x)$, a contradiction. Therefore, $y$ is untouched. But then from Observation~\ref{obs: untouched}, there is a vertex $y'$ in $H$ with $\delta(y')\leq \delta(y)$. 
As before, we get that $\delta_x(y)+\ell(x,y)>\delta(y)\geq \delta(y')\geq \delta(x)$ from the choice of $x$, a contradiction. Therefore, $y$ is currently attached to $T$. From Invariant~(\ref{prop: extra}), $\delta(x)=\delta_x(y)+\ell(x,y)$.
\end{proof}

\begin{observation}\label{obs: large distance}
Let $x$ be a vertex in $H$, with minimum $\delta(x)$, and assume that $\delta(x)>(1+\eps)D$. Then for every vertex $y$ that is either currently in $H$ or is untouched, $\delta(y)> (1+\eps)D$, and $\dist(s,y)>D$.
\end{observation}

\begin{proof}
Let $y$ be some vertex that is either currently in $H$ or is untouched. We first claim that $\delta(y)>(1+\eps)D$. Indeed, if $y\in H$, then, from the choice of $x$, $\delta(y)\geq \delta(x)>(1+\eps)D$. Otherwise, $y$ is an untouched vertex, so from Observation~\ref{obs: untouched}, there is some vertex $y'\in H$ with $\delta(y)\geq\delta(y')$.
But then $\delta(y')\geq \delta(x)>(1+\eps)D$, and $\delta(y)>(1+\eps)D$.

Next, we show that $\dist(s,y)>D$. Let $P^*$ be the shortest path connecting $s$ to $y$ in $\tG$, that we view as directed from $y$ to $s$, and assume for contradiction that the length of $P^*$ is at most $D$. Let $z$ be the last vertex on $P^*$ that either belongs to $H$ or is untouched, and let $z'$ be the vertex following $z$ on $P^*$, so that $z'$ is currently attached to $T$. Then $\dist(z,s)=\dist(z',s)+\ell(z,z')$. From Claim~\ref{claim: approx}, $\delta(z')\leq (1+\eps)\dist(z',s)$  held before the current update procedure, and, since $\delta(z')$ did not change over the course of the procedure, while the distances may only increase, this continues to hold. From Invariant~(\ref{inv: regular edge}), $\delta_z(z')\leq \delta(z')+\eps\ell(z,z')\leq (1+\eps)\dist(z',s)+\eps\ell(z,z')$. Therefore, vertex $z'$ is stored in $H_z$ with the key $\delta_z(z')+\ell(z,z')\leq (1+\eps)(\dist(z',s)+\ell(z,z'))\leq (1+\eps)\dist(z,s)\leq (1+\eps)D$. But, from Invariant~(\ref{prop: extra}), $\delta(z)\leq \delta_z(z')+\ell(z,z')\leq (1+\eps)D$. However, since $z$ is either untouched or it belongs to $H$, we have already established that $\delta(z)>(1+\eps)D$, a contradiction.
\end{proof}

We are now ready to describe the update procedure, that is performed as long as $H\neq \emptyset$.
In every iteration, we consider the vertex $x\in H$ with the smallest key $\delta(x)$. If $\delta(x)>(1+\eps)D$, then let $S$ be the set of all vertices that currently lie in $H$ or are untouched. From Observation~\ref{obs: large distance}, for each such vertex $y$, $\delta(y)>(1+\eps)D$ and $D(s,y)>D$. We terminate the update procedure, and the vertices of $S$ remain un-attached to $T$. 

From now on, we assume that $\delta(x)\leq (1+\eps)D$. Let $y\in H_x$ be the vertex minimizing $\delta_x(y)+\ell(x,y)$. We now consider two cases. The first case happens when  $\delta_x(y)+\ell(x,y)\leq \delta(x)$. From Observation~\ref{obs: smallest key},  $y$ is currently attached to $T$, and $\delta(x)=\delta_x(y)+\ell(x,y)$. We attach $x$ to $T$ as the child of $y$. All children of $x$ that are currently untouched, as well as all their descendants, now become automatically attached to $T$. It is immediate to verify that Invariants~(\ref{inv: parent})--(\ref{prop: extra}) continue to hold.

The second case is when $\delta_x(y)+\ell(x,y)> \delta(x)$. We then increase $\delta(x)$ by (additive) $\eps$, so it remains an integral multiple of $\eps$, and return $x$ to $H$. 
Additionally, we inspect every edge $e'=(x,w)$, where $w$ lies in the bucket $B_{\delta(x)/\eps}$ for the original value $\delta(x)$. For each such edge, we add $w$ to $H$, and update the value $\delta_w(x)$ to $\round_{e'}(\delta(x))$ (with respect to the new $\delta(x)$ value), updating also the key of $x$ in $H_w$. We also add $w$ to the bucket $B_{\round_{e'}(\delta(x))/\eps}(x)$ of $x$, again using the new value $\delta(x)$. Vertex $w$ now becomes a suspicious vertex.  It is immediate to verify that all invariants continue to hold.


We now analyze the running time of the update procedure. 
We say that a vertex $x$ is touched whenever it is removed from $H$ for inspection, or added to $H$ for the first time. 
Notice that $x$ can be added to $H$ in two cases: either the algorithm removed the edge $e$ incident to $x$, in which case both the addition of $x$ to $H$ and its first inspection are charged to $e$, or when the value $\round_{(x,y)}(\delta(y))$ increases for its parent $y$ in $T$. 
In this case, we again charge the addition of $x$ to $H$ and the first inspection of $x$ to the edge $(x,y)$.
The number of the remaining inspections of $x$ is bounded by the number of times the value $\delta(x)$ has increased in the current update procedure.

 When $x$ is inspected, we perform two basic heap operations: delete $x$ from $H$, and find the smallest element of $H_x$. 
If $x$ is attached to the tree $T$, there is nothing more to do (we note that we do not explicitly set the status of the untouched children of $x$ and their descendants as ``settled''; such vertices will never be added to the heap $H$, and so they will remain attached to the tree $T$ through their current parents).
If $\delta(x)$ increases, then we inspect the corresponding bucket $B_{(\delta(x)/\eps-1)}(x)$ of $x$. For every neighbor $w$ of $x$ that lies in that bucket, we need to update the key of $x$ in $H_w$ and to add $w$ to $H$. This involves a constant number of heap operations for each such neighbor $w$ of $x$, and is charged to the inspection of the edge $e=(w,x)$ due to $x$, since $\round_e(x)$ has increased. Therefore, every inspection of a vertex $x$ takes time $O(\log n)$, plus additional time that is charged to the inspection of its adjacent edges.

Overall, since every increase of the value $\delta(x)$ is by at least $\eps$, and $\delta(x)\leq 2D$, the number of inspections charged to $x$ is bounded by $O(D/\eps)$, and the total time spent on these inspections is $O\left(\frac{D\log m}{\eps}\right)$. The total running time of the inspections that are charged to the edges was analyzed before, and is bounded by $O\left (\sum_{e\in E}\frac{D\log m}{\eps\ell(e)}\right )$. Recall that the initialization takes time $O\left(m+\frac{Dn\log m}{\eps}\right)$. Therefore, the total running time of the algorithm is $O\left(\frac{Dn\log n}{\eps}\right)+O\left (\sum_{e\in E}\frac{D\log n}{\eps\ell(e)}\right )$.